\documentclass[12pt]{article}
\usepackage{amsmath}
\usepackage{amsfonts}
\usepackage{makeidx}
\usepackage{amssymb}
\usepackage{times}
\usepackage{anysize}

\marginsize{3cm}{3cm}{2cm}{2cm}

\newtheorem{satz}{Theorem}[section]
\newtheorem{definition}[satz]{Definition}
\newtheorem{lemma}[satz]{Lemma}
\newtheorem{koro}[satz]{Corollary}
\newtheorem{bemerkung}[satz]{Remark}

\newtheorem{proposition}[satz]{Proposition}
\newtheorem{notation}[satz]{Notation}
\newenvironment{proof}{\par\noindent {\it Proof:} \hspace{7pt}}{\hfill\hbox{\vrule width 7pt depth 0pt height 7pt}
\par\vspace{10pt}}

\begin{document}

\title{AC--Conductivity Measure from Heat Production of Free Fermions in
Disordered Media}
\author{J.-B. Bru \and W. de Siqueira Pedra \and C. Hertling}
\date{\today }
\maketitle

\begin{abstract}
We extend \cite{OhmII} in order to study the linear response of free
fermions on the lattice within a (independently and identically distributed)
random potential to a macroscopic electric field that is time-- and
space--dependent. We obtain the notion of a macroscopic AC--conducti%
\-%
vity measure which only results from the second principle of thermodynamics.
The latter corresponds here to the positivity of the heat production for
cyclic processes on equilibrium states. Its Fourier transform is a
continuous bounded function which is naturally called (macroscopic)
conductivity. We additionally derive Green--Kubo relations involving
time--correlations of bosonic fields coming from current fluctuations in the
system. This is reminiscent of non--commutative central limit theorems.
\end{abstract}

\tableofcontents%

\section{Introduction}

Klein, Lenoble and M\"{u}ller introduced in \cite{Annale} the concept of a
\textquotedblleft conduc%
\-%
tivity measure\textquotedblright\ $\mu _{\mathrm{KLM}}$ for a system of
non--interacting fermions subjected to a random potential. They considered
the Anderson tight--binding model in presence of a time--dependent spatially
homogeneous electric field $\mathcal{E}=\mathcal{E}_{t}$ that is
adiabatically switched on. Then they showed that the \emph{in--phase} linear
response current density is, at any time $t\in \mathbb{R}$, given by%
\begin{equation}  \label{Ohm.KLM}
J_{\mathrm{lin}}^{\mathrm{in}}(t;\mathcal{E})=\int\nolimits_{\mathbb{R}}%
\mathcal{\hat{E}}_{\nu }\ \mathrm{e}^{i\nu t}\ \mu _{\mathrm{KLM}}(\mathrm{d}%
\nu )\ ,
\end{equation}%
cf. \cite[Eq. (2.14)]{Annale}. Here, $\mathcal{\hat{E}}$ is the Fourier
transform of $\mathcal{E}$ and is compactly supported. See also \cite{jfa}
for further details on linear response theory of such a model. The fermionic
nature of charge carriers -- electrons or holes in crystals -- was
implemented by choosing the Fermi--Dirac distribution as the initial%
\footnote{%
This corresponds to $t\rightarrow -\infty $ in their approach.} density
matrix of particles. A conducti%
\-%
vity measure can be defined without the localization assumption and at any
positive temperature, see \cite{JMP-autre}. Inspired by their work, we
propose here a notion of (macroscopic) conducti%
\-%
vity measure\ based on the second principle of thermodynamics which
corresponds here to the positivity of the heat production for cyclic
processes on equilibrium states. In fact, we seek to get a rigorous
microscopic description of the phenomenon of linear conductivity from basic
principles of thermodynamics (the second one) and quantum mechanics, only.

The present paper belongs to a succession of works on Ohm and Joule's laws
starting with \cite{OhmI,OhmII}. Indeed, we mathematically define and
analyze in \cite{OhmI} the heat production of the fermion system which is
considered here. It is a first preliminary step towards a mathematical
description from thermal considerations of transport properties of fermions
in disordered media. Then, in \cite{OhmII} we derive Ohm and Joule's laws at
the microscopic scale. This second technical step serves as a springboard to
the results presented here. Note that in the second paper so--called
microscopic \emph{conductivity distributions} are defined from microscopic
conductivity measures. The same construction can be done here to obtain
\emph{macroscopic} conductivity distributions, whose real and imaginary
parts satisfy Kramers--Kronig relations. Such arguments are not performed in
the present paper because they are already explained in detail in \cite[%
Section 3.5]{OhmII}. The same remark can be done for the derivation of
Joule's law in its original formulation under macroscopic electric fields,
see \cite[Section 4.5]{OhmII}. We present now the mathematical framework we
use and our results by only focusing on conductivity measures and current
fluctuations. For more details and additional information, see Sections \ref%
{Section main results}--\ref{Sect Conductivity Measure From Joule's Law}.

We consider the random two--parameter family $\{\mathrm{U}_{t,s}^{(\omega
)}\}_{t\geq s}$ of unitary operators on $\ell ^{2}(\mathbb{Z}^{d})$
generated by the time--dependent Hamiltonian%
\begin{equation*}
\Delta _{\mathrm{d}}^{(\mathbf{A}(t,\cdot ))}+\lambda V_{\omega }\in
\mathcal{B}(\ell ^{2}(\mathbb{Z}^{d}))\ ,
\end{equation*}%
where the parameter $\omega $ runs in a probability space and $\lambda
V_{\omega }$ is a random potential with strength $\lambda \in \mathbb{R}%
_{0}^{+}$ (i.e., $\lambda \geq 0$). Without electromagnetic potential, i.e.,
if $\mathbf{A}\equiv 0$, this Hamiltonian corresponds to the Anderson
tight--binding model, just as in \cite{Annale,JMP-autre}. The vector
potential $\mathbf{A}=\mathbf{A}(t,x)\in C_{0}^{\infty }(\mathbb{R}\times
\mathbb{R}^{d};\mathbb{R}^{d})$ represents a time--dependent \emph{spatially
inhomogeneous} electromagnetic field which is minimally coupled to (minus)
the discrete Laplacian $\Delta _{\mathrm{d}}$. We will use in the following
the Weyl (or temporal) gauge for the electromagnetic field. In contrast with
\cite{Annale,JMP-autre}, the electromagnetic field is supported in an
arbitrarily large but bounded region of space and is switched off for times
outside some finite interval $[t_{0},t_{1}]$.

The family $\{\mathrm{U}_{t,s}^{(\omega )}\}_{t\geq s}$ of unitary operators
on $\ell ^{2}(\mathbb{Z}^{d})$ induces a random two--parameter family $%
\{\tau _{t,s}^{(\omega )}\}_{t\geq s}$ of Bogoliubov automorphisms of a CAR
algebra $\mathcal{U}$ associated with (non--relativistic) fermions in the
cubic lattice $\mathbb{Z}^{d}$. Indeed, the canonical anti--commutation
relations (CAR) encode the Pauli exclusion principle. The $C^{\ast }$%
--algebra $\mathcal{U}$ corresponds to a system of (possibly) infinitely
many fermions which is infinitely extended. As initial state of the system
at time $t_{0}\in \mathbb{R}$, we take the unique KMS state on $\mathcal{U}$
related to the (autonomous) dynamics for $\mathbf{A\equiv 0}$ and inverse
temperature $\beta \in \mathbb{R}^{+}$ (i.e., $\beta >0$). We then analyze
this fermion system, which is subjected to a time--dependent electric field,
for \emph{all times} $t\in \mathbb{R}$. However, for the sake of simplicity,
in this introduction we present our main results only for times $t\geq t_{1}$
when the electromagnetic field is switched off.

The produced heat up to times $t\geq t_{1}$ is almost surely equal to%
\begin{equation}
\mathbf{Q}\left( t\right) =\int\nolimits_{\mathbb{R}^{d}}\mathrm{d}%
^{d}x\int\nolimits_{t_{0}}^{t}\mathrm{d}s_{1}\int\nolimits_{t_{0}}^{s_{1}}%
\mathrm{d}s_{2}\ \mathbf{\sigma }(s_{1}-s_{2})\left\langle E_{\mathbf{A}%
}(s_{2},x),E_{\mathbf{A}}(s_{1},x)\right\rangle \geq 0  \label{expression}
\end{equation}%
at leading order, basically up to terms of order $\mathcal{O}(|E_{\mathbf{A}%
}|^{3}+\left\vert \nabla _{x}E_{\mathbf{A}}\right\vert |E_{\mathbf{A}}|)$.
Indeed, $\mathbf{Q}\left( t\right) $ is even constant for $t\geq t_{1}$, as,
by definition of $t_{1}$, $E_{\mathbf{A}}(s,x)$ vanishes whenever $s\geq
t_{1}$. Here, $\left\langle \cdot ,\cdot \right\rangle $ stands for the
scalar product in $\mathbb{R}^{d}$,%
\begin{equation*}
E_{\mathbf{A}}(t,x):=-\partial _{t}\mathbf{A}(t,x)\ ,\quad t\in \mathbb{R},\
x\in \mathbb{R}^{d}\ ,
\end{equation*}%
and $\mathbf{\sigma }:\mathbb{R\rightarrow R}$ is a \emph{deterministic}
continuous bounded function that can be made explicit. Compare Equation (\ref%
{expression}) with \cite[Theorem 4.1 (Q), Theorem 5.12 (p)]{OhmII} at times $%
t\geq t_{1}$. As in \cite{JMP-autre}, no localization assumption is made.
Observe that (minus) the time--derivative of the vector potential $\mathbf{A}
$ is the electric field $E_{\mathbf{A}}$ because we use the Weyl gauge.
Thus, by interpreting
\begin{equation*}
\int\nolimits_{t_{0}}^{s_{1}}\mathbf{\sigma }(s_{1}-s_{2})E_{\mathbf{A}%
}(s_{2},x)\mathrm{d}s_{2}
\end{equation*}%
as the current density at time $s_{1}$\ and space position $x\in \mathbb{R}%
^{d}$, $\mathbf{\sigma }$ can be seen as the conductivity of the system.
Hence, (\ref{expression}) is the energy delivered by the electric field to
the system in the form predicted by Joule and Ohm's laws.

This interpretation is justified in Section \ref{Green--Kubo relations for
the AC--Conductivity} for all $s_{1}\geq t_{0}$. Indeed, by Theorem \ref%
{main 1 copy(8)}, (\ref{Ohm law full}) and (\ref{71bis}), $\mathbf{\sigma }$
is the linear response coefficient associated with the current density $J_{%
\mathrm{lin}}$ induced by a \emph{spatially homogeneous} time--dependent
electric field $\mathcal{E}\in C_{0}^{\infty }\left( \mathbb{R};\mathbb{R}%
\right) $ (along a fixed direction $\vec{w}\in \mathbb{R}^{d}$):
\begin{equation}
J_{\mathrm{lin}}(t)=\int\nolimits_{t_{0}}^{t}\mathbf{\sigma }\left(
t-s\right) \mathcal{E}_{s}\ \mathrm{d}s\ ,\qquad t\geq t_{0}\ .
\label{J intro}
\end{equation}%
This equation is nothing but Ohm's law (\ref{Ohm.KLM}) written in time space.

Moreover, in Section \ref{section Current Fluctuations} we show that $%
\mathbf{\sigma }$ is a time correlation function of Bose fields $\Phi _{%
\mathrm{fl}}(\mathbf{i}_{t})$ of current fluctuations $\mathbf{i}_{t}$ at
time $t\in \mathbb{R}$:
\begin{equation}
\mathbf{\sigma }\left( t\right) =-4\mathrm{Im}\left\{ \varrho _{\mathrm{fl}%
}\left( \Phi _{\mathrm{fl}}\left( \mathbf{i}_{0}\right) \Phi _{\mathrm{fl}%
}\left( \int\nolimits_{0}^{\left\vert t\right\vert }\mathbf{i}_{s}\mathrm{d}%
s\right) \right) \right\} \ .  \label{J introbis}
\end{equation}%
In particular, (\ref{J intro})--(\ref{J introbis}) yield \emph{Green--Kubo
relations}. Here, the self--adjoint (unbounded field) operators $\Phi _{%
\mathrm{fl}}(\mathbf{i}_{t})$ generate Weyl operators $\mathrm{e}^{i\Phi _{%
\mathrm{fl}}(\mathbf{i}_{t})}$ of a CCR algebra of normal current
fluctuations with respect to (w.r.t.) the initial state. $\varrho _{\mathrm{%
fl}}$ is an appropriate regular state of this CCR algebra and the right hand
side (r.h.s.)\ of (\ref{J introbis}) is thus well--defined. This is
reminiscent of non--commutative central limit theorems.

It follows from the total heat production (\ref{expression}) that
\begin{equation}
\int\nolimits_{t_{0}}^{\infty }\mathrm{d}s_{1}\int\nolimits_{t_{0}}^{s_{1}}%
\mathrm{d}s_{2}\ \mathbf{\sigma }(s_{1}-s_{2})\mathcal{E}_{s_{2}}\mathcal{E}%
_{s_{1}}\geq 0  \label{expressionbis}
\end{equation}%
for any arbitrary smooth compactly supported function $\mathcal{E}\in
C_{0}^{\infty }(\mathbb{R};\mathbb{R})$ satisfying the so--called
AC--condition%
\begin{equation}
\int_{\mathbb{R}}\mathcal{E}_{t}\mathrm{d}t=0\ .  \label{cond ac intro}
\end{equation}%
This condition follows from the fact that $\mathcal{E}$ is the derivative of
a smooth function with compact support. Under the form (\ref{expressionbis}%
)--(\ref{cond ac intro}), the positivity of the heat production can be used
together with the Bochner--Schwartz theorem \cite[Theorem IX.10]{ReedSimonII}
to obtain the existence of a \emph{positive} measure $\mu _{\mathrm{AC}}$ of
at most polynomial growth such that
\begin{equation}
\int\nolimits_{t_{0}}^{\infty }\mathrm{d}s_{1}\int\nolimits_{t_{0}}^{s_{1}}%
\mathrm{d}s_{2}\ \mathbf{\sigma }(s_{1}-s_{2})\mathcal{E}_{s_{2}}\mathcal{E}%
_{s_{1}}=\int_{\mathbb{R}\backslash \{0\}}|\mathcal{\hat{E}}_{\nu }|^{2}\
\mu _{\mathrm{AC}}(\mathrm{d}\nu )  \label{mu sigma intro}
\end{equation}%
for all $\mathcal{E}\in C_{0}^{\infty }(\mathbb{R};\mathbb{R})$ obeying (\ref%
{cond ac intro}), with $\mathcal{\hat{E}}$ being the Fourier transform of $%
\mathcal{E}$.

The measure $\mu _{\mathrm{AC}}$ is naturally named $\emph{in}$\emph{--}$%
\emph{phase\ AC}$\emph{--conducti%
\-%
vity measure} of the fermion system as $|\mathcal{\hat{E}}_{\nu }|^{2}\mu _{%
\mathrm{AC}}(\mathrm{d}\nu )$ is the heat production due to the component of
frequency $\nu $ of the electric field, in accordance with Joule's law in
the AC--regime. Its properties are then studied in more details in a
subsequent paper \cite{OhmIV}. For instance, we study in \cite[Theorem 4.6]%
{OhmIV} the limits of perfect insulators (strong disorder, complete
localization) and perfect conductors (absence of disorder) and obtain in
both cases vanishing AC--conducti%
\-%
vity measures. We also show in \cite[Theorem 4.7]{OhmIV} that $\mu _{\mathrm{%
AC}}(\mathbb{R}\backslash \{0\})>0$, at least for large temperatures $%
T=\beta ^{-1}$ and small randomness $\lambda >0$.

To conclude, our main assertions are Theorems \ref{thm charged transport
coefficient} (charge transport coefficients), \ref{main 1 copy(8)} (Ohm's
law), \ref{main 1 copy(10)} (Green--Kubo relations and current
fluctuations), \ref{main 1} (Joule's law) and \ref{Theorem AC conductivity
measure} (AC--conducti%
\-%
vity measure). This paper is organized as follows:

\begin{itemize}
\item In Section \ref{Section main results} we define our model and
highlight the relations between our approach and \cite{Annale,JMP-autre}.

\item We define in Section \ref{Green--Kubo relations for the
AC--Conductivity} a CCR algebra of fluctuations intimately related to Ohm's
law.

\item Section \ref{Sect Conductivity Measure From Joule's Law} states
Joule's law from which we deduce the existence of the (macroscopic)
AC--conducti%
\-%
vity measure.

\item Section \ref{Section technical proof Ohm-III} gathers technical proofs
on which Sections \ref{Green--Kubo relations for the AC--Conductivity}--\ref%
{Sect Conductivity Measure From Joule's Law} are based. The arguments
strongly use the results of \cite{OhmI,OhmII}.
\end{itemize}

\begin{bemerkung}[Linear conductivity in Physics]
\mbox{
}\newline
Ohm and Joule's laws, Green--Kubo relations (in some form) and its relations
with non--com%
\-%
mutative central limit theorems (cf. Remark \ref{remark conjecture}) have
already been discussed a long time ago in theoretical physics. See, e.g.,
\cite{brupedrahistoire} for a historical perspective. However, altogether
the classical theory of linear conductivity is more like a makeshift
theoretical construction than a smooth and complete theory. As claimed in
the famous paper \cite[p. 505]{meanfreepath}, \textquotedblleft \textit{it
must be admitted that there is no entirely rigorous quantum theory of
conductivity}.\textquotedblright\ In particular, we find it unsatisfactory
to use the Drude (or the Drude--Lorentz) model -- which does not take into
account quantum mechanics -- together with certain ad hoc hypotheses as a
proper microscopic explanation of conductivity. For instance, in \cite%
{NS1,NS2,SE,Y}, the (normally fixed) relaxation time of the Drude model has
to be taken as an effective frequency--dependent parameter to fit with
experimental data \cite{T} on usual metals like gold. Such mathematical
issues are \emph{not only aesthetic}, but also yield \emph{new} results far
beyond classical theories: Explicit (paramagnetic and diamagnetic)
conductivities (Section \ref{Sect Trans coeef ddef}), notion of quantum
current viscosity (cf. Remark \ref{remark quantum current viscosity}),
description of heat/entropy productions via different energy and current
increments as explained here and in \cite[Sections 3.4-.3.5, 4]{OhmII},
existence of (AC--) conductivity measures from the 2nd principle of
thermodynamics (cf. Section \ref{sect ac cond}) and as a spectral
(excitation) measure from current fluctuations (cf. Remark \ref{Remark
Current Duhamel Fluctuations} and \cite{OhmIV}) are all examples of new
important concepts previously not discussed in Physics.
\end{bemerkung}

\begin{notation}[Generic constants]
\label{remark constant}\mbox{
}\newline
To simplify notation, we denote by $D$ any generic positive and finite
constant. These constants do not need to be the same from one statement to
another.
\end{notation}

\section{Setup of the Problem\label{Section main results}}

Up to the probability space, the mathematical setting of our study,
including notation, is taken from \cite{OhmI,OhmII}. For the reader's
convenience and completeness, we now briefly recall it and highlight the
relations to the mathematical framework of \cite{Annale}. For further
details we refer to \cite{OhmI,OhmII}.

\subsection{Anderson Tight--Binding Model\label{Section impurities}}

The $d$--dimensional cubic lattice $\mathfrak{L}:=\mathbb{Z}^{d}$ ($d\in
\mathbb{N}$) represents the crystal and we define $\mathcal{P}_{f}(\mathfrak{%
L})\subset 2^{\mathfrak{L}}$ to be the set of all \emph{finite} subsets of $%
\mathfrak{L}$.

Disorder in the crystal is modeled by a random potential coming from a
probability space $(\Omega ,\mathfrak{A}_{\Omega },\mathfrak{a}_{\Omega })$
defined as follows: Let $\Omega :=[-1,1]^{\mathfrak{L}}$ and $\Omega _{x}$, $%
x\in \mathfrak{L}$, be an arbitrary element of the Borel $\sigma $--algebra
of the interval $[-1,1]$ w.r.t. the usual metric topology. Then, $\mathfrak{A%
}_{\Omega }$ is the $\sigma $--algebra generated by the cylinder sets $%
\prod\nolimits_{x\in \mathfrak{L}}\Omega _{x}$, where $\Omega _{x}=[-1,1]$
for all but finitely many $x\in \mathfrak{L}$. The measure $\mathfrak{a}%
_{\Omega }$ is the product measure%
\begin{equation}
\mathfrak{a}_{\Omega }%
\Big(%
\prod_{x\in \mathfrak{L}}\Omega _{x}%
\Big)%
:=\prod_{x\in \mathfrak{L}}\mathfrak{a}_{\mathbf{0}}(\Omega _{x})\ ,
\label{probability measure}
\end{equation}%
where $\mathfrak{a}_{\mathbf{0}}$ is any fixed probability measure on the
interval $[-1,1]$. In other words, the random variables are independently
and identically distributed (i.i.d.). We denote by $\mathbb{E}[\ \cdot \ ]$
the expectation value associated with $\mathfrak{a}_{\Omega }$. Note that
the i.i.d. property of the random variables is not essential for our
results. We could take any ergodic ensemble instead. In fact, the i.i.d.
property is mainly used to simplify the arguments of Section \ref{sect ac
cond}. 

For any realization $\omega \in \Omega $, $V_{\omega }\in \mathcal{B}(\ell
^{2}(\mathfrak{L}))$ is the self--adjoint multiplication operator with the
function $\omega :\mathfrak{L}\rightarrow \lbrack -1,1]$. Then, we consider
the \emph{Anderson tight--binding Hamiltonian} $(\Delta _{\mathrm{d}%
}+\lambda V_{\omega })$ acting on the Hilbert space $\ell ^{2}(\mathfrak{L})$%
, where $\Delta _{\mathrm{d}}\in \mathcal{B}(\ell ^{2}(\mathfrak{L}))$ is
(up to a minus sign) the usual $d$--dimensional discrete Laplacian defined by%
\begin{equation}
\lbrack \Delta _{\mathrm{d}}(\psi )](x):=2d\psi (x)-\sum\limits_{z\in
\mathfrak{L},\text{ }|z|=1}\psi (x+z)\ ,\text{\qquad }x\in \mathfrak{L},\
\psi \in \ell ^{2}(\mathfrak{L})\ .  \label{discrete laplacian}
\end{equation}%
Note that we could add some constant (chemical) potential to the discrete
Laplacian without changing our proofs. We will use in particular the random
unitary group $\{\mathrm{U}_{t}^{(\omega ,\lambda )}\}_{t\in \mathbb{R}}$
generated by the Hamiltonians $(\Delta _{\mathrm{d}}+\lambda V_{\omega })$
for $\omega \in \Omega $ and $\lambda \in \mathbb{R}_{0}^{+}$:%
\begin{equation}
\mathrm{U}_{t}^{(\omega ,\lambda )}:=\exp (-it(\Delta _{\mathrm{d}}+\lambda
V_{\omega }))\in \mathcal{B}(\ell ^{2}(\mathfrak{L}))\ ,\text{\qquad }t\in
\mathbb{R}\ .  \label{rescaled}
\end{equation}%
This group defines our free one--particle dynamics, like in \cite{Annale}.

\subsection{Coupling to Electromagnetic Fields\label{Section Electromagnetic
Fields}}

The electromagnetic potential is defined by a compactly supported potential%
\begin{equation*}
\mathbf{A}\in \mathbf{C}_{0}^{\infty }=\underset{l\in \mathbb{R}^{+}}{%
\mathop{\displaystyle \bigcup }}C_{0}^{\infty }(\mathbb{R}\times \left[ -l,l%
\right] ^{d};({\mathbb{R}}^{d})^{\ast })\ .
\end{equation*}%
Here, $({\mathbb{R}}^{d})^{\ast }$ is the set of one--forms\footnote{%
In a strict sense, one should take the dual space of the tangent spaces $T({%
\mathbb{R}}^{d})_{x}$, $x\in {\mathbb{R}}^{d}$.} on ${\mathbb{R}}^{d}$ that
take values in $\mathbb{R}$ and $\mathbf{A}(t,x)\equiv 0$ whenever $x\notin
\lbrack -l,l]^{d}$ and $\mathbf{A}\in C_{0}^{\infty }(\mathbb{R}\times \left[
-l,l\right] ^{d};({\mathbb{R}}^{d})^{\ast })$. Using any orthonormal basis $%
\{e_{k}\}_{k=1}^{d}$ of the Euclidian space $\mathbb{R}^{d}$, we define the
scalar product between two fields $E^{(1,2)}\in ({\mathbb{R}}^{d})^{\ast }$
as usual by%
\begin{equation}
\left\langle E^{(1)},E^{(2)}\right\rangle :=\sum_{k=1}^{d}E^{(1)}\left(
e_{k}\right) E^{(2)}\left( e_{k}\right) \ .
\label{scalar product two fields}
\end{equation}%
Since $\mathbf{A}\in \mathbf{C}_{0}^{\infty }$, $\mathbf{A}(t,x)=0$ for all $%
t\leq t_{0}$, where $t_{0}\in \mathbb{R}$ is some initial time. We use the
Weyl gauge for the electromagnetic field and as a consequence,%
\begin{equation}
E_{\mathbf{A}}(t,x):=-\partial _{t}\mathbf{A}(t,x)\ ,\quad t\in \mathbb{R},\
x\in \mathbb{R}^{d}\ ,  \label{V bar 0}
\end{equation}%
is the electric field associated with $\mathbf{A}$.

\begin{bemerkung}
To simplify notation, we identify in the sequel $(\mathbb{R}^d)^\ast$ with $%
\mathbb{R}^d$ via the canonical scalar product of $\mathbb{R}^d$.
\end{bemerkung}

Since $\mathbf{A}$ is by assumption compactly supported, the corresponding
electric field satisfies the \emph{AC--condition}%
\begin{equation}
\int\nolimits_{t_{0}}^{t}E_{\mathbf{A}}(s,x)\mathrm{d}s=0\ ,\qquad x\in
\mathbb{R}^{d}\ ,  \label{zero mean field}
\end{equation}%
for sufficiently large times $t\geq t_{1}\geq t_{0}$. From (\ref{zero mean
field})%
\begin{equation}
t_{1}:=\min \left\{ t\geq t_{0}:\quad \int\nolimits_{t_{0}}^{t^{\prime }}E_{%
\mathbf{A}}(s,x)\mathrm{d}s=0\quad \text{for all }x\in \mathbb{R}^{d}\text{
and }t^{\prime }\geq t\right\}  \label{zero mean field assumption}
\end{equation}%
is the (arbitrary) time at which the electric field is turned off.

We consider without loss of generality (w.l.o.g.) \emph{negatively} charged
fermions. Thus, using the (minimal) coupling of $\mathbf{A}\in \mathbf{C}%
_{0}^{\infty }$ to the discrete Laplacian $-\Delta _{\mathrm{d}}$, the
discrete \emph{time--dependent} magnetic Laplacian is (up to a minus sign)
the self--adjoint operator
\begin{equation*}
\Delta _{\mathrm{d}}^{(\mathbf{A})}\equiv \Delta _{\mathrm{d}}^{(\mathbf{A}%
(t,\cdot ))}\in \mathcal{B}(\ell ^{2}(\mathfrak{L}))\ ,\qquad t\in \mathbb{R}%
\ ,
\end{equation*}%
defined\footnote{%
Observe that the sign of the coupling between the electromagnetic potential $%
\mathbf{A}\in \mathbf{C}_{0}^{\infty }$ and the laplacian is wrong in \cite[%
Eq. (2.8)]{OhmI}.} by%
\begin{equation}
\langle \mathfrak{e}_{x},\Delta _{\mathrm{d}}^{(\mathbf{A})}\mathfrak{e}%
_{y}\rangle =\exp \left( i\int\nolimits_{0}^{1}\left[ \mathbf{A}(t,\alpha
y+(1-\alpha )x)\right] (y-x)\mathrm{d}\alpha \right) \langle \mathfrak{e}%
_{x},\Delta _{\mathrm{d}}\mathfrak{e}_{y}\rangle  \label{eq discrete lapla A}
\end{equation}%
for all $t\in \mathbb{R}$ and $x,y\in \mathfrak{L}$. Here, $\langle \cdot
,\cdot \rangle $ is the scalar product in $\ell ^{2}(\mathfrak{L})$ and $%
\left\{ \mathfrak{e}_{x}\right\} _{x\in \mathfrak{L}}$ is the canonical
orthonormal basis $\mathfrak{e}_{x}(y)\equiv \delta _{x,y}$ of $\ell ^{2}(%
\mathfrak{L})$. In (\ref{eq discrete lapla A}), $\alpha y+(1-\alpha )x$ and $%
y-x$ are seen as elements of ${\mathbb{R}}^{d}$.

Therefore, in presence of electromagnetic fields, the Anderson
tight--binding Hamiltonian $(\Delta _{\mathrm{d}}+\lambda V_{\omega })$ for $%
\omega \in \Omega $ and $\lambda \in \mathbb{R}_{0}^{+}$ is replaced by the
time--dependent one $(\Delta _{\mathrm{d}}^{(\mathbf{A})}+\lambda V_{\omega
})$. As explained in \cite[Section 2.3]{OhmI}, the interaction between
magnetic fields and electron spins is here neglected because such a term
becomes negligible for electromagnetic potentials slowly varying in space.
This yields a perturbed dynamics defined by the random two--parameter family
$\{\mathrm{U}_{t,s}^{(\omega ,\lambda ,\mathbf{A})}\}_{t\geq s}$ of unitary
operators on $\ell ^{2}(\mathfrak{L})$ which is the unique solution, for any
$\omega \in \Omega $, $\lambda \in \mathbb{R}_{0}^{+}$ and $\mathbf{A}\in
\mathbf{C}_{0}^{\infty }$, of the non--autonomous evolution equation
\begin{equation}
\forall s,t\in {\mathbb{R}},\ t\geq s:\quad \partial _{t}\mathrm{U}%
_{t,s}^{(\omega ,\lambda ,\mathbf{A})}=-i(\Delta _{\mathrm{d}}^{(\mathbf{A}%
(t,\cdot ))}+\lambda V_{\omega })\mathrm{U}_{t,s}^{(\omega ,\lambda ,\mathbf{%
A})}\ ,\quad \mathrm{U}_{s,s}^{(\omega ,\lambda ,\mathbf{A})}:=\mathbf{1}\ .
\label{time evolution one-particle}
\end{equation}

The physical situation considered here to investigate Ohm and Joule's laws
is as follows: We start with a macroscopic bulk containing conducting
fermions. This is idealized by taking the one--particle Hilbert space $\ell
^{2}(\mathfrak{L})$. Then, the heat production or the conductivity is
measured in a local region which is very small w.r.t. the size of the bulk,
but very large w.r.t. the lattice spacing of the crystal. We implement this
hierarchy of space scales by rescaling the vector potentials. That means,
for any $l\in \mathbb{R}^{+}$ and $\mathbf{A}\in \mathbf{C}_{0}^{\infty }$,
we consider the space--rescaled vector potential
\begin{equation}
\mathbf{A}_{l}(t,x):=\mathbf{A}(t,l^{-1}x)\ ,\quad t\in \mathbb{R},\ x\in
\mathbb{R}^{d}\ .  \label{rescaled vector potential}
\end{equation}%
Then, to ensure that a macroscopic number of lattice sites is involved, we
eventually perform the limit $l\rightarrow \infty $. Indeed, the scaling
factor $l^{-1}$ used in (\ref{rescaled vector potential}) means, at fixed $l$%
, that the space scale of the electric field (\ref{V bar 0}) is
infinitesimal w.r.t. the macroscopic bulk (which is the whole space),
whereas the lattice spacing gets infinitesimal w.r.t. the space scale of the
vector potential when $l\rightarrow \infty $. Furthermore, Ohm's law is a
linear\emph{\ }response to electric fields. Therefore, we also rescale the
strength of the electromagnetic potential $\mathbf{A}_{l}$ by a parameter $%
\eta \in \mathbb{R}$ and eventually take the limit $\eta \rightarrow 0$. All
together, this yields the random two--parameter family $\{\mathrm{U}%
_{t,s}^{(\omega ,\lambda ,\eta \mathbf{A}_{l})}\}_{t\geq s}$ to be studied
in the limit $l\rightarrow \infty $, $\eta \rightarrow 0$. For more
discussions, see \cite[Section 2.3]{OhmII}.

Recall that in \cite{Annale,JMP-autre} the authors use a time--dependent
spatially homogeneous electric field that is adiabatically switched on. This
situation is thus rather different from our study where the electromagnetic
field is supported in an arbitrarily large but bounded region of space and
is switched off for times outside the finite interval $[t_{0},t_{1}]$.

\subsection{Algebraic Approach}

Although there is no interaction between fermions, we do \emph{not} restrict
our analyses to the one--particle Hilbert space to study transport
properties. Instead, our approach is based on the algebraic formulation of
fermion systems on lattices because it makes the role played by
many--fermion correlations due to the Pauli exclusion principle, i.e., the
antisymmetry of the many--body wave function, more transparent:

\begin{itemize}
\item The positivity required for the existence of the in--phase AC--conducti%
\-%
vity measure is directly related to the passivity property of (thermal
equilibrium) states on the CAR algebra $\mathcal{U}$ defined below.

\item The conductivity is naturally defined from current--current
correlations, that is,\ four--point correlation functions, in this framework.

\item The algebraic formulation allows a clear link between transport
properties of fermion systems and the CCR algebra of current fluctuations.
The latter is related to non--commutative central limit theorems (see, e.g.,
\cite{CCR fluctuations}).

\item Moreover, this approach can be naturally used to define conductivity
measures for \emph{interacting} fermions as well. This paper can thus be
seen as a preparation for such mathematical studies.
\end{itemize}

The CAR $C^{\ast }$--algebra of the infinite system is denoted by $\mathcal{U%
}$. We define annihilation and creation operators of (spinless) fermions
with wave functions $\psi \in \ell ^{2}(\mathfrak{L})$ by
\begin{equation}
a(\psi ):=\sum\limits_{x\in \mathfrak{L}}\overline{\psi (x)}a_{x}\in
\mathcal{U}\ ,\quad a^{\ast }(\psi ):=\sum\limits_{x\in \mathfrak{L}}\psi
(x)a_{x}^{\ast }\in \mathcal{U}\ .  \label{creation operators}
\end{equation}%
Here, $a_{x},a_{x}^{\ast }$, $x\in \mathfrak{L}$, and the identity $\mathbf{1%
}$ are generators of $\mathcal{U}$ and satisfy the canonical
anti--commutation relations.

For all $\omega \in \Omega $ and $\lambda \in \mathbb{R}_{0}^{+}$, the free
dynamics on $\mathcal{U}$ is defined by the unique one--parameter strongly
continuous group $\tau ^{(\omega ,\lambda )}:=\{\tau _{t}^{(\omega ,\lambda
)}\}_{t\in {\mathbb{R}}}$ of (Bogoliubov) automorphisms of $\mathcal{U}$
satisfying the condition%
\begin{equation}
\tau _{t}^{(\omega ,\lambda )}(a(\psi ))=a((\mathrm{U}_{t}^{(\omega ,\lambda
)})^{\ast }(\psi ))\ ,\text{\qquad }t\in \mathbb{R},\ \psi \in \ell ^{2}(%
\mathfrak{L})\ .  \label{rescaledbis}
\end{equation}%
See (\ref{rescaled}) and \cite[Theorem 5.2.5]{BratteliRobinson}.

Similarly, in presence of electromagnetic potentials $\mathbf{A}\in \mathbf{C%
}_{0}^{\infty }$ the dynamics on $\mathcal{U}$ is defined by the unique
family $\{\tau _{t,s}^{(\omega ,\lambda ,\mathbf{A})}\}_{t\geq s}$ of random
(Bogoliubov) automorphisms with
\begin{equation}
\tau _{t,s}^{(\omega ,\lambda ,\mathbf{A})}(a(\psi ))=a((\mathrm{U}%
_{t,s}^{(\omega ,\lambda ,\mathbf{A})})^{\ast }(\psi ))\ ,\text{\qquad }%
t\geq s,\ \psi \in \ell ^{2}(\mathfrak{L})\ ,  \label{Cauchy problem 0}
\end{equation}%
for all $\omega \in \Omega $, $\lambda \in \mathbb{R}_{0}^{+}$. See (\ref%
{time evolution one-particle}) and \cite[Theorem 5.2.5]{BratteliRobinson}.
The family $\{\tau _{t,s}^{(\omega ,\lambda ,\mathbf{A})}\}_{t\geq s}$ is
itself the solution of a non--autonomous evolution equation, see \cite[%
Sections 5.2-5.3]{OhmI}.

States on the $C^{\ast }$--algebra $\mathcal{U}$ are, by definition,
continuous linear functionals $\rho \in \mathcal{U}^{\ast }$ which are
normalized and positive, i.e., $\rho (\mathbf{1})=1$ and $\rho (A^{\ast
}A)\geq 0$ for all $A\in \mathcal{U}$. As explained for instance in \cite[%
Section 2.5]{OhmI}, the thermodynamic equilibrium of the system at inverse
temperature $\beta \in \mathbb{R}^{+}$ (i.e., $\beta >0$) is described by
the unique $(\tau ^{(\omega ,\lambda )},\beta )$--KMS state $\varrho
^{(\beta ,\omega ,\lambda )}$. See also \cite[Example 5.3.2.]%
{BratteliRobinson} or \cite[Theorem 5.9]{AttalJoyePillet2006a}. The choice
of KMS states as thermal equilibrium states is sustained by the second
principle of thermodynamics \cite{PW}. See discussions of Section \ref{Sect
Conductivity Measure From Joule's Law}. It is well--known that such states
are stationary w.r.t. the dynamics, that is,
\begin{equation}
\varrho ^{(\beta ,\omega ,\lambda )}\circ \tau _{t}^{(\omega ,\lambda
)}=\varrho ^{(\beta ,\omega ,\lambda )}\ ,\qquad \beta \in \mathbb{R}^{+},\
\omega \in \Omega ,\ \lambda \in \mathbb{R}_{0}^{+},\ t\in \mathbb{R}\ .
\label{stationary}
\end{equation}%
Since $\mathbf{A}(t,x)=0$ for all $t\leq t_{0}$, the time evolution of the
state of the system is thus%
\begin{equation}
\rho _{t}^{(\beta ,\omega ,\lambda ,\mathbf{A})}:=\left\{
\begin{array}{lll}
\varrho ^{(\beta ,\omega ,\lambda )} & , & \qquad t\leq t_{0}\ , \\
\varrho ^{(\beta ,\omega ,\lambda )}\circ \tau _{t,t_{0}}^{(\omega ,\lambda ,%
\mathbf{A})} & , & \qquad t\geq t_{0}\ .%
\end{array}%
\right.  \label{time dependent state}
\end{equation}%
This time--evolving state is \emph{quasi--free} by construction for all
times. Such quasi--free states are uniquely characterized by bounded
positive operators $\mathbf{d}\in \mathcal{B}(\ell ^{2}(\mathfrak{L}))$
obeying $0\leq \mathbf{d}\leq \mathbf{1}$. These operators are named \emph{%
symbols} of the corresponding states. The symbol of $\varrho ^{(\beta
,\omega ,\lambda )}$ is
\begin{equation}
\mathbf{d}_{\mathrm{fermi}}^{(\beta ,\omega ,\lambda )}:=\frac{1}{1+\mathrm{e%
}^{\beta \left( \Delta _{\mathrm{d}}+\lambda V_{\omega }\right) }}\in
\mathcal{B}(\ell ^{2}(\mathfrak{L}))\ .  \label{Fermi statistic}
\end{equation}%
We infer from the definitions (\ref{time evolution one-particle}), (\ref%
{Cauchy problem 0}) and (\ref{time dependent state}) together with the
evolution law (\ref{time evolution one-particle}) that the symbol $\mathbf{d}%
_{t}^{(\beta ,\omega ,\lambda ,\mathbf{A})}$ of the quasi--free state $\rho
_{t}^{(\beta ,\omega ,\lambda ,\mathbf{A})}$ is the solution to the \emph{%
Liouville equation}%
\begin{equation}
\forall t\geq t_{0}:\quad \partial _{t}\mathbf{d}_{t}^{(\beta ,\omega
,\lambda ,\mathbf{A})}=-i[\Delta _{\mathrm{d}}^{(\mathbf{A})}+\lambda
V_{\omega },\mathbf{d}_{t}^{(\beta ,\omega ,\lambda ,\mathbf{A})}]\ ,\quad
\mathbf{d}_{t_{0}}^{(\beta ,\omega ,\lambda ,\mathbf{A})}:=\mathbf{d}_{%
\mathrm{fermi}}^{(\beta ,\omega ,\lambda )}\ ,  \label{KLM discusssion}
\end{equation}%
for every realization $\omega \in \Omega $, $\lambda \in \mathbb{R}_{0}^{+}$
and $\beta \in \mathbb{R}^{+}$. In \cite{jfa,Annale,JMP-autre} the authors
consider an evolution equation similar to (\ref{KLM discusssion}) with $%
t_{0}=-\infty $ and use the expectation value of the velocity observable
w.r.t. the trace per unit volume of $\mathbf{d}_{t}^{(\beta ,\omega ,\lambda
,\mathbf{A})}\in \mathcal{B}(\ell ^{2}(\mathfrak{L}))$ to define a current
density. See, e.g., \cite[Eqs. (2.5)--(2.6)]{Annale}. Spatially local
perturbations $\mathbf{A}\in \mathbf{C}_{0}^{\infty }$ of the
electromagnetic field do not influence the mean velocity of an infinite
system of particles. Thus, by contrast, the electromagnetic perturbation
considered in \cite{jfa,Annale,JMP-autre} is infinitely extended as it is
space--homogeneous. Indeed, w.r.t. the time--evolving density operator $%
\mathbf{d}_{t}^{(\beta ,\omega ,\lambda ,\mathbf{A})}$, the main quantities
we analyze are \emph{not trace densities}, but rather the infinite volume
limit of (finite volume) traces, see, e.g., (\ref{incr one-part}) below.
Note however that, by considering space--homogeneous electromagnetic
perturbations $\mathbf{A}_{l}$ in finite boxes $\Lambda _{l}$ and the
corresponding current densities, up to the different convention on $\mathbf{d%
}_{t}^{(\beta ,\omega ,\lambda ,\mathbf{A})}$ for the initial condition, one
would obtain in the limit $l\rightarrow \infty $ a notion of conductivity
corresponding quite well to the one introduced in \cite[Eqs. (2.5)--(2.6)]%
{Annale}, even if this correspondence is not totally explicit and the
approaches are conceptually different. See also discussions around (\ref{Jn
one-part}).

\section{CCR Algebra of Fluctuations of Ohm's Law\label{Green--Kubo
relations for the AC--Conductivity}}

The study of classical (macroscopic) Ohm's law for fermions within
disordered media leads us to consider a CCR $C^{\ast }$--algebra of current
fluctuations. Exactly like in \cite[Section 3]{OhmII}, we only consider
space--homogeneous (though time--dependent) electric fields in the box%
\begin{equation}
\Lambda _{l}:=\{(x_{1},\ldots ,x_{d})\in \mathfrak{L}\,:\,|x_{1}|,\ldots
,|x_{d}|\leq l\}\in \mathcal{P}_{f}(\mathfrak{L})  \label{eq:def lambda n}
\end{equation}%
with $l\in \mathbb{R}^{+}$. More precisely, let $\vec{w}:=(w_{1},\ldots
,w_{d})\in \mathbb{R}^{d}$ be any (normalized) vector, $\mathcal{A}\in
C_{0}^{\infty }\left( \mathbb{R};\mathbb{R}\right) $ and set $\mathcal{E}%
_{t}:=-\partial _{t}\mathcal{A}_{t}$ for all $t\in \mathbb{R}$. Then, $%
\mathbf{\bar{A}}\in \mathbf{C}_{0}^{\infty }$ is defined to be the
electromagnetic potential whose electric field equals $\mathcal{E}_{t}\vec{w}
$ at time $t\in \mathbb{R}$ for all $x\in \left[ -1,1\right] ^{d}$, and $%
(0,0,\ldots ,0)$ for $t\in \mathbb{R}$ and $x\notin \left[ -1,1\right] ^{d}$%
. See (\ref{condition AC example})--(\ref{cond ac ex 3}) for more details.
This choice yields rescaled electromagnetic potentials $\eta \mathbf{\bar{A}}%
_{l}$ as defined by (\ref{rescaled vector potential}) for $l\in \mathbb{R}%
^{+}$ and $\eta \in \mathbb{R}$.

\subsection{Macroscopic Transport Coefficients\label{Sect Trans coeef ddef}}

For any pair $\mathbf{x}:=(x^{(1)},x^{(2)})\in \mathfrak{L}^{2}$, we define
the paramagnetic\ and diamagnetic current observables $I_{\mathbf{x}}=I_{%
\mathbf{x}}^{\ast }$ and $\mathrm{I}_{\mathbf{x}}^{\mathbf{A}}=(\mathrm{I}_{%
\mathbf{x}}^{\mathbf{A}})^{\ast }$ for $\mathbf{A}\in \mathbf{C}_{0}^{\infty
}$ at time $t\in \mathbb{R}$ by%
\begin{equation}
I_{\mathbf{x}}:=-2\mathrm{Im}(a_{x^{(2)}}^{\ast
}a_{x^{(1)}})=i(a_{x^{(2)}}^{\ast }a_{x^{(1)}}-a_{x^{(1)}}^{\ast
}a_{x^{(2)}})  \label{current observable}
\end{equation}%
and%
\begin{equation}
\mathrm{I}_{\mathbf{x}}^{\mathbf{A}}:=-2\mathrm{Im}\left( \left( \mathrm{e}%
^{i\int\nolimits_{0}^{1}[\mathbf{A}(t,\alpha x^{(2)}+(1-\alpha
)x^{(1)})](x^{(2)}-x^{(1)})\mathrm{d}\alpha }-1\right) a_{x^{(2)}}^{\ast
}a_{x^{(1)}}\right) \ .  \label{current observable new}
\end{equation}%
Here, $I_{(x,y)}$ is the observable related to the flow of negatively
charged particles from the lattice site $x$ to the lattice site $y$ or the
current from $y$ to $x$ without external electromagnetic potential. $\mathrm{%
I}_{\mathbf{x}}^{\mathbf{A}}$ is the current observable corresponding to the
acceleration of charged particles induced by the electromagnetic field. See
\cite[Section 3.1]{OhmII} for more details. We also denote by
\begin{equation}
P_{\mathbf{x}}:=-a_{x^{(2)}}^{\ast }a_{x^{(1)}}-a_{x^{(1)}}^{\ast
}a_{x^{(2)}}\ ,\qquad \mathbf{x}:=(x^{(1)},x^{(2)})\in \mathfrak{L}^{2}\ ,
\label{R x}
\end{equation}%
the second--quantization of the adjacency matrix of the oriented graph
containing exactly the edges $(x^{(2)},x^{(1)})$ and $(x^{(1)},x^{(2)})$.

Now, for any $\beta \in \mathbb{R}^{+}$, $\omega \in \Omega $ and $\lambda
\in \mathbb{R}_{0}^{+}$ we define two important functions associated with
the observables $I_{\mathbf{x}}$ and $P_{\mathbf{x}}$:

\begin{itemize}
\item[(p)] The paramagnetic transport coefficient $\sigma _{\mathrm{p}%
}^{(\omega )}\equiv \sigma _{\mathrm{p}}^{(\beta ,\omega ,\lambda )}$ is
defined by
\begin{equation}
\sigma _{\mathrm{p}}^{(\omega )}\left( \mathbf{x},\mathbf{y},t\right)
:=\int\nolimits_{0}^{t}\varrho ^{(\beta ,\omega ,\lambda )}\left( i[I_{%
\mathbf{y}},\tau _{s}^{(\omega ,\lambda )}(I_{\mathbf{x}})]\right) \mathrm{d}%
s\ ,\quad \mathbf{x},\mathbf{y}\in \mathfrak{L}^{2}\ ,\ t\in \mathbb{R}\ .
\label{backwards -1bis}
\end{equation}

\item[(d)] The diamagnetic transport coefficient $\sigma _{\mathrm{d}%
}^{(\omega )}\equiv \sigma _{\mathrm{d}}^{(\beta ,\omega ,\lambda )}$ is
defined by%
\begin{equation}
\sigma _{\mathrm{d}}^{(\omega )}\left( \mathbf{x}\right) :=\varrho ^{(\beta
,\omega ,\lambda )}\left( P_{\mathbf{x}}\right) \ ,\qquad \mathbf{x}\in
\mathfrak{L}^{2}\ .  \label{backwards -1bispara}
\end{equation}
\end{itemize}

\noindent As explained in \cite[Section 3.3]{OhmII}, $\sigma _{\mathrm{p}%
}^{(\omega )}$ is related with a quantum current viscosity whereas $\sigma _{%
\mathrm{d}}^{(\omega )}$ describes the ballistic movement of charged
particles within the electric field.

For large regions $\Lambda_l \subset \mathfrak{L}$, we then define the
space--averaged paramagnetic transport coefficient
\begin{equation*}
t\mapsto \Xi _{\mathrm{p},l}^{(\omega )}\left( t\right) \equiv \Xi _{\mathrm{%
p},l}^{(\beta ,\omega ,\lambda )}\left( t\right) \in \mathcal{B}(\mathbb{R}%
^{d})
\end{equation*}%
w.r.t. the canonical orthonormal basis $\{e_{k}\}_{k=1}^{d}$ of the
Euclidian space $\mathbb{R}^{d}$ by%
\begin{equation}
\left\{ \Xi _{\mathrm{p},l}^{(\omega )}\left( t\right) \right\} _{k,q}:=%
\frac{1}{\left\vert \Lambda _{l}\right\vert }\underset{x,y\in \Lambda _{l}}{%
\sum }\sigma _{\mathrm{p}}^{(\omega )}\left( x+e_{q},x,y+e_{k},y,t\right)
\label{average conductivity}
\end{equation}%
for any $l,\beta \in \mathbb{R}^{+}$, $\omega \in \Omega $, $\lambda \in
\mathbb{R}_{0}^{+}$, $k,q\in \{1,\ldots ,d\}$ and $t\in \mathbb{R}$.
Similarly, the space--averaged diamagnetic transport coefficient
\begin{equation*}
\Xi _{\mathrm{d},l}^{(\omega )}\equiv \Xi _{\mathrm{d},l}^{(\beta ,\omega
,\lambda )}\in \mathcal{B}(\mathbb{R}^{d})
\end{equation*}%
corresponds, w.r.t. the canonical orthonormal basis $\{e_{k}\}_{k=1}^{d}$,
to the diagonal matrix
\begin{equation}
\left\{ \Xi _{\mathrm{d},l}^{(\omega )}\right\} _{k,q}:=\frac{\delta _{k,q}}{%
\left\vert \Lambda _{l}\right\vert }\underset{x\in \Lambda _{l}}{\sum }%
\sigma _{\mathrm{d}}^{(\omega )}\left( x+e_{k},x\right) \in \left[ -2,2%
\right] \ .  \label{average conductivity +1}
\end{equation}%
See \cite[Eq. (37), Theorem 3.1, Corollary 3.2]{OhmII} for details on the
mathematical properties of these random transport coefficients. They are
directly linked to Ohm's law as explained in \cite[Theorem 3.3]{OhmII} and
it is natural to consider their expectation values:

We define the deterministic paramagnetic transport coefficient
\begin{equation*}
t\mapsto \mathbf{\Xi }_{\mathrm{p}}\left( t\right) \equiv \mathbf{\Xi }_{%
\mathrm{p}}^{(\beta ,\lambda )}\left( t\right) \in \mathcal{B}(\mathbb{R}%
^{d})
\end{equation*}%
by%
\begin{equation}
\mathbf{\Xi }_{\mathrm{p}}\left( t\right):=\underset{l\rightarrow \infty }{%
\lim }\mathbb{E}\left[ \Xi _{\mathrm{p},l}^{(\omega )}\left( t\right) \right]
\label{paramagnetic transport coefficient macro}
\end{equation}%
for any $\beta \in \mathbb{R}^{+}$, $\lambda \in \mathbb{R}_{0}^{+}$, $%
k,q\in \{1,\ldots ,d\}$ and $t\in \mathbb{R}$. This transport coefficient is
well--defined, see, e.g., Equation (\ref{well-defined sigma}). Furthermore,
the convergence is uniform w.r.t. times t in compact sets. By \cite[%
Corollary 3.2 (i)-(ii) and (iv)]{OhmII}, $\mathbf{\Xi }_{\mathrm{p}}\in C(%
\mathbb{R};\mathcal{B}_{-}(\mathbb{R}^{d}))$ and $\mathbf{\Xi }_{\mathrm{p}%
}(t)=\mathbf{\Xi }_{\mathrm{p}}(|t|)$ with $\mathbf{\Xi }_{\mathrm{p}}(0)=0$%
. Here, $\mathcal{B}_{-}(\mathbb{R}^{d})$ is the set of negative linear
operators on $\mathbb{R}^{d}$. Analogously, we also introduce the
deterministic diamagnetic transport coefficient
\begin{equation*}
\mathbf{\Xi }_{\mathrm{d}}\equiv \mathbf{\Xi }_{\mathrm{d}}^{(\beta ,\lambda
)}\in \mathcal{B}(\mathbb{R}^{d})
\end{equation*}%
defined, for any $\beta \in \mathbb{R}^{+}$ and $\lambda \in \mathbb{R}%
_{0}^{+}$, by
\begin{equation}
\mathbf{\Xi }_{\mathrm{d}}:=\underset{l \rightarrow \infty }{\lim }\mathbb{E}%
\left[ \Xi _{\mathrm{d},l}^{(\omega )}\right] \ .  \label{def sigma_d}
\end{equation}
Indeed, by translation invariance of the probability measure $\mathfrak{a}%
_\Omega$ and the uniqueness of the KMS states $\rho^{(\beta,\omega,\lambda)}$%
, we even have, for all $l >0$,
\begin{equation*}
\mathbf{\Xi }_{\mathrm{d}} =\mathbb{E}\left[ \Xi _{\mathrm{d},l}^{(\omega )}%
\right] \, .
\end{equation*}

By using the Akcoglu--Krengel ergodic theorem (cf. Theorem \ref%
{Ackoglu--Krengel ergodic theorem II copy(1)}) we show that the limits $%
l\rightarrow \infty $ of $\Xi _{\mathrm{p},l}^{(\omega )}$ and $\Xi _{%
\mathrm{d},l}^{(\omega )}$ converge almost surely to $\mathbf{\Xi }_{\mathrm{%
p}}$ and $\mathbf{\Xi }_{\mathrm{d}}$:

\begin{satz}[Macroscopic charge transport coefficients]
\label{thm charged transport coefficient}\mbox{
}\newline
Let $\beta \in \mathbb{R}^{+}$ and $\lambda \in \mathbb{R}_{0}^{+}$. Then,
there is a measurable subset $\tilde{\Omega}\equiv \tilde{\Omega}^{(\beta
,\lambda )}\subset \Omega $ of full measure such that, for any $\omega \in
\tilde{\Omega}$, one has:\newline
\emph{(p)} Paramagnetic charge transport coefficient: For all $t\in \mathbb{R%
}$,%
\begin{equation}
\mathbf{\Xi }_{\mathrm{p}}\left( t\right) =\ \underset{l\rightarrow \infty }{%
\lim }\Xi _{\mathrm{p},l}^{(\omega )}\left( t\right) \in \mathcal{B}_{-}(%
\mathbb{R}^{d})\ .  \notag
\end{equation}%
The limit above is uniform for times $t$ on compact sets.\newline
\emph{(d)} Diamagnetic charge transport coefficient:
\begin{equation}
\mathbf{\Xi }_{\mathrm{d}}=\ \underset{l\rightarrow \infty }{\lim }\Xi _{%
\mathrm{d},l}^{(\omega )}\in \mathcal{B}(\mathbb{R}^{d}), \quad \{ \mathbf{%
\Xi }_{\mathrm{d}} \}_{k,k} \in [-2,2], \ k\in\{1,\ldots,d\} \ .  \notag
\end{equation}
\end{satz}

\begin{proof}
(p) Take electric fields which equal $\vec{w}:=(w_{1},\ldots ,w_{d})\in
\mathbb{R}^{d}$ at time $t\in \mathbb{R}$ for all $x\in \left[ -1,1\right]
^{d}$ and $(0,0,\ldots ,0)$ for $t\in \mathbb{R}$ and $x\notin \left[ -1,1%
\right] ^{d}$. Then, the first assertion is a direct consequence of (\ref%
{average conductivity}), (\ref{paramagnetic transport coefficient macro}), (%
\ref{Lemma non existing2}), Theorem \ref{lemma conductivty4 copy(6)} and
Lemma \ref{lemma conductivty4 copy(1)} combined with \cite[Lemma 5.2]{OhmII}.

(d) is Corollary \ref{main 1 copy(21)} (ii). Note additionally that the
intersection of two measurable sets of full measure has full measure.
\end{proof}

In \cite[Eq. (47)]{OhmII} we introduce the (linear) conductivity $\mathbf{%
\Sigma }_{l}^{(\omega )}$ of the fermion system in the box $\Lambda _{l}$
from its paramagnetic and diamagnetic charge transport coefficients. Exactly
in the same way, we define the macroscopic conductivity $\mathbf{\Sigma }$
as follows:

\begin{definition}[Macroscopic conductivity]
\label{AC--conductivity}\mbox{ }\newline
For $\beta \in \mathbb{R}^{+}$ and $\lambda \in \mathbb{R}_{0}^{+}$, the
macroscopic conductivity is the map
\begin{equation*}
t\mapsto \mathbf{\Sigma }\left( t\right) \equiv \mathbf{\Sigma }^{(\beta
,\lambda )}\left( t\right) :=\left\{
\begin{array}{lll}
0 & , & \qquad t\leq 0\ . \\
\mathbf{\Xi }_{\mathrm{d}}+\mathbf{\Xi }_{\mathrm{p}}\left( t\right) & , &
\qquad t\geq 0\ .%
\end{array}%
\right.
\end{equation*}
\end{definition}

\noindent Indeed, by Theorem \ref{thm charged transport coefficient}, the
local conductivity $\mathbf{\Sigma }_{l}^{(\omega )}$ defined by \cite[Eq.
(47)]{OhmII} converges almost surely to the macroscopic conductivity $%
\mathbf{\Sigma }$, as $l\rightarrow \infty $.

\begin{bemerkung}[Current viscosity]
\label{remark quantum current viscosity}\mbox{
}\newline
For $\beta \in \mathbb{R}^{+}$, $\lambda \in \mathbb{R}_{0}^{+}$ and $t\in
\mathbb{R}$, the quantity
\begin{equation*}
\mathbf{V}\left( t\right) :=\left( \mathbf{\Xi }_{\mathrm{d}}\right)
^{-1}\partial _{t}\mathbf{\Xi }_{\mathrm{p}}\left( t\right) \in \mathcal{B}(%
\mathbb{R}^{d})
\end{equation*}%
defines a \emph{macroscopic current viscosity}, similar to \cite[Eq. (40)]%
{OhmII}.
\end{bemerkung}

\subsection{Classical Ohm's Law\label{Sect Classical Ohm's Law}}

For any $l,\beta \in \mathbb{R}^{+}$, $\omega \in \Omega $, $\lambda \in
\mathbb{R}_{0}^{+}$, $\eta \in \mathbb{R}$, $\vec{w}\in \mathbb{R}^{d}$ and $%
\mathcal{A}\in C_{0}^{\infty }\left( \mathbb{R};\mathbb{R}\right) $, the
current density due to the space--homogeneous electric perturbation $%
\mathcal{E}$ in the box $\Lambda _{l}$ is the sum of three current densities
defined from (\ref{current observable})--(\ref{current observable new}):

\begin{itemize}
\item[(th)] The thermal current density
\begin{equation*}
\mathbb{J}_{\mathrm{th}}^{(\omega ,l)}\equiv \mathbb{J}_{\mathrm{th}%
}^{(\beta ,\omega ,\lambda ,l)}\in \mathbb{R}^{d}
\end{equation*}%
at equilibrium inside the box $\Lambda _{l}$ is defined, for any $k\in
\{1,\ldots ,d\}$, by
\begin{equation}
\left\{ \mathbb{J}_{\mathrm{th}}^{(\omega ,l)}\right\} _{k}:=\left\vert
\Lambda _{l}\right\vert ^{-1}\underset{x\in \Lambda _{l}}{\sum }\varrho
^{(\beta ,\omega ,\lambda )}(I_{(x+e_{k},x)})\ .  \label{free current}
\end{equation}

\item[(p)] The paramagnetic current density is the map
\begin{equation*}
t\mapsto \mathbb{J}_{\mathrm{p}}^{(\omega ,\eta \mathbf{\bar{A}}_{l})}\left(
t\right) \equiv \mathbb{J}_{\mathrm{p}}^{(\beta ,\omega ,\lambda ,\eta
\mathbf{\bar{A}}_{l})}\left( t\right) \in \mathbb{R}^{d}
\end{equation*}%
defined by the space average of the current increment vector inside the box $%
\Lambda _{l}$ at times $t\geq t_{0}$, that is for any $k\in \{1,\ldots ,d\}$%
,
\begin{equation}
\left\{ \mathbb{J}_{\mathrm{p}}^{(\omega ,\eta \mathbf{\bar{A}}_{l})}\left(
t\right) \right\} _{k}:=\left\vert \Lambda _{l}\right\vert ^{-1}\underset{%
x\in \Lambda _{l}}{\sum }\rho _{t}^{(\beta ,\omega ,\lambda ,\eta \mathbf{%
\bar{A}}_{l})}\left( I_{(x+e_{k},x)}\right) -\varrho ^{(\beta ,\omega
,\lambda )}\left( I_{(x+e_{k},x)}\right) \ .
\label{finite volume current density}
\end{equation}

\item[(d)] The diamagnetic (or ballistic) current density
\begin{equation*}
t\mapsto \mathbb{J}_{\mathrm{d}}^{(\omega ,\eta \mathbf{\bar{A}}_{l})}\left(
t\right) \equiv \mathbb{J}_{\mathrm{d}}^{(\beta ,\omega ,\lambda ,\eta
\mathbf{\bar{A}}_{l})}\left( t\right) \in \mathbb{R}^{d}
\end{equation*}%
is defined analogously, for any $t\geq t_{0}$ and $k\in \{1,\ldots ,d\}$, by
\begin{equation}
\left\{ \mathbb{J}_{\mathrm{d}}^{(\omega ,\eta \mathbf{\bar{A}}_{l})}\left(
t\right) \right\} _{k}:=\left\vert \Lambda _{l}\right\vert ^{-1}\underset{%
x\in \Lambda _{l}}{\sum }\rho _{t}^{(\beta ,\omega ,\lambda ,\eta \mathbf{%
\bar{A}}_{l})}(\mathrm{I}_{(x+e_{k},x)}^{\eta \mathbf{\bar{A}}_{l}})\ .
\label{finite volume current density2}
\end{equation}
\end{itemize}

\noindent Thermal currents are due to the space inhomogeneity of the fermion
system for $\lambda \in \mathbb{R}^{+}$. The paramagnetic current density is
only related to the change of internal state $\rho _{t}^{(\beta ,\omega
,\lambda ,\mathbf{A})}$ produced by the electromagnetic field. We show in
\cite[Theorem 4.1]{OhmII} that it carries the paramagnetic energy increment
defined in Section \ref{sect 2.7}. The diamagnetic current density
corresponds to a raw ballistic flow of charged particles caused by the
electric field. It yields the diamagnetic energy again defined in Section %
\ref{sect 2.7}. Paramagnetic and diamagnetic currents correspond to
different physical phenomena. See \cite[Sections 3.4-3.5, 4.4]{OhmII} for
more details.

In order to compare the objects we study in the present paper with those of
\cite{Annale} we rewrite the current densities in terms of the one--particle
Hilbert space $\ell ^{2}(\mathfrak{L})$. Indeed, by using the time--evolving
symbols $\mathbf{d}_{t}^{(\beta ,\omega ,\lambda ,\eta \mathbf{\bar{A}}%
_{l})}\in \mathcal{B}(\ell ^{2}(\mathfrak{L})) $ of the quasi--free state $%
\rho _{t}^{(\beta ,\omega ,\lambda ,\eta \mathbf{\bar{A}}_{l})}$, the (full)
current density on the direction $e_{k}$, $k\in \{1,\ldots ,d\}$, can be
seen as a trace on the one--particle Hilbert space $\ell ^{2}(\mathfrak{L})$
for every $l\in \mathbb{R}^{+}$:
\begin{eqnarray}
&&\left\{ \mathbb{J}_{\mathrm{th}}^{(\omega ,l)}+\mathbb{J}_{\mathrm{p}%
}^{(\omega ,\eta \mathbf{\bar{A}}_{l})}\left( t\right) +\mathbb{J}_{\mathrm{d%
}}^{(\omega ,\eta \mathbf{\bar{A}}_{l})}\left( t\right) \right\} _{k}  \notag
\\
&=&-\left\vert \Lambda _{l}\right\vert ^{-1}\mathrm{Tr}_{\ell ^{2}(\mathfrak{%
L})}\left[ \mathbf{d}_{t}^{(\beta ,\omega ,\lambda ,\eta \mathbf{\bar{A}}%
_{l})}\mathrm{P}_{l}i[\Delta _{\mathrm{d}}^{(\eta \mathbf{\bar{A}}%
_{l})},X_{k}]\mathrm{P}_{l}\right] +\mathcal{O}(l^{-1})\ ,
\label{Jn one-part}
\end{eqnarray}%
uniformly w.r.t. all parameters. Here, for any $l\in \mathbb{R}^{+}$, $%
\mathrm{P}_{l}\in \mathcal{B}(\ell ^{2}(\mathfrak{L}))$ is the orthogonal
projector with range $\mathrm{lin}\{\mathfrak{e}_{x}$ $:$ $x\in \Lambda
_{l}\}$, i.e., the multiplication operator with the characteristic function
of the box $\Lambda _{l}$. For any $k\in \{1,\ldots ,d\}$, $X_{k}$ is the
(unbounded) multiplication operator on $\ell ^{2}(\mathfrak{L})$ with the $%
k^{\text{th}}$ space component:
\begin{equation*}
X_{k}(\psi )(x_{1},\ldots ,x_{d}):=x_{k}\psi (x_{1},\ldots ,x_{d})
\end{equation*}%
for all $\psi \in \ell ^{2}(\mathfrak{L})$ in the domain of definition of $%
X_{k}$. The term of order $\mathcal{O}(l^{-1})$\ in (\ref{Jn one-part})
results from the existence of $\mathcal{O}(l^{d-1})$ points $x\in \Lambda
_{l}$ such that $x+e_{k}\notin \Lambda _{l}$. Therefore, by (\ref{Jn
one-part}), the full current density can be seen as a kind of density of
trace of a velocity operator on the one--particle space $\ell ^{2}(\mathfrak{%
L})$ like in \cite[Eq. (2.6)]{Annale}. However, as compared with \cite%
{Annale,JMP-autre}, the density operator $\mathbf{d}_{t}^{(\beta ,\omega
,\lambda ,\eta \mathbf{\bar{A}}_{l})}$ depends on the size of the box in our
formulation.

By \cite[Theorem 3.3]{OhmII}, the current density behaves, at small $|\eta |$
and uniformly w.r.t. the size of the box, linearly w.r.t. the parameter $%
\eta $: For any $\vec{w}\in \mathbb{R}^{d}$ and $\mathcal{A}\in
C_{0}^{\infty }\left( \mathbb{R};\mathbb{R}\right) $, there is a strictly
positive number $\eta _{0}\in \mathbb{R}^{+}$ such that, for $|\eta |\in
\lbrack 0,\eta _{0}]$,
\begin{eqnarray*}
\mathbb{J}_{\mathrm{p}}^{(\omega ,\eta \mathbf{\bar{A}}_{l})}\left( t\right)
&=&\eta J_{\mathrm{p},l}^{(\omega ,\mathcal{A})}(t)+\mathcal{O}\left( \eta
^{2}\right) \ ,\qquad J_{\mathrm{p},l}^{(\omega ,\mathcal{A})}(t)=\mathcal{O}%
\left( 1\right) \ , \\
\mathbb{J}_{\mathrm{d}}^{(\omega ,\eta \mathbf{\bar{A}}_{l})}\left( t\right)
&=&\eta J_{\mathrm{d},l}^{(\omega ,\mathcal{A})}(t)+\mathcal{O}\left( \eta
^{2}\right) \ ,\qquad J_{\mathrm{d},l}^{(\omega ,\mathcal{A})}(t)=\mathcal{O}%
\left( 1\right) \ ,
\end{eqnarray*}%
uniformly for $l,\beta \in \mathbb{R}^{+}$, $\omega \in \Omega $, $\lambda
\in \mathbb{R}_{0}^{+}$ and $t\geq t_{0}$.

The $\mathbb{R}^{d}$--valued linear coefficients
\begin{equation*}
J_{\mathrm{p},l}^{(\omega ,\mathcal{A})}\equiv J_{\mathrm{p},l}^{(\beta
,\omega ,\lambda ,\vec{w},\mathcal{A})}\qquad \text{and}\qquad J_{\mathrm{d}%
,l}^{(\omega ,\mathcal{A})}\equiv J_{\mathrm{d},l}^{(\beta ,\omega ,\lambda ,%
\vec{w},\mathcal{A})}
\end{equation*}%
of the paramagnetic and diamagnetic current densities, respectively, become
deterministic for large boxes. They are directly related to the charge
transport coefficients $\mathbf{\Xi }_{\mathrm{p}}$ and $\mathbf{\Xi }_{%
\mathrm{d}}$ via Ohm's law:

\begin{satz}[Classical Ohm's law]
\label{main 1 copy(8)}\mbox{
}\newline
Let $\beta \in \mathbb{R}^{+}$ and $\lambda \in \mathbb{R}_{0}^{+}$. Then,
there is a measurable subset $\tilde{\Omega}\equiv \tilde{\Omega}^{(\beta
,\lambda )}\subset \Omega $ of full measure such that, for any $\omega \in
\tilde{\Omega}$, $\vec{w}\in \mathbb{R}^{d}$, $\mathcal{A}\in C_{0}^{\infty
}\left( \mathbb{R};\mathbb{R}\right) $ and $t\geq t_{0}$, the following
assertions hold true:\newline
\emph{(th)} Thermal current density:
\begin{equation*}
\underset{l\rightarrow \infty }{\lim }\ \mathbb{J}_{\mathrm{th}}^{(\omega
,l)}=\left( 0,\ldots ,0\right) \ .
\end{equation*}%
\emph{(p)} Paramagnetic current density:%
\begin{equation}
\underset{l\rightarrow \infty }{\lim }J_{\mathrm{p},l}^{(\omega ,\mathcal{A}%
)}(t)=\underset{l\rightarrow \infty }{\lim }\left( \left. \partial _{\eta }%
\mathbb{J}_{\mathrm{p}}^{(\omega ,\eta \mathbf{\bar{A}}_{l})}\left( t\right)
\right\vert _{\eta =0}\right) =\int_{t_{0}}^{t}\left( \mathbf{\Xi }_{\mathrm{%
p}}\left( t-s\right) \vec{w}\right) \mathcal{E}_{s}\mathrm{d}s\ .  \notag
\end{equation}%
\emph{(d)} Diamagnetic current density:
\begin{equation}
\underset{l\rightarrow \infty }{\lim }J_{\mathrm{d},l}^{(\omega ,\mathcal{A}%
)}(t)=\underset{l\rightarrow \infty }{\lim }\left( \left. \partial _{\eta }%
\mathbb{J}_{\mathrm{d}}^{(\omega ,\eta \mathbf{\bar{A}}_{l})}\left( t\right)
\right\vert _{\eta =0}\right) =\left( \mathbf{\Xi }_{\mathrm{d}}\vec{w}%
\right) \int_{t_{0}}^{t}\mathcal{E}_{s}\mathrm{d}s\ .  \notag
\end{equation}
\end{satz}

\begin{proof}
(th) is Corollary \ref{main 1 copy(21)} (th). Assertions (p) and (d) are
deduced from Theorem \ref{thm charged transport coefficient}, \cite[Eqs
(44)-(45)]{OhmII} and Lebesgue's dominated convergence theorem. Note that
any countable intersection of measurable sets of full measure has full
measure.
\end{proof}

\noindent Exactly like \cite[Theorem 3.3]{OhmII}, Theorem \ref{main 1
copy(8)} can be extended to macroscopically space--inhomogeneous
electromagnetic fields, that is, for all space--rescaled vector potentials $%
\mathbf{A}_{l}$ (\ref{rescaled vector potential}) with $\mathbf{A}\in
\mathbf{C}_{0}^{\infty }$, by exactly the same methods as in the proof of
Theorem \ref{main 1}. We refrain from doing it at this point, for technical
simplicity. Such a result can indeed be deduced from Theorem \ref{main 1},
see Equations (\ref{current para limit})--(\ref{current dia limit}).

Therefore, for any $\beta \in \mathbb{R}^{+}$, $\lambda \in \mathbb{R}%
_{0}^{+}$ and $\mathcal{A}\in C_{0}^{\infty }\left( \mathbb{R};\mathbb{R}%
\right) $, the \emph{full} current density linear response $J_{\mathrm{lin}%
}\equiv J_{\mathrm{lin}}^{(\beta ,\lambda ,\mathcal{A})}$ of the infinite
volume fermion system equals
\begin{equation}
J_{\mathrm{lin}}(t)=\int\nolimits_{-\infty }^{t}\left( \mathbf{\Sigma }%
\left( t-s\right) \vec{w}\right) \mathcal{E}_{s}\ \mathrm{d}s\ ,\qquad t\in
\mathbb{R}\ .  \label{Ohm law full}
\end{equation}
Moreover, $J_{\mathrm{lin}}$ is the sum of paramagnetic and diamagnetic
current densities. Such a decomposition of the current is well--known in
theoretical physics, see, e.g., \cite[Eq. (A2.14)]{dia-current}. For a
discussion on the physical meaning of the paramagnetic and diamagnetic
components of the current, in particular in which concerns heat production,
we refere to \cite[Section 3.5]{OhmII}.

\subsection{Green--Kubo Relations and CCR Algebra of Current Fluctuations
\label{section Current Fluctuations}}

For every (lattice translation) $x\in \mathfrak{L}$, the condition%
\begin{equation*}
\chi _{x}(a_{y})=a_{y+x}\ ,\quad y\in \mathfrak{L}\ ,
\end{equation*}%
uniquely defines a $\ast $--automorphism $\chi _{x}$ of the CAR $C^{\ast }$%
--algebra $\mathcal{U}$. For any $l\in \mathbb{R}^{+}$ and $B\in \mathcal{U}$%
, set
\begin{equation}
\mathbb{F}^{(l)}\left( B\right) :=\frac{1}{\left\vert \Lambda
_{l}\right\vert ^{1/2}}\underset{x\in \Lambda _{l}}{\sum }\left\{ \chi
_{x}\left( B\right) -\varrho ^{(\beta ,\omega ,\lambda )}\left( \chi
_{x}\left( B\right) \right) \mathbf{1}\right\} \ .  \label{Fluctuation2}
\end{equation}%
We name it the \emph{fluctuation observable }of the element $B\in \mathcal{U}
$.

By Theorem \ref{thm charged transport coefficient} and Lebesgue's dominated
convergence theorem together with Equations (\ref{backwards -1bis}), (\ref%
{average conductivity}) and (\ref{paramagnetic transport coefficient macro}%
), one obtains \emph{Green--Kubo relations}: For any $\beta \in \mathbb{R}%
^{+}$, $\lambda \in \mathbb{R}_{0}^{+}$, $t\in \mathbb{R}$ and $k,q\in
\{1,\ldots ,d\}$,%
\begin{equation}
\left\{ \mathbf{\Xi }_{\mathrm{p}}\left( t\right) \right\}
_{k,q}=\int\nolimits_{0}^{t}\underset{l\rightarrow \infty }{\lim }\mathbb{E}%
\left[ \varrho ^{(\beta ,\omega ,\lambda )}\left( i\left[ \mathbb{F}%
^{(l)}(I_{(e_{k},0)}),\mathbb{F}^{(l)}(\tau _{s}^{(\omega ,\lambda
)}(I_{(e_{q},0)}))\right] \right) \right] \mathrm{d}s  \label{green kubo}
\end{equation}%
with the current observable $I_{(x,y)}$ defined by (\ref{current observable}%
). See also \cite[Eq. (46)]{OhmII} and Theorem \ref{scalar product} which
ensures the existence of the limit integrand. It is therefore natural to
introduce the so--called CCR $C^{\ast }$--algebra of current fluctuations,
which is reminiscent of non--commutative central limit theorems (see, e.g.,
\cite{CCR fluctuations}).

To this end, we define the linear subspace%
\begin{equation}
\mathcal{I}:=\mathrm{lin}\left\{ \mathrm{Im}(a^{\ast }\left( \psi
_{1}\right) a\left( \psi _{2}\right) ):\psi _{1},\psi _{2}\in \ell ^{1}(%
\mathfrak{L})\subset \ell ^{2}(\mathfrak{L})\right\} \subset \mathcal{U}\ ,
\label{current space fluct}
\end{equation}%
which is the linear hull ($\mathrm{lin}$) of short range bond currents. As
explained in \cite[Section 5.1.2]{OhmII}, it is an invariant subspace of the
one--parameter group $\tau ^{(\omega ,\lambda )}$ for any $\omega \in \Omega
$ and $\lambda \in \mathbb{R}_{0}^{+}$. We define from $\mathcal{I}$ a
pre--Hilbert space $\mathcal{\check{H}}_{\mathrm{fl}}$ of current
fluctuations by using the following positive sesquilinear form $\langle
\cdot ,\cdot \rangle _{\mathcal{I}}$:

\begin{satz}[Positive sesquilinear form from current fluctuations]
\label{scalar product}\mbox{
}\newline
Let $\beta \in \mathbb{R}^{+}$ and $\lambda \in \mathbb{R}_{0}^{+}$. Then,
there is a measurable set $\tilde{\Omega}\equiv \tilde{\Omega}^{(\beta
,\lambda )}\subset \Omega $ of full measure such that, for any $\omega \in
\tilde{\Omega}$, the limit%
\begin{equation*}
\langle I,I^{\prime }\rangle _{\mathcal{I}}=\underset{l\rightarrow \infty }{%
\lim }\varrho ^{(\beta ,\omega ,\lambda )}\left( \mathbb{F}^{(l)}\left(
I\right) ^{\ast }\mathbb{F}^{(l)}\left( I^{\prime }\right) \right) \ ,\qquad
I,I^{\prime }\in \mathcal{I}\ ,
\end{equation*}%
exists and does not depend on $\omega \in \tilde{\Omega}$.
\end{satz}

\begin{proof}
See Theorem \ref{toto fluctbis}.
\end{proof}

Because of the Cauchy--Schwarz inequality for $\langle \cdot ,\cdot \rangle
_{\mathcal{I}}$, the set%
\begin{equation*}
\mathcal{I}_{0}:=\left\{ I\in \mathcal{I}:\langle I,I\rangle _{\mathcal{I}%
}=0\right\}
\end{equation*}%
is a subspace of $\mathcal{I}$ and the quotient $\mathcal{\check{H}}_{%
\mathrm{fl}}:=\mathcal{I}/\mathcal{I}_{0}$ is a pre--Hilbert space w.r.t. to
the (well--defined) scalar product
\begin{equation*}
\langle \lbrack I],[I^{\prime }]\rangle _{\mathcal{\check{H}}_{\mathrm{fl}%
}}:=\ \langle I,I^{\prime }\rangle _{\mathcal{I}}\ ,\qquad \lbrack
I],[I^{\prime }]\in \mathcal{\check{H}}_{\mathrm{fl}}\ .
\end{equation*}%
Note that $\mathcal{I}_{0}\neq \mathcal{I}$, in general. Observe also that $%
\mathcal{I}_{0}$ is an invariant subspace of $\tau ^{(\omega ,\lambda )}$
because of (\ref{stationary}). In particular, for any $t\in {\mathbb{R}}$,%
\begin{equation*}
\lbrack \tau _{t}^{(\omega ,\lambda )}(I)]=[\tau _{t}^{(\omega ,\lambda
)}(I^{\prime })]\ ,\qquad I,I^{\prime }\in \lbrack I]\in \mathcal{\check{H}}%
_{\mathrm{fl}}\ .
\end{equation*}%
Therefore, by Theorem \ref{bound incr 1 Lemma copy(1)}, the dynamics defined
by $\tau ^{(\omega ,\lambda )}$ on $\mathcal{U}$ induces a unitary time
evolution on the Hilbert space $\mathcal{H}_{\mathrm{fl}}$, the completion
of $\mathcal{\check{H}}_{\mathrm{fl}}$ w.r.t. the scalar product $\langle
\cdot ,\cdot \rangle _{\mathcal{\check{H}}_{\mathrm{fl}}}$: Let $\beta \in
\mathbb{R}^{+}$ and $\lambda \in \mathbb{R}_{0}^{+}$. Then, there is a
measurable set $\tilde{\Omega}\equiv \tilde{\Omega}^{(\beta ,\lambda
)}\subset \Omega $ of full measure such that, for any $\omega \in \tilde{%
\Omega}$, there is a unique strongly continuous one--parameter unitary group
$\{\mathrm{V}_{t}^{(\omega ,\lambda )}\}_{t\in {\mathbb{R}}}$ on the Hilbert
space $\mathcal{H}_{\mathrm{fl}}$ obeying, for any $t\in {\mathbb{R}}$,
\begin{equation}
\mathrm{V}_{t}^{(\omega ,\lambda )}([I])=[\tau _{t}^{(\omega ,\lambda
)}(I)]\ ,\qquad \lbrack I]\in \mathcal{\check{H}}_{\mathrm{fl}}\ .
\label{CCR0}
\end{equation}%
In Section \ref{Section average dynamics} another construction is considered
for an one--parameter unitary group $\{\mathrm{\bar{V}}_{t}^{(\lambda
)}\}_{t\in {\mathbb{R}}}$ that has similar properties but does not depend on
$\omega $. It is not given here for the sake of technical simplicity.

\begin{bemerkung}[Current Duhamel fluctuations]
\label{Remark Current Duhamel Fluctuations}\mbox{
}\newline
In \cite[Section 3]{OhmIV}, we study current fluctuations through the
Duhamel two-point function. The latter is a positive sesquilinear form which
has appeared in different contexts like in linear response theory. See \cite[%
Section A]{OhmII} and references therein for more details. This approach
yields another Hilbert Space $\mathcal{\tilde{H}}_{\mathrm{fl}}$ of
so--called current Duhamel fluctuations from which one can also construct a
CCR algebra, as explained below. $\mathcal{\tilde{H}}_{\mathrm{fl}}$ is in
some sense more natural than $\mathcal{H}_{\mathrm{fl}}$. For instance, it
allows to see the (macroscopic) conductivity measure as a spectral measure
by rewriting (\ref{green kubo}) as \cite[Eq. (21)]{OhmIV}. The
\textquotedblleft Duhamel\textquotedblright\ view point is, however,
non--standard and we postponed this study to \cite[Section 3]{OhmIV} (which
has, chronologically, been written after the present paper).
\end{bemerkung}

We define next a non--degenerate symplectic bilinear form $\mathfrak{s}%
\equiv \mathfrak{s}^{(\beta ,\lambda )}$ on $\mathcal{\check{H}}_{\mathrm{fl}%
}$ (seen as a \emph{real} vector space) by
\begin{equation}
\mathfrak{s}\left( [I],[I^{\prime }]\right) :=\mathrm{Im}\langle \lbrack
I],[I^{\prime }]\rangle _{\mathcal{\check{H}}_{\mathrm{fl}}}\ ,\qquad
\lbrack I],[I^{\prime }]\in \mathcal{\check{H}}_{\mathrm{fl}}\ .
\label{CCR00}
\end{equation}%
Let
\begin{equation*}
\mathcal{W}\equiv \mathcal{W}^{(\beta ,\lambda )}\equiv \mathcal{W}\left(
\mathcal{\check{H}}_{\mathrm{fl}},\mathfrak{s}\right)
\end{equation*}%
be the CCR algebra over the symplectic space $\left( \mathcal{\check{H}}_{%
\mathrm{fl}},\mathfrak{s}\right) $, i.e., the $C^{\ast }$--algebra generated
by the Weyl operators $\mathbf{W}\left( [I]\right) $, $[I]\in \mathcal{%
\check{H}}_{\mathrm{fl}}$, fulfilling (the Weyl form of) the canonical
commutation relations%
\begin{equation}
\mathbf{W}\left( [I]\right) \mathbf{W}\left( [I^{\prime }]\right) =\mathrm{e}%
^{-\frac{i}{2}\mathfrak{s}\left( [I],[I^{\prime }]\right) }\mathbf{W}\left(
[I]+[I^{\prime }]\right) \ ,\qquad \lbrack I],[I^{\prime }]\in \mathcal{%
\check{H}}_{\mathrm{fl}}\ .  \label{CCR1}
\end{equation}%
See, e.g., \cite[Section 5.2.2.2.]{BratteliRobinson} for more details.
Because of Remark \ref{remark conjecture} we name the space $\mathcal{W}$
the \emph{algebra of current fluctuations} of the system at inverse
temperature $\beta \in \mathbb{R}^{+}$ and strength $\lambda \in \mathbb{R}%
_{0}^{+}$ of disorder.

Take any regular state $\rho $ on the $C^{\ast }$--algebra $\mathcal{W}$. An
example of such a state is the so--called Fock state, uniquely defined by
\begin{equation*}
\varrho _{\mathrm{Fock}}^{(\beta ,\lambda )}\left( \mathbf{W}\left(
[I]\right) \right) =G_{\mathrm{Fock}}^{(\beta ,\lambda )}\left( [I]\right) :=%
\mathrm{e}^{-\frac{1}{4}\left\Vert \left[ I\right] \right\Vert _{\mathcal{%
\check{H}}_{\mathrm{fl}}}^{2}}\ ,\qquad \lbrack I]\in \mathcal{\check{H}}_{%
\mathrm{fl}}\ ,
\end{equation*}%
see \cite[Section 5.2.3.]{BratteliRobinson}. Let $(\mathfrak{H}_{\rho },\pi
_{\rho },\Psi _{\rho })$ be the GNS representation of $\mathcal{W}$ w.r.t.
the state $\rho $. Then, for any $[I]\in \mathcal{\check{H}}_{\mathrm{fl}}$,
there exists a \emph{Bose field }$\Phi \left( \lbrack I]\right) $ -- a
self--adjoint operator affiliated with the von Neumann algebra $\pi _{\rho
}\left( \mathcal{W}\right) ^{\prime \prime }$ -- such that
\begin{equation}
\pi _{\rho }\left( \mathbf{W}\left( [I]\right) \right) =\exp \left( i\Phi
\left( \lbrack I]\right) \right) \ ,\qquad \lbrack I]\in \mathcal{\check{H}}%
_{\mathrm{fl}}\ .  \label{CCR2}
\end{equation}%
Assume additionally that
\begin{equation}
\Psi _{\rho }\in \mathrm{Dom}\left( \Phi \left( \lbrack I]\right) ^{\infty
}\right) \ ,\qquad \lbrack I]\in \mathcal{\check{H}}_{\mathrm{fl}}\ .
\label{CCR3}
\end{equation}%
This assumption is for instance satisfied by the Fock state $\varrho _{%
\mathrm{Fock}}^{(\beta ,\lambda )}$.

Now, by Equations (\ref{green kubo}) and (\ref{CCR0})--(\ref{CCR00})
together with Theorem \ref{scalar product}, observe that%
\begin{equation}
\left\{ \mathbf{\Xi }_{\mathrm{p}}\left( t\right) \right\} _{k,q}=-2%
\mathfrak{s}\left( [I_{e_{k},0}],\int\nolimits_{0}^{t}\mathrm{V}%
_{s}^{(\omega ,\lambda )}([I_{e_{q},0}])\mathrm{d}s\right)  \label{CCR4}
\end{equation}%
for any $\omega \in \tilde{\Omega}$, $t\in \mathbb{R}$ and $k,q\in
\{1,\ldots ,d\}$. Using this last assertion together with (\ref{CCR1})--(\ref%
{CCR3}) we then arrive at the equality
\begin{equation*}
\left\{ \mathbf{\Xi }_{\mathrm{p}}\left( t\right) \right\} _{k,q}=-4\mathrm{%
Im}\left\{ \left\langle \Psi _{\rho },\Phi \left( \lbrack
I_{e_{k},0}]\right) \Phi \left( \int\nolimits_{0}^{t}\mathrm{V}_{s}^{(\omega
,\lambda )}([I_{e_{q},0}])\mathrm{d}s\right) \Psi _{\rho }\right\rangle _{%
\mathfrak{H}_{\rho }}\right\}
\end{equation*}%
for $\omega \in \tilde{\Omega}$, $t\in \mathbb{R}$ and all $k,q\in
\{1,\ldots ,d\}$. In particular, by Theorem \ref{main 1 copy(8)} (p), we can
rewrite the Green--Kubo relations for the paramagnetic current density at
any time in terms of time--correlations of the Bose fields $\Phi $ defined
above:

\begin{satz}[Green--Kubo relations]
\label{main 1 copy(10)}\mbox{
}\newline
Let $\beta \in \mathbb{R}^{+}$ and $\lambda \in \mathbb{R}_{0}^{+}$. Then,
there is a measurable set $\tilde{\Omega}\equiv \tilde{\Omega}^{(\beta
,\lambda )}\subset \Omega $ of full measure such that, for any $\vec{w}%
:=(w_{1},\ldots ,w_{d})\in \mathbb{R}^{d}$, $\mathcal{A}\in C_{0}^{\infty
}\left( \mathbb{R};\mathbb{R}\right) $, $t\geq t_{0}$, $\omega \in \tilde{%
\Omega}$ and $k\in \{1,\ldots ,d\}$,%
\begin{eqnarray*}
&&\underset{l\rightarrow \infty }{\lim }\left\{ J_{\mathrm{p},l}^{(\omega ,%
\mathcal{A})}(t)\right\} _{k} \\
&=&-4\mathrm{Im}\left\{ \int\nolimits_{t_{0}}^{t}\underset{q=1}{\overset{d}{%
\sum }}\left\langle \Psi _{\rho },\Phi \left( \lbrack I_{e_{k},0}]\right)
\Phi \left( \int\nolimits_{0}^{t-s_{1}}\mathrm{V}_{s_{2}}^{(\omega ,\lambda
)}([I_{e_{q},0}])\mathrm{d}s_{2}\right) \Psi _{\rho }\right\rangle _{%
\mathfrak{H}_{\rho }}w_{q}\mathcal{E}_{s_{1}}\mathrm{d}s_{1}\right\} \ .
\end{eqnarray*}
\end{satz}

By using additionally the i.i.d. property of the random variables, one can
simplify this last equation to obtain, almost surely, that
\begin{eqnarray*}
&&\underset{l\rightarrow \infty }{\lim }J_{\mathrm{p},l}^{(\omega ,\mathcal{A%
})}(t) \\
&=&-4\vec{w}\mathrm{Im}\left\{ \int\nolimits_{t_{0}}^{t}\left\langle \Psi
_{\rho },\Phi \left( \lbrack I_{e_{1},0}]\right) \Phi \left(
\int\nolimits_{0}^{t-s_{1}}\mathrm{V}_{s_{2}}^{(\omega ,\lambda
)}([I_{e_{1},0}])\mathrm{d}s_{2}\right) \Psi _{\rho }\right\rangle _{%
\mathfrak{H}_{\rho }}\mathcal{E}_{s_{1}}\ \mathrm{d}s_{1}\right\} \ .
\end{eqnarray*}%
See Equation (\ref{paramagnetic transport coefficient macrobis}) below.

\begin{bemerkung}[Non--commutative central limit theorems]
\label{remark conjecture}\mbox{
}\newline
By analogy with previous results on non--commutative central limit theorems
(see, e.g., \cite{CCR fluctuations}), we conjecture that the generating
functional
\begin{equation*}
G_{\mathrm{fl}}^{(\beta ,\lambda )}\left( [I]\right) :=\underset{%
l\rightarrow \infty }{\lim }\varrho ^{(\beta ,\omega ,\lambda )}\left( \exp
\left( i\mathbb{F}^{(l)}\left( I\right) \right) \right) \ ,\qquad \lbrack
I]\in \mathcal{\check{H}}_{\mathrm{fl}}\ ,
\end{equation*}%
exists for all $\omega $ in a set $\tilde{\Omega}\subset \Omega $ of full
measure, is independent of $\omega $ in this set, and uniquely defines a
quasi--free state $\varrho _{\mathrm{fl}}^{(\beta ,\lambda )}$ on $\mathcal{W%
}$ with
\begin{equation*}
\varrho _{\mathrm{fl}}^{(\beta ,\lambda )}\left( \mathbf{W}\left( [I]\right)
\right) =G_{\mathrm{fl}}^{(\beta ,\lambda )}\left( [I]\right) \ ,\qquad
\lbrack I]\in \mathcal{\check{H}}_{\mathrm{fl}}\ .
\end{equation*}%
From the physical point of view, the natural choice of the regular state in
Theorem \ref{main 1 copy(10)} should be $\varrho _{\mathrm{fl}}^{(\beta
,\lambda )}$.
\end{bemerkung}

\section{AC--Conductivity Measure From Joule's Law\label{Sect Conductivity
Measure From Joule's Law}}

\noindent \textit{There is however one important property of the equilibrium
states which is related to the concept of the work. Namely for such states }$%
L^{h}(\omega )\geqq 0$\textit{\ [i.e., the energy transmitted to the system
is positive,] provided the final external conditions coincide with the
original ones: }$h_{T}=0$\textit{. This fact is strongly related to the
second principle of thermodynamics saying that systems in the equilibrium
are unable to perform mechanical work in cyclic processes. We describe this
property saying that the equilibrium states are passive. }\smallskip

\hfill \lbrack Pusz--Woronowicz, 1978]\bigskip

As discussed in the Introduction, our derivation of an AC--conducti%
\-%
vity measure is based on the \emph{second principle of thermodynamics} in
the above form (cf. \cite[p. 276]{PW}). It dovetails with the positivity of
the heat production for cyclic processes on equilibrium states. Indeed, we
analyze in \cite{OhmI} the heat production of the fermion system considered
here by using Araki's relative entropy for states of infinitely extended
systems. We verify in particular the first law of thermodynamics by
identifying the heat production with an energy increment defined below, see
\cite[Theorem 3.2]{OhmI}. As originally observed by J. P. Joule, in
conducting media, electric energy is always lost in form of heat. This is
(part of) the celebrated Joule's law of (classical) electricity theory. It
corresponds to the \emph{passivity }of equilibrium states, which is proven
for KMS states in \cite[Theorem 1.2]{PW}. It follows from Joule's law that
the in--phase paramagnetic conductivity is the kernel of a positive
quadratic form on the space of smooth, compactly supported functions
satisfying an AC--condition. Recall that the latter is related to \emph{%
cyclic} electromagnetic processes. Together with the Bochner--Schwartz
theorem \cite[Theorem IX.10]{ReedSimonII}, it in turn yields the existence
of the AC--conducti%
\-%
vity measure.

We thus need to prove Joule's law, as it is done for microscopic electric
fields in \cite[Theorem 4.1]{OhmII}. It is an important result of this paper
because we seek to get a rigorous microscopic description of the phenomenon
of linear conductivity from basic principles of thermodynamics and quantum
mechanics, only. To this end we start by introducing energy densities, in
particular the heat production density.

\subsection{Energy Densities\label{sect 2.7}}

For any $L\in \mathbb{R}^{+}$, the internal energy observable in the box $%
\Lambda _{L}$ (\ref{eq:def lambda n}) is defined by%
\begin{equation}
H_{L}^{(\omega ,\lambda )}:=\sum\limits_{x,y\in \Lambda _{L}}\langle
\mathfrak{e}_{x},(\Delta _{\mathrm{d}}+\lambda V_{\omega })\mathfrak{e}%
_{y}\rangle a_{x}^{\ast }a_{y}\in \mathcal{U}\ .  \label{definition de Hl}
\end{equation}%
When the electromagnetic field is switched on, i.e., for $t\geq t_{0}$, the
(time--dependent) total energy observable in the box $\Lambda _{L}$ is then
equal to $H_{L}^{(\omega ,\lambda )}+W_{t}^{\mathbf{A}}$, where, for any $%
\mathbf{A}\in \mathbf{C}_{0}^{\infty }$ and $t\in \mathbb{R}$,
\begin{equation}
W_{t}^{\mathbf{A}}:=\sum\limits_{x,y\in \Lambda _{L}}\langle \mathfrak{e}%
_{x},(\Delta _{\mathrm{d}}^{(\mathbf{A})}-\Delta _{\mathrm{d}})\mathfrak{e}%
_{y}\rangle a_{x}^{\ast }a_{y}\in \mathcal{U}  \label{eq def W}
\end{equation}%
is the electromagnetic\emph{\ }potential energy observable.

Like in \cite[Section 4]{OhmII}, we now define below four sorts of energy
associated with the fermion system for any $\beta \in \mathbb{R}^{+}$, $%
\omega \in \Omega $, $\lambda \in \mathbb{R}_{0}^{+}$ and $\mathbf{A}\in
\mathbf{C}_{0}^{\infty }$:

\begin{itemize}
\item[($\mathbf{Q}$)] The internal energy increment $\mathbf{S}^{(\omega ,%
\mathbf{A})}\equiv \mathbf{S}^{(\beta ,\omega ,\lambda ,\mathbf{A})}$ is a
map from $\mathbb{R}$ to $\mathbb{R}_{0}^{+}$ defined by%
\begin{equation}
\mathbf{S}^{(\omega ,\mathbf{A})}\left( t\right) :=\lim_{L\rightarrow \infty
}\left\{ \rho _{t}^{(\beta ,\omega ,\lambda ,\mathbf{A})}(H_{L}^{(\omega
,\lambda )})-\varrho ^{(\beta ,\omega ,\lambda )}(H_{L}^{(\omega ,\lambda
)})\right\} \ .  \label{entropic energy increment}
\end{equation}%
It takes positive finite values because of \cite[Theorem 3.2]{OhmI}.

\item[($\mathbf{P}$)] The electromagnetic potential energy (increment) $%
\mathbf{P}^{(\omega ,\mathbf{A})}\equiv \mathbf{P}^{(\beta ,\omega ,\lambda ,%
\mathbf{A})}$ is a map from $\mathbb{R}$ to $\mathbb{R}$ defined by%
\begin{equation}
\mathbf{P}^{(\omega ,\mathbf{A})}\left( t\right) :=\rho _{t}^{(\beta ,\omega
,\lambda ,\mathbf{A})}(W_{t}^{\mathbf{A}})=\rho _{t}^{(\beta ,\omega
,\lambda ,\mathbf{A})}(W_{t}^{\mathbf{A}})-\varrho ^{(\beta ,\omega ,\lambda
)}(W_{t_{0}}^{\mathbf{A}})\ .  \label{electro free energy}
\end{equation}

\item[(p)] The paramagnetic energy increment $\mathfrak{J}_{\mathrm{p}%
}^{(\omega ,\mathbf{A})}\equiv \mathfrak{I}_{\mathrm{p}}^{(\beta ,\omega
,\lambda ,\mathbf{A})}$ is the map from $\mathbb{R}$ to $\mathbb{R}$ defined
by
\begin{equation}
\mathfrak{I}_{\mathrm{p}}^{(\omega ,\mathbf{A})}\left( t\right)
:=\lim_{L\rightarrow \infty }\left\{ \rho _{t}^{(\beta ,\omega ,\lambda ,%
\mathbf{A})}(H_{L}^{(\omega ,\lambda )}+W_{t}^{\mathbf{A}})-\varrho ^{(\beta
,\omega ,\lambda )}(H_{L}^{(\omega ,\lambda )}+W_{t}^{\mathbf{A}})\right\} \
.  \label{lim_en_incr}
\end{equation}

\item[(d)] The diamagnetic energy (increment) $\mathfrak{I}_{\mathrm{d}%
}^{(\omega ,\mathbf{A})}\equiv \mathfrak{I}_{\mathrm{d}}^{(\beta ,\omega
,\lambda ,\mathbf{A})}$ is the map from $\mathbb{R}$ to $\mathbb{R}$ defined
by
\begin{equation}
\mathfrak{I}_{\mathrm{d}}^{(\omega ,\mathbf{A})}\left( t\right) :=\varrho
^{(\beta ,\omega ,\lambda )}(W_{t}^{\mathbf{A}})=\varrho ^{(\beta ,\omega
,\lambda )}(W_{t}^{\mathbf{A}})-\varrho ^{(\beta ,\omega ,\lambda
)}(W_{t_{0}}^{\mathbf{A}})\ .  \label{lim_en_incr dia}
\end{equation}
\end{itemize}

\noindent In other words, $\mathbf{P}^{(\omega ,\mathbf{A})}$ and $\mathfrak{%
I}_{\mathrm{d}}^{(\omega ,\mathbf{A})}$are the electromagnetic\emph{\ }%
potential energy of the fermion system in the internal state $\rho
_{t}^{(\beta ,\omega ,\lambda ,\mathbf{A})}$ and thermal equilibrium state $%
\varrho ^{(\beta ,\omega ,\lambda )}$, respectively. $\mathbf{S}^{(\omega ,%
\mathbf{A})}$ represents the increase of internal energy, while $\mathfrak{J}%
_{\mathrm{p}}^{(\omega ,\mathbf{A})}$ is the part of electromagnetic work
implying a change of the internal state of the system. Note that the limits (%
\ref{entropic energy increment}) and (\ref{lim_en_incr}) exist at all times
because the increase of total energy of the infinite system equals
\begin{equation}
\mathbf{S}^{(\omega ,\mathbf{A})}\left( t\right) +\mathbf{P}^{(\omega ,%
\mathbf{A})}\left( t\right) =\mathfrak{I}_{\mathrm{p}}^{(\omega ,\mathbf{A}%
)}\left( t\right) +\mathfrak{I}_{\mathrm{d}}^{(\omega ,\mathbf{A})}\left(
t\right) \ .  \label{total work}
\end{equation}%
The increase of total energy is shown in \cite[Theorem 5.8]{OhmI} to be the
work performed by the electric field. See also \cite[Sections 4.2-4.3]{OhmII}
for more details.

By using the time--evolving symbols $\mathbf{d}_{t}^{(\beta ,\omega ,\lambda
,\mathbf{A})}\in \mathcal{B}(\ell ^{2}(\mathfrak{L}))$, all energy
increments can be seen as limits of traces on the one--particle Hilbert
space $\ell ^{2}(\mathfrak{L})$. For instance,
\begin{equation}
\mathfrak{I}_{\mathrm{p}}^{(\omega ,\mathbf{A})}\left( t\right)
=\lim_{L\rightarrow \infty }\mathrm{Tr}_{\ell ^{2}(\mathfrak{L})}\left[ (%
\mathbf{d}_{t}^{(\beta ,\omega ,\lambda ,\mathbf{A})}-\mathbf{d}%
_{t_{0}}^{(\beta ,\omega ,\lambda ,\mathbf{A})})\mathrm{P}_{L}(\Delta _{%
\mathrm{d}}^{(\mathbf{A})}+\lambda V_{\omega })\mathrm{P}_{L}\right] \ ,
\label{incr one-part}
\end{equation}%
where, for any $L\in \mathbb{R}^{+}$, $\mathrm{P}_{L}\in \mathcal{B}(\ell
^{2}(\mathfrak{L}))$ is the orthogonal projector with range $\mathrm{lin}\{%
\mathfrak{e}_{x}$ $:$ $x\in \Lambda _{L}\}$, i.e., the multiplication
operator with the characteristic function of the box $\Lambda _{L}$.

Observe that the energies
\begin{equation*}
\mathbf{P}^{(\omega ,\eta \mathbf{A}_{l})}, \quad \mathfrak{I}_{\mathrm{p}%
}^{(\omega ,\eta \mathbf{A}_{l})}\left( t\right) , \quad \mathfrak{I}_{%
\mathrm{d}}^{(\omega ,\eta \mathbf{A}_{l})}\left( t\right) ,
\end{equation*}
are all of order $\mathcal{O}(\eta ^{2}l^{d})$, by \cite[Theorem 4.1]{OhmII}%
. This can physically be understood by the fact that, by classical
electrodynamics, the energy carried by electromagnetic fields are
proportional to their $L^{2}$--norms. These are also of order $\mathcal{O}%
(\eta ^{2}l^{d})$ in the case of a potential of the form $\eta \mathbf{A}%
_{l} $. As a consequence, for any $\beta \in \mathbb{R}^{+}$, $\omega \in
\Omega $, $\lambda \in \mathbb{R}_{0}^{+}$ and $\mathbf{A}\in \mathbf{C}%
_{0}^{\infty }$, we define four energy densities:

\begin{itemize}
\item[($\mathbf{Q}$)] The \emph{heat production} (or internal energy)
density $\mathbf{s}\equiv \mathbf{s}^{(\beta ,\omega ,\lambda ,\mathbf{A})}$
is a map from $\mathbb{R}$ to $\mathbb{R}_{0}^{+}$ defined by%
\begin{equation}
\mathbf{s}\left( t\right) :=\underset{(\eta ,l^{-1})\rightarrow (0,0)}{\lim }%
\left\{ \left( \eta ^{2}\left\vert \Lambda _{l}\right\vert \right) ^{-1}%
\mathbf{S}^{(\omega ,\eta \mathbf{A}_{l})}\left( t\right) \right\} \ .
\label{entropy density}
\end{equation}%
This map has positive finite value because of \cite[Theorem 3.2]{OhmI}.

\item[($\mathbf{P}$)] The (electromagnetic) \emph{potential} energy density $%
\mathbf{p}\equiv \mathbf{p}^{(\beta ,\omega ,\lambda ,\mathbf{A})}$ is a map
from $\mathbb{R}$ to $\mathbb{R}$ defined by%
\begin{equation}
\mathbf{p}\left( t\right) :=\underset{\eta \rightarrow 0}{\lim }\ \underset{%
l\rightarrow \infty }{\lim }\left\{ \left( \eta ^{2}\left\vert \Lambda
_{l}\right\vert \right) ^{-1}\mathbf{P}^{(\omega ,\eta \mathbf{A}%
_{l})}\left( t\right) \right\} \ .  \label{free-energy density}
\end{equation}

\item[(p)] The \emph{paramagnetic} energy density $\mathfrak{i}_{\mathrm{p}%
}\equiv \mathfrak{i}_{\mathrm{p}}^{(\beta ,\omega ,\lambda ,\mathbf{A})}$ is
a map from $\mathbb{R}$ to $\mathbb{R}$ defined by%
\begin{equation}
\mathfrak{i}_{\mathrm{p}}\left( t\right) :=\underset{(\eta
,l^{-1})\rightarrow (0,0)}{\lim }\left\{ \left( \eta ^{2}\left\vert \Lambda
_{l}\right\vert \right) ^{-1}\mathfrak{I}_{\mathrm{p}}^{(\omega ,\eta
\mathbf{A}_{l})}\left( t\right) \right\} \ .
\label{paramagnetic energy density}
\end{equation}

\item[(d)] The \emph{diamagnetic} energy density $\mathfrak{i}_{\mathrm{d}%
}\equiv \mathfrak{i}_{\mathrm{d}}^{(\beta ,\omega ,\lambda ,\mathbf{A})}$ a
map from $\mathbb{R}$ to $\mathbb{R}$ defined by
\begin{equation}
\mathfrak{i}_{\mathrm{d}}\left( t\right) :=\underset{\eta \rightarrow 0}{%
\lim }\ \underset{l\rightarrow \infty }{\lim }\left\{ \left( \eta
^{2}\left\vert \Lambda _{l}\right\vert \right) ^{-1}\mathfrak{I}_{\mathrm{d}%
}^{(\omega ,\eta \mathbf{A}_{l})}\left( t\right) \right\} \ .
\label{diamagnetic energy density}
\end{equation}
\end{itemize}

\noindent For $\omega$ in a measurable subset of full measure, all energy
densities above exist and become \emph{deterministic}:

\subsection{Classical Joule's Law\label{sect Classical Joule's Law}}

Note that the probability measure $\mathfrak{a}_{\Omega }$ defined by (\ref%
{probability measure}) is translation invariant. As a consequence, charge
transport properties on \emph{macroscopic} scales are invariant under space
translation. By following the heuristics presented in \cite[Section 4.4]%
{OhmII}, we deduce from Theorem \ref{main 1 copy(8)} that, for $\beta \in
\mathbb{R}^{+}$, $\lambda \in \mathbb{R}_{0}^{+}$ and any electromagnetic
potential $\mathbf{A}\in \mathbf{C}_{0}^{\infty }$, the electric field $E_{%
\mathbf{A}}$ yields paramagnetic and diamagnetic currents linear response
coefficients respectively equal to%
\begin{eqnarray}
J_{\mathrm{p}}(t,x) &\equiv &J_{\mathrm{p}}^{(\beta ,\lambda ,\mathbf{A}%
)}(t,x):=\int_{t_{0}}^{t}\mathbf{\Xi }_{\mathrm{p}}\left( t-s\right) E_{%
\mathbf{A}}(s,x)\mathrm{d}s\ ,  \label{current para limit} \\
J_{\mathrm{d}}(t,x) &\equiv &J_{\mathrm{d}}^{(\beta ,\lambda ,\mathbf{A}%
)}(t,x):=\mathbf{\Xi }_{\mathrm{d}}\int_{t_{0}}^{t}E_{\mathbf{A}}(s,x)%
\mathrm{d}s\ ,  \label{current dia limit}
\end{eqnarray}%
at any position $x\in \mathbb{R}^{d}$ and time $t\in \mathbb{R}$.

Therefore, we expect the (density of) work delivered at time $t\geq t_{0}$
by paramagnetic and diamagnetic currents to be equal to%
\begin{eqnarray}
&&\int\nolimits_{\mathbb{R}^{d}}\mathrm{d}^{d}x\int\nolimits_{t_{0}}^{t}%
\mathrm{d}s\left\langle E_{\mathbf{A}}(s,x) , J_{\mathrm{p}%
}(s,x)\right\rangle  \label{heuristic0new} \\
&=&\int\nolimits_{\mathbb{R}^{d}}\mathrm{d}^{d}x\int\nolimits_{t_{0}}^{t}%
\mathrm{d}s_{1}\int\nolimits_{t_{0}}^{s_{1}}\mathrm{d}s_{2}\left\langle E_{%
\mathbf{A}}(s_{1},x), \mathbf{\Xi }_{\mathrm{p}}(s_{1}-s_{2})E_{\mathbf{A}%
}(s_{2},x) \right\rangle  \notag
\end{eqnarray}%
and
\begin{eqnarray}
&&\int\nolimits_{\mathbb{R}^{d}}\mathrm{d}^{d}x\int\nolimits_{t_{0}}^{t}%
\mathrm{d}s\left\langle E_{\mathbf{A}}(s,x), J_{\mathrm{d}}(s,x)
\right\rangle  \label{heuristic1new} \\
&=&\int\nolimits_{\mathbb{R}^{d}}\mathrm{d}^{d}x\int\nolimits_{t_{0}}^{t}%
\mathrm{d}s_{1}\int\nolimits_{t_{0}}^{s_{1}}\mathrm{d}s_{2}\left\langle E_{%
\mathbf{A}}(s_{1},x) , \mathbf{\Xi }_{\mathrm{d}}E_{\mathbf{A}}(s_{2},x)
\right\rangle \ ,  \notag
\end{eqnarray}
respectively.

In contrast to \cite[Section 4.4]{OhmII}, there is no current density at
equilibrium, by Theorem \ref{main 1 copy(8)} (th). Hence, (\ref%
{heuristic0new}) and (\ref{heuristic1new}) should be the \emph{work} density
performed by the electromagnetic field. Moreover, since the diamagnetic
energy (\ref{heuristic1new}) vanishes for $t\geq t_{1}$ when there is not
anymore any electric field, (\ref{heuristic0new}) should be the heat density
$\mathbf{s}\left( t\right) $ (\ref{entropy density}), by the second
principle of thermodynamics. We prove this heuristics in Section \ref%
{Section technical proof Ohm-III} and obtain the following theorem:

\begin{satz}[Classical Joule's law]
\label{main 1}\mbox{
}\newline
Let $\beta \in \mathbb{R}^{+}$ and $\lambda \in \mathbb{R}_{0}^{+}$. Then,
there is a measurable subset $\tilde{\Omega}\equiv \tilde{\Omega}^{(\beta
,\lambda )}\subset \Omega $ of full measure such that, for any $\omega \in
\tilde{\Omega}$, $\mathbf{A}\in \mathbf{C}_{0}^{\infty }$ and $t\geq t_{0}$:
\newline
\emph{(p)} Paramagnetic energy density:%
\begin{equation*}
\mathfrak{i}_{\mathrm{p}}\left( t\right) \equiv \mathfrak{i}_{\mathrm{p}%
}^{(\beta ,\omega ,\lambda ,\mathbf{A})}=\int\nolimits_{\mathbb{R}^{d}}%
\mathrm{d}^{d}x\int\nolimits_{t_{0}}^{t}\mathrm{d}s\left\langle E_{\mathbf{A}%
}(s,x),J_{\mathrm{p}}(s,x)\right\rangle \ .
\end{equation*}%
\emph{(d)} Diamagnetic energy density:%
\begin{equation*}
\mathfrak{i}_{\mathrm{d}}\left( t\right) \equiv \mathfrak{i}_{\mathrm{d}%
}^{(\beta ,\omega ,\lambda ,\mathbf{A})}=\int\nolimits_{\mathbb{R}^{d}}%
\mathrm{d}^{d}x\int\nolimits_{t_{0}}^{t}\mathrm{d}s\left\langle E_{\mathbf{A}%
}(s,x),J_{\mathrm{d}}(s,x)\right\rangle \ .
\end{equation*}%
\emph{(\textbf{Q})} Heat density:%
\begin{equation*}
\mathbf{s}\left( t\right) =\mathfrak{i}_{\mathrm{p}}\left( t\right)
-\int\nolimits_{\mathbb{R}^{d}}\mathrm{d}^{d}x\int\nolimits_{t_{0}}^{t}%
\mathrm{d}s\left\langle E_{\mathbf{A}}(s,x),J_{\mathrm{p}}(t,x)\right\rangle
\ .
\end{equation*}%
\emph{(\textbf{P})} Electromagnetic potential energy density:%
\begin{equation*}
\mathbf{p}\left( t\right) =\mathfrak{i}_{\mathrm{d}}\left( t\right)
+\int\nolimits_{\mathbb{R}^{d}}\mathrm{d}^{d}x\int\nolimits_{t_{0}}^{t}%
\mathrm{d}s\left\langle E_{\mathbf{A}}(s,x),J_{\mathrm{p}}(t,x)\right\rangle
\ .
\end{equation*}
\end{satz}

\begin{proof}
Assertion (p) corresponds to Theorem \ref{main 1 copy(2)} and Lemma \ref%
{lemma conductivty4 copy(1)} together with Fubini's theorem, while (d) is
Theorem \ref{main 1 copy(14)}. Then, (\textbf{Q}) and (\textbf{P}) follow
from (\ref{total work}) combined with Theorem \ref{Thm pot en density}.
\end{proof}

By using Definition \ref{AC--conductivity} and Theorem \ref{main 1}, for any
$\beta \in \mathbb{R}^{+}$, $\lambda \in \mathbb{R}_{0}^{+}$ and $\mathbf{A}%
\in \mathbf{C}_{0}^{\infty }$, the coefficient
\begin{equation*}
\mathfrak{E}_{\mathrm{lin}}(t)\equiv \mathfrak{E}_{\mathrm{lin}}^{(\beta
,\lambda ,\mathbf{A})}(t):= \mathfrak{i}_{\mathrm{p}}\left( t\right) +
\mathfrak{i}_{\mathrm{s}}\left( t\right) = \mathbf{s}(t) + \mathbf{p}(t)
\end{equation*}
corresponding to the \emph{total} energy density delivered to the system by
small electromagnetic fields at time $t\in \mathbb{R}$ is equal to%
\begin{equation}
\mathfrak{E}_{\mathrm{lin}}\left( t\right) =\int\nolimits_{\mathbb{R}^{d}}%
\mathrm{d}^{d}x\int\nolimits_{t_{0}}^{t}\mathrm{d}s_{1}\int%
\nolimits_{t_{0}}^{s_{1}}\mathrm{d}s_{2}\left\langle E_{\mathbf{A}}(s_{1},x),%
\mathbf{\Sigma }(s_{1}-s_{2})E_{\mathbf{A}}(s_{2},x)\right\rangle \ .
\label{Ohms law0}
\end{equation}%
For any $t\geq t_{1}$, i.e., for cyclic electromagnetic processes, we deduce
from \cite[Theorem 3.2]{OhmI} that
\begin{equation}
\mathfrak{i}_{\mathrm{d}}\left( t\right) =\mathbf{p}\left( t\right) =0\qquad
\text{and}\qquad \mathfrak{E}_{\mathrm{lin}}\left( t\right) =\mathbf{s}%
\left( t\right) =\mathfrak{i}_{\mathrm{p}}\left( t\right) \geq 0\ .
\label{Ohms law}
\end{equation}
In other words, when the electromagnetic field is switched off, all the
electromagnetic work has been converted into heat, as expected from the
second principle of thermodynamics. This phenomenology is related to Joule's
law in the AC--regime.

Indeed, the J. P. Joule originally observed that the heat (per second)
produced within a circuit is proportional to the electric resistance and the
square of the current.\ There are two clear similarities to the results
presented here:

\begin{itemize}
\item Like Joule's law, Theorem \ref{main 1} (\textbf{Q}) describes the rate
at which resistance in the fermion system converts electric energy into heat
energy. This thermal effect is directly related with \emph{current
fluctuations} via Green--Kubo relations, as explained in Section \ref%
{section Current Fluctuations}.

\item Quantitatively, Theorem \ref{main 1} is the version of Joule's law, in
the DC-- and AC--regimes, with currents and resistance replaced by electric
fields and (in phase) conductivity. 
\end{itemize}

\noindent In fact, the derivation of Joule's law in its original
formulation, that is, with currents and resistance rather than electric
fields and conductivity, can be performed by using the arguments of \cite[%
Section 4.5]{OhmII}. We omit the details.

In presence of electromagnetic fields, i.e., at times $t\in \left[
t_{0},t_{1}\right] $ for which the AC--condition (\ref{zero mean field})
does not hold, the situation is exactly the one described in \cite{OhmII}:
The raw ballistic movement of charged particles, that is responsible for the
diamagnetic currents, creates a kind of \textquotedblleft propagating wave
front\textquotedblright\ that destabilizes the whole system by changing its
internal state. These induce, at their turn, paramagnetic currents, by an
effect of (quantum) viscosity. The latter (a) modify the potential energy
and (b) produce a quantity of entropy (heat) that survives even after
turning off the electromagnetic potential. For more detailed discussions on
these effects see \cite[Section 4.4]{OhmII}.


\begin{bemerkung}[Role of field and space scales]
\mbox{
}\newline
Note that the limits $\eta \rightarrow 0$ and $l\rightarrow \infty $ in (\ref%
{free-energy density}) and (\ref{diamagnetic energy density}) do not
generally commute for $\lambda \in \mathbb{R}^{+}$. This feature comes from
the existence of thermal currents for space--inhomogeneous media which are
order $\mathcal{O}(\eta)$ (an not $\mathcal{O}(\eta^2)$), at fixed $l\in
\mathbb{R}^{+}$. See \cite[Theorem 4.1]{OhmII}, (\ref{G2-1})--(\ref{G1}) and
Theorem \ref{main 1 copy(18)}. In this case, the heat production and the
paramagnetic energy increment, which are always of order $\mathcal{O}(\eta
^{2}l^{d})$, are negligible at times $t\in (t_{0},t_{1})$ as compared to the
potential and diamagnetic energies, as $\eta \to 0$. However, for $t\geq
t_{1}$, i.e., for cyclic electromagnetic processes, the results of Theorem %
\ref{main 1} hold true even if one interchanges the limits $\eta \rightarrow
0$ and $l\rightarrow \infty $ in (\ref{free-energy density}) and (\ref%
{diamagnetic energy density}). This is of course coherent with the second
principle of thermodynamics and Joule's law in the AC--regime.
\end{bemerkung}

\subsection{AC--Conductivity Measure\label{sect ac cond}}

Recall that the random variables are independently and identically
distributed (i.i.d.), see Section \ref{Section impurities}. As to be
expected, this yields \emph{scalar} paramagnetic charge transport
coefficients (Lemma \ref{lemma conductivty4 copy(3)}):%
\begin{equation}
\mathbf{\Xi }_{\mathrm{p}}\left( t\right) =\mathbf{\sigma }_{\mathrm{p}%
}\left( t\right) \ \mathrm{Id}_{\mathbb{R}^{d}}\ ,\qquad t\in \mathbb{R}\ ,
\label{paramagnetic transport coefficient macrobis}
\end{equation}%
where, for any $\beta \in \mathbb{R}^{+}$ and $\lambda \in \mathbb{R}%
_{0}^{+} $, $\mathbf{\sigma }_{\mathrm{p}}\equiv \mathbf{\sigma }_{\mathrm{p}%
}^{(\beta ,\lambda )}$ is a well--defined real function. By \cite[Corollary
3.2 (i)-(ii) and (iv)]{OhmII} and (\ref{paramagnetic transport coefficient
macrobis}), $\mathbf{\sigma }_{\mathrm{p}}\in C(\mathbb{R};\mathbb{R}%
_{0}^{-})$ and $\mathbf{\sigma }_{\mathrm{p}}(t)=\mathbf{\sigma }_{\mathrm{p}%
}(|t|)$ with $\mathbf{\sigma }_{\mathrm{p}}(0)=0$. A detailed study of its
properties will be performed in the subsequent paper. In the same way, by (%
\ref{average conductivity +1}),
\begin{equation}
\mathbf{\Xi }_{\mathrm{d}}=\mathbf{\sigma }_{\mathrm{d}}\ \mathrm{Id}_{%
\mathbb{R}^{d}}\ ,  \label{def sigma_dbis}
\end{equation}%
where, for any $\beta \in \mathbb{R}^{+}$ and $\lambda \in \mathbb{R}%
_{0}^{+} $, $\mathbf{\sigma }_{\mathrm{d}}\equiv \mathbf{\sigma }_{\mathrm{d}%
}^{(\beta ,\lambda )}$ is the constant defined by%
\begin{equation}
\mathbf{\sigma }_{\mathrm{d}}:=\mathbb{E}\left[ \sigma _{\mathrm{d}%
}^{(\omega )}\left( e_{1},0\right) \right] \in \left[ -2,2\right] \ .
\label{def sigma_dbisbis}
\end{equation}%
In fact, one can use the scalar product $\langle \cdot ,\cdot \rangle $ in $%
\ell ^{2}(\mathfrak{L})$, the canonical orthonormal basis $\left\{ \mathfrak{%
e}_{x}\right\} _{x\in \mathfrak{L}}$ of $\ell ^{2}(\mathfrak{L})$ defined by
$\mathfrak{e}_{x}(y)\equiv \delta _{x,y}$, and the symbol $\mathbf{d}_{%
\mathrm{fermi}}^{(\beta ,\omega ,\lambda )}$ defined by (\ref{Fermi
statistic}), to get that%
\begin{equation}
\mathbf{\sigma }_{\mathrm{d}}=2\mathrm{Re}\left\{ \mathbb{E}\left[ \langle
\mathfrak{e}_{e_{1}},\mathbf{d}_{\mathrm{fermi}}^{(\beta ,\omega ,\lambda )}%
\mathfrak{e}_{0}\rangle \right] \right\}  \label{explicit form conductivity}
\end{equation}%
for any $\beta \in \mathbb{R}^{+}$ and $\lambda \in \mathbb{R}_{0}^{+}$. By (%
\ref{paramagnetic transport coefficient macrobis})--(\ref{def sigma_dbis}),
the (full) conductivity $\mathbf{\Sigma }$, which is by Definition \ref%
{AC--conductivity} the sum of the paramagnetic and diamagnetic
conductivities, equals
\begin{equation}
\mathbf{\Sigma }\left( t\right) \mathbf{=}\left\{
\begin{array}{lll}
0 & , & \qquad t\leq 0\ , \\
\mathbf{\sigma }\left( t\right) \ \mathrm{Id}_{\mathbb{R}^{d}} & , & \qquad
t\geq 0\ ,%
\end{array}%
\right.  \label{71bis}
\end{equation}%
for any $\beta \in \mathbb{R}^{+}$, $\lambda \in \mathbb{R}_{0}^{+}$ and $%
t\in \mathbb{R}$. Here, the so--called in--phase conductivity $\mathbf{%
\sigma }\equiv \mathbf{\sigma }^{(\beta ,\lambda )}\in C\left( \mathbb{R};%
\mathbb{R}\right) $ at $\beta \in \mathbb{R}^{+}$ and $\lambda \in \mathbb{R}%
_{0}^{+}$ is the continuous function defined, for any $t\in \mathbb{R}$, by%
\begin{equation*}
\mathbf{\sigma }\left( t\right) :=\mathbf{\sigma }_{\mathrm{d}}+\mathbf{%
\sigma }_{\mathrm{p}}\left( t\right) \ .
\end{equation*}%
This is a special situation which results from the i.i.d. property of the
random variables. For external potentials having non--trivial space
correlations the conductivity $\mathbf{\Sigma }(t)\in \mathcal{B}(\mathbb{R}%
^{d})$ is, in general, not of the form $\mathbf{\sigma }(t)\mathrm{Id}_{%
\mathbb{R}^{d}}$, $\mathbf{\sigma }(t)\in \mathbb{R}$.

Only the in--phase paramagnetic conductivity $\mathbf{\sigma }_{\mathrm{p}}$
is responsible for heat production. In fact, Theorem \ref{main 1} (p)
together with the AC--condition (\ref{zero mean field}) uniquely determines
the quantity $\mathbf{\sigma }_{\mathrm{p}}(t)$ for all $t\in {\mathbb{R}}$
because $\mathbf{\sigma }_{\mathrm{p}}(0)=0$ and $\mathbf{\sigma }_{\mathrm{p%
}}(t)=\mathbf{\sigma }_{\mathrm{p}}(|t|)$. To see this more explicitly we
consider the following choice of electromagnetic potential $\mathbf{A}\in
\mathbf{C}_{0}^{\infty }$: Take any smooth, compactly supported functions $%
\mathcal{E}\in C_{0}^{\infty }(\mathbb{R})$ and $\psi \in C_{0}^{\infty }(%
\mathbb{R}^{d};\mathbb{R})$ such that
\begin{equation}
\int_{\mathbb{R}}\mathcal{E}_{t}\ \mathrm{d}t=0\qquad \text{and}\qquad
\int\nolimits_{\mathbb{R}^{d}}\psi ^{2}\left( x\right) \mathrm{d}^{d}x=1\ .
\label{condition AC example}
\end{equation}%
Next, pick any normalized vector $\vec{w}:=(w_{1},\ldots ,w_{d})\in \mathbb{R%
}^{d}$ ($\left\vert \vec{w}\right\vert =1$) and define
\begin{equation}
\mathbf{A}^{(\mathcal{E},\psi )}(t,x):=\mathbf{A}_{1}^{(\mathcal{E},\psi
)}(t,x)e_{1}+\cdots +\mathbf{A}_{d}^{(\mathcal{E},\psi )}(t,x)e_{d}
\label{cond ac ex 2}
\end{equation}%
for all $t\in {\mathbb{R}}$ and $x\in {\mathbb{R}}^{d}$, where, for any $%
k\in \{1,\ldots ,d\}$,
\begin{equation}
\mathbf{A}_{k}^{(\mathcal{E},\psi )}(t,x):=w_{k}\psi \left( x\right)
\int\nolimits_{\mathbb{-\infty }}^{t}\mathcal{E}_{s}\ \mathrm{d}s\ .
\label{cond ac ex 3}
\end{equation}%
Here, $\{e_{k}\}_{k=1}^{d}$ is the canonical orthonormal basis of ${\mathbb{R%
}}^{d}$. By (\ref{condition AC example}), the vector potential $\mathbf{A}^{(%
\mathcal{E},\psi )}\in \mathbf{C}_{0}^{\infty }$ satisfies the AC--condition
(\ref{zero mean field}) for sufficiently large times $t$. Then, we infer
from Theorem \ref{main 1} (p), (\ref{Ohms law0}) and (\ref{Ohms law})
applied to the vector potential $\mathbf{A}^{(\mathcal{E},\psi )}$ that%
\begin{equation}
\frac{1}{2}\int\nolimits_{\mathbb{R}}\mathrm{d}s_{1}\int\nolimits_{\mathbb{R}%
}\mathrm{d}s_{2}\ \mathbf{\sigma }_{\mathrm{p}}(s_{1}-s_{2})\mathcal{E}%
_{s_{2}}\mathcal{E}_{s_{1}}\geq 0  \label{eq sympa}
\end{equation}%
for all $\mathcal{E}\in C_{0}^{\infty }(\mathbb{R})$ satisfying (\ref%
{condition AC example}). The latter is nothing but the density of heat
finally produced within the fermion system when the electromagnetic field $%
\mathbf{A}^{(\mathcal{E},\psi )}$ is turned off.

Since $\mathbf{\sigma }_{\mathrm{p}}\equiv \mathbf{\sigma }_{\mathrm{p}%
}^{(\beta ,\lambda )}$ is a continuous function obeying $\mathbf{\sigma }_{%
\mathrm{p}}(t)=\mathbf{\sigma }_{\mathrm{p}}(-t)$ and $\mathbf{\sigma }_{%
\mathrm{p}}(0)=0$, it is then straightforward to verify that the real
numbers $\mathbf{\sigma }_{\mathrm{p}}(t)$, $t\in \mathbb{R}$, are unique
and only depend on the inverse temperature $\beta \in \mathbb{R}^{+}$ and
the strength $\lambda \in \mathbb{R}_{0}^{+}$ of disorder. Moreover, by (\ref%
{eq sympa}), the in--phase paramagnetic conductivity $\mathbf{\sigma }_{%
\mathrm{p}}$ is the kernel of a \emph{positive} quadratic form on the space
of smooth, compactly supported functions $\mathcal{E}\in C_{0}^{\infty }(%
\mathbb{R})$ satisfying (\ref{condition AC example}). We show below that
this positivity property of $\mathbf{\sigma }_{\mathrm{p}}$ together with
the Bochner--Schwartz theorem leads to the existence of the AC--conducti%
\-%
vity measure for the system of lattice fermions considered here:

\begin{satz}[AC--conductivity measure]
\label{Theorem AC conductivity measure}\mbox{
}\newline
For any $\beta \in \mathbb{R}^{+}$ and $\lambda \in \mathbb{R}_{0}^{+}$,
there is a positive measure $\mu _{\mathrm{AC}}\equiv \mu _{\mathrm{AC}%
}^{(\beta ,\lambda )}$ of at most polynomial growth on $\mathbb{R}\backslash
\{0\}$ and a constant $D \in \mathbb{R}_0^+$ such that, for any $\mathbf{A}%
\in \mathbf{C}_{0}^{\infty }$ and $t\geq t_{1}$,
\begin{eqnarray*}
\mathbf{s}\left( t\right) =\mathfrak{i}_{\mathrm{p}}\left( t\right) &=&\frac{%
1}{2}\int\nolimits_{\mathbb{R}^{d}}\mathrm{d}^{d}x\int\nolimits_{\mathbb{R}%
\backslash \{0\}}\mu _{\mathrm{AC}}(\mathrm{d}\nu )\langle \hat{E}_{\mathbf{A%
}}(\nu ,x),\hat{E}_{\mathbf{A}}(\nu ,x)\rangle \\
&& + D \int\nolimits_{\mathbb{R}^{d}}\mathrm{d}^{d}x \left |\partial_\nu
\hat{E}_{\mathbf{A}}( \nu, x)|_{\nu=0} \right|^2
\end{eqnarray*}%
with $\hat{E}_{\mathbf{A}}$ being the Fourier transform of the
electromagnetic field $E_{\mathbf{A}}$ (\ref{V bar 0}).
\end{satz}

\begin{proof}
Fix $\beta \in \mathbb{R}^{+}$ and $\lambda \in \mathbb{R}_{0}^{+}$. It
follows from (\ref{eq sympa}) that%
\begin{equation}
\int\nolimits_{\mathbb{R}}\mathbf{\sigma }_{\mathrm{p}}(s)\left[ \tilde{%
\varphi}\ast \varphi \right] (s)\mathrm{d}s:=\int\nolimits_{\mathbb{R}}%
\mathrm{d}s_{1}\mathbf{\sigma }_{\mathrm{p}}(s_{1})\int\nolimits_{\mathbb{R}}%
\mathrm{d}s_{2}\ \varphi (s_{2}-s_{1})\varphi (s_{2})\geq 0
\label{distrib sympa 1}
\end{equation}%
for all $\varphi \in C_{0}^{\infty }(\mathbb{R})$ satisfying the first
condition of (\ref{condition AC example}). Here, $\tilde{\varphi}$ is the
function defined by $\tilde{\varphi}(s):=\varphi (-s)$ for any $s\in \mathbb{%
R}$. For all $\varphi \in C_{0}^{\infty }(\mathbb{R})$, observe that $%
\varphi ^{\prime }\in C_{0}^{\infty }(\mathbb{R})$ and
\begin{equation}
\int\nolimits_{\mathbb{R}}\varphi ^{\prime }\left( s\right) \mathrm{d}s=0\ .
\label{distrib sympa 2}
\end{equation}%
Using this in order to get rid of the first assumption in (\ref{condition AC
example}), i.e., the AC--condition, we define the real--valued distribution $%
\varsigma \equiv \varsigma ^{(\beta ,\lambda )}$ by%
\begin{equation*}
\varsigma (\varphi ):=-\int\nolimits_{\mathbb{R}}\mathbf{\sigma }_{\mathrm{p}%
}(s)\varphi ^{\prime \prime }(s)\mathrm{d}s\ ,\qquad \varphi \in
C_{0}^{\infty }(\mathbb{R})\ .
\end{equation*}%
Indeed, by (\ref{distrib sympa 1})--(\ref{distrib sympa 2}), note that%
\begin{equation*}
\varsigma (\tilde{\varphi}\ast \varphi )=\int\nolimits_{\mathbb{R}}\mathbf{%
\sigma }_{\mathrm{p}}(s)\left[ \widetilde{\left( \varphi ^{\prime }\right) }%
\ast \varphi ^{\prime }\right] (s)\mathrm{d}s\geq 0\ .
\end{equation*}%
Since $\mathbf{\sigma }_{\mathrm{p}}(t)=\mathbf{\sigma }_{\mathrm{p}}(-t)$,
this equality can easily be extended to complex--valued functions $\varphi
\in C_{0}^{\infty }(\mathbb{R};\mathbb{C})$ by replacing $\tilde{\varphi}$
with its conjugate $\overline{\tilde{\varphi}}$. So, the distribution $%
\varsigma $ is of positive type. In particular, applying the
Bochner--Schwartz theorem \cite[Theorem IX.10]{ReedSimonII} we deduce that $%
\varsigma $ is a tempered distribution which is the Fourier transform of a
positive measure $\mu _{\varsigma }$ of at most polynomial growth. Then
define the measure $\mu _{\mathrm{AC}}\equiv \mu _{\mathrm{AC}}^{(\beta
,\lambda )}$ as the restriction of $\nu ^{-2}\mu _{\varsigma }(\mathrm{d}\nu
)$ on $\mathbb{R}\backslash \{0\}$ and observe that
\begin{eqnarray*}
\mathbf{s}\left( t\right) &=&\int\nolimits_{\mathbb{R}^{d}}\mathrm{d}%
^{d}x\int\nolimits_{t_{0}}^{t}\mathrm{d}s_{1}\int\nolimits_{t_{0}}^{s_{1}}%
\mathrm{d}s_{2}\ \mathbf{\sigma }_{\mathrm{p}}(s_{1}-s_{2})\left\langle E_{%
\mathbf{A}}(s_{1},x),E_{\mathbf{A}}(s_{2},x)\right\rangle \\
&=&\frac{1}{2}\int\nolimits_{\mathbb{R}^{d}}\mathrm{d}^{d}x\int\nolimits_{%
\mathbb{R}\backslash \{0\}}\mu _{\mathrm{AC}}(\mathrm{d}\nu )\langle \hat{E}%
_{\mathbf{A}}(\nu ,x),\hat{E}_{\mathbf{A}}(\nu ,x)\rangle \\
&& + \mu _{\varsigma }(\{0\}) \int\nolimits_{\mathbb{R}^{d}}\mathrm{d}^{d}x
\left| \int\nolimits_{\mathbb{R}} \mathrm{d} s \, \mathbf{A}(s ,x) \right|^2
\end{eqnarray*}%
for any $\mathbf{A}\in \mathbf{C}_{0}^{\infty }$ and $t\geq t_{1}$, using
Theorem \ref{main 1}, Fubini's theorem, as well as Equations (\ref{zero mean
field assumption}) and (\ref{paramagnetic transport coefficient macrobis}).
\end{proof}

\noindent This theorem uniquely defines the measure $\mu _{\mathrm{AC}}$,
named the (in--phase)\emph{\ AC--cond%
\-%
uctivity measure}. It characterizes the heat production per unit volume due
to the component of frequency $\nu \in \mathbb{R}\backslash \{0\}$ of the
electric field, in accordance with Joule's law in the AC--regime.

In fact, because of the restriction on functions satisfying the
AC--condition (\ref{condition AC example}), (\ref{distrib sympa 1}) is
weaker than the condition defining functions of positive type. Above we
overcome this problem by introducing the distribution $\varsigma $ which is
clearly of positive type, but only as a general tempered distribution. In
the subsequent paper we will show that the function $\mathbf{\sigma }_{%
\mathrm{p}}$, the in--phase paramagnetic conductivity, is indeed of positive
type up to a constant. This allows us to use the original theorem of Bochner
\cite[Theorem IX.9]{ReedSimonII} on functions of positive type (or the
spectral theorem) to show the existence of a finite positive measure $\mu _{%
\mathrm{p}}$ such that
\begin{equation*}
\mathbf{\sigma }_{\mathrm{p}}(t)=\int_{\mathbb{R}}\left( \cos \left( t\nu
\right) -1\right) \mu _{\mathrm{p}}(\mathrm{d}\nu )\ ,\qquad t\in \mathbb{R}%
\ ,
\end{equation*}%
similar to \cite[Theorem 3.1]{OhmII}. Observe that this fact implies, in
particular, that Theorem \ref{Theorem AC conductivity measure} actually
holds with $D = 0$. Such analysis is technically a bit more involved than
the proof above and requires the use of the analyticity of time--correlation
functions of KMS states, but has the advantage of automatically implying the
finiteness of $\mu _{\mathrm{p}}$ with $\mu _{\mathrm{AC}}=\mu _{\mathrm{p}%
}|_{\mathbb{R}\backslash \{0\}}$, similar to the microscopic case \cite[%
Theorem 3.1]{OhmII}.

Note that the case $\lambda =0$ can be interpreted as the perfect conductor.
We show in the subsequent paper that%
\begin{equation*}
\mu _{\mathrm{AC}}^{(\beta ,0)}(\mathbb{R}\backslash \{0\})=0
\end{equation*}%
and hence, the heat production vanishes in this special case. Analogously,
the limit $\lambda \rightarrow \infty $ corresponds to the perfect insulator
and also leads to a vanishing heat production. Positivity of the
AC--conductivity measure means that the fermion system cannot transfer any
energy to the electromagnetic field. This property is a consequence of the
second principle of thermodynamics. In fact, the fermion system even
absorbs, in general, some non--vanishing amount of electromagnetic energy.
These points will be all addressed in the subsequent paper.

\begin{bemerkung}
\mbox{
}\newline
For electric fields slowly varying in time, charge carriers have time to
move and significantly change the charge density, producing an additional,
self--generated, internal electric field. This contribution is not taken
into account in our model. Thus, the physical meaning of the behavior of the
AC--conducti%
\-%
vity measure $\mu _{\mathrm{AC}}$ at low frequencies is not clear, in
general. However, if one imposes regularity of $\mu _{\mathrm{AC}}$ near $%
\nu =0$ then this behavior becomes physically relevant. One natural way to
obtain some regularity of $\mu _{\mathrm{AC}}$ at low frequencies is to
avoid the presence of free charge carriers by imposing some localization
condition. This is done, for instance, in \cite{Annale} where the validity
of Mott's formula for the conductivity of quantum mechanical charged
carriers is studied.
\end{bemerkung}

\section{Technical Proofs\label{Section technical proof Ohm-III}}

We start our study by two technical results that are used in various proofs
of Sections \ref{Section AC--Ohm's Law copy(1)}--\ref{section Current
Fluctuations copy(2)}: a decomposition of complex--time two--point
correlation functions (Section \ref{section Decay Bounds on Propagators})
and a relatively simple extension of the Akcoglu--Krengel ergodic theorem to
non--regular sequences (Section \ref{Sec Akcoglu--Krengel}). Then, we tackle
the proof of Theorem \ref{main 1} in Sections \ref{Section AC--Ohm's Law
copy(1)}--\ref{Section AC--Ohm's Law}. Finally, Section \ref{section Current
Fluctuations copy(2)} justifies the construction done in Section \ref%
{section Current Fluctuations}. In particular, we prove Theorem \ref{scalar
product} in that subsection.

\subsection{Complex--Time Two--Point Correlation Functions\label{section
Decay Bounds on Propagators}}

By \cite[Lemma 5.2]{OhmII}, the microscopic paramagnetic transport
coefficient (\ref{backwards -1bis}) can be expressed in terms of
complex--time two--point correlation functions. The latter are explicitly
given in terms of quantities involving the Anderson tight--binding
Hamiltonian (Section \ref{Section impurities}). Indeed, by \cite[Eq. (101)]%
{OhmII}, for all $\beta \in \mathbb{R}^{+}$, $\omega \in \Omega $, $\lambda
\in \mathbb{R}_{0}^{+}$, $t\in {\mathbb{R}}$, and $\alpha \in \lbrack
0,\beta ]$,%
\begin{equation}
C_{t+i\alpha }^{(\omega )}(\mathbf{x})=\langle \mathfrak{e}_{x^{(2)}},%
\mathrm{e}^{-it\left( \Delta _{\mathrm{d}}+\lambda V_{\omega }\right)
}F_{\alpha }^{\beta }\left( \Delta _{\mathrm{d}}+\lambda V_{\omega }\right)
\mathfrak{e}_{x^{(1)}}\rangle \ ,\quad \mathbf{x}:=(x^{(1)},x^{(2)})\in
\mathfrak{L}^{2}\ ,  \label{correlation operator}
\end{equation}%
where $F_{\alpha }^{\beta }$ is the real function defined, for every $\beta
\in \mathbb{R}^{+}$ and $\alpha \in {\mathbb{R}}$, by
\begin{equation}
F_{\alpha }^{\beta }\left( \varkappa \right) :=\frac{\mathrm{e}^{\alpha
\varkappa }}{1+\mathrm{e}^{\beta \varkappa }}\ ,\qquad \varkappa \in {%
\mathbb{R}}\ .  \label{Schwartz function F}
\end{equation}%
Equation (\ref{correlation operator}) provides useful estimates like
space--decay properties of $C_{t+i\alpha }^{(\omega )}$. Note that the
notation $\Vert \cdot \Vert _{\mathrm{op}}$ stands for the operator norm.

\begin{satz}[Decomposition of two-point correlation functions]
\label{decay bound theorem}\mbox{
}\newline
For any $\varepsilon ,\beta \in \mathbb{R}^{+}$, $\omega \in \Omega $, $%
\lambda \in \mathbb{R}_{0}^{+}$, $t\in {\mathbb{R}}$, $\upsilon \in (0,\beta
/2)$ and $\alpha \in \lbrack \upsilon ,\beta -\upsilon ]$, the complex--time
two--point correlation function $C_{t+i\alpha }^{(\omega )}$ can be
decomposed as%
\begin{equation*}
C_{t+i\alpha }^{(\omega )}\left( \mathbf{x}\right) =A_{t+i\alpha ,\upsilon
,\varepsilon }^{(\omega )}\left( \mathbf{x}\right) +B_{t+i\alpha ,\upsilon
,\varepsilon }^{(\omega )}\left( \mathbf{x}\right) \ ,\qquad \mathbf{x}%
:=(x^{(1)},x^{(2)})\in \mathfrak{L}^{2}\ ,
\end{equation*}%
where $A_{t+i\alpha ,\upsilon ,\varepsilon }^{(\omega )}\left( \cdot \right)
$ and $B_{t+i\alpha ,\upsilon ,\varepsilon }^{(\omega )}\left( \cdot \right)
$ are kernels (w.r.t. $\{\mathfrak{e}_{x}\}_{x\in \mathfrak{L}}$) of bounded
operators $A_{t+i\alpha ,\upsilon ,\varepsilon }^{(\omega )}\equiv
A_{t+i\alpha ,\upsilon ,\varepsilon }^{(\beta ,\omega ,\lambda )}$ and $%
B_{t+i\alpha ,\upsilon ,\varepsilon }^{(\omega )}\equiv B_{t+i\alpha
,\upsilon ,\varepsilon }^{(\beta ,\omega ,\lambda )}$ acting on $\ell ^{2}(%
\mathfrak{L})$ and satisfying the following properties: \newline
\emph{(i)} Boundedness: There is a finite constant $D\in \mathbb{R}^{+}$
only depending on $\beta ,\upsilon $ such that%
\begin{equation*}
\left\Vert A_{t+i\alpha ,\upsilon ,\varepsilon }^{(\omega )}\right\Vert _{%
\mathrm{op}}\leq \varepsilon \qquad \text{and}\qquad \left\Vert B_{t+i\alpha
,\upsilon ,\varepsilon }^{(\omega )}\right\Vert _{\mathrm{op}}\leq D\ .
\end{equation*}%
\emph{(ii)} Decay: If $T\in \mathbb{R}^{+}$ and $t\in \lbrack -T,T]$, then
there is a finite constant $D\in \mathbb{R}^{+}$ only depending on $%
\varepsilon ,\beta ,\upsilon ,d,T$ such that%
\begin{equation*}
\left\vert B_{t+i\alpha ,\upsilon ,\varepsilon }^{(\omega )}\left( \mathbf{x}%
\right) \right\vert \leq \frac{D}{1+|x^{(1)}-x^{(2)}|^{d^{2}+1}}\ ,\qquad
\mathbf{x}\in \mathfrak{L}^{2}\ .
\end{equation*}%
\emph{(iii)} Continuity w.r.t. times: If $T\in \mathbb{R}^{+}$ and $%
s_{1},s_{2}\in \lbrack -T,T]$, then there is a finite constant $\eta \in
\mathbb{R}^{+}$ only depending on $\varepsilon ,\beta ,\upsilon ,d,T$ such
that
\begin{equation*}
\left\vert B_{s_{1}+i\alpha ,\upsilon ,\varepsilon }^{(\omega )}\left(
\mathbf{x}\right) -B_{s_{2}+i\alpha ,\upsilon ,\varepsilon }^{(\omega
)}\left( \mathbf{x}\right) \right\vert \leq \frac{\varepsilon \left(
1+\lambda \right) }{1+|x^{(1)}-x^{(2)}|^{d^{2}+1}}\ ,\qquad \mathbf{x}\in
\mathfrak{L}^{2}\ ,
\end{equation*}%
whenever $\left\vert s_{2}-s_{1}\right\vert \leq \eta $.\newline
\emph{(iv)} Continuity w.r.t. random variables: For any $\mathbf{x}\in
\mathfrak{L}^{2}$, the maps
\begin{equation*}
\omega \mapsto C_{t+i\alpha }^{(\omega )}\left( \mathbf{x}\right) \ ,\
\omega \mapsto A_{t+i\alpha ,\upsilon ,\varepsilon }^{(\omega )}\left(
\mathbf{x}\right) \ ,\ \omega \mapsto B_{t+i\alpha ,\upsilon ,\varepsilon
}^{(\omega )}\left( \mathbf{x}\right)
\end{equation*}%
from $\Omega $ to $\mathbb{R}$ are continuous w.r.t. the topology on $\Omega
$ of which $\mathfrak{A}_{\Omega }$ is the Borel $\sigma $--algebra .
\end{satz}

\begin{proof}
(i) The spectral theorem applied to the bounded self--adjoint operator $%
(\Delta _{\mathrm{d}}+\lambda V_{\omega })\in \mathcal{B}(\ell ^{2}(%
\mathfrak{L}))$ implies from (\ref{correlation operator}) that
\begin{equation*}
C_{t+i\alpha }^{(\omega )}(\mathbf{x})=\int F_{\alpha }^{\beta }(\varkappa )%
\mathrm{e}^{-it\varkappa }\mathrm{d}\kappa _{\mathbf{x}}^{(\omega )}\left(
\varkappa \right)
\end{equation*}%
with $\kappa _{\mathbf{x}}^{(\omega )}\equiv \kappa _{\mathbf{x}}^{(\omega
,\lambda )}$ being the spectral measure of $(\Delta _{\mathrm{d}}+\lambda
V_{\omega })$ w.r.t. $\mathfrak{e}_{x^{(1)}},\mathfrak{e}_{x^{(2)}}\in \ell
^{2}(\mathfrak{L})$. Note that $F_{\alpha }^{\beta }$ (\ref{Schwartz
function F}) is a Schwartz function for all $\beta \in \mathbb{R}^{+}$ and $%
\alpha \in (0,\beta )$. Therefore, its Fourier transform $\hat{F}_{\alpha
}^{\beta }$ is again a Schwartz function. Moreover, for all $\beta >0$ and $%
\upsilon \in (0,\beta /2)$, there is a finite constant $D_{\beta ,\upsilon
}\in \mathbb{R}^{+}$ such that, for any $\alpha \in \lbrack \upsilon ,\beta
-\upsilon ]$ and all $\nu \in \mathbb{R}$,
\begin{equation}
\left\vert \hat{F}_{\alpha }^{\beta }\left( \nu \right) \right\vert \leq
\frac{D_{\beta ,\upsilon }}{1+\nu ^{2}}\ .  \label{inequality=}
\end{equation}%
In particular, for any $\varepsilon \in \mathbb{R}^{+}$, there is $M_{\beta
,\upsilon ,\varepsilon }\in \mathbb{R}^{+}$ such that%
\begin{equation}
\int\nolimits_{|\nu |\geq M_{\beta ,\upsilon ,\varepsilon }}\left\vert \hat{F%
}_{\alpha }^{\beta }\left( \nu \right) \right\vert \mathrm{d}\nu \leq
\int\nolimits_{|\nu |\geq M_{\beta ,\upsilon ,\varepsilon }}\frac{D_{\beta
,\upsilon }}{1+\nu ^{2}}\mathrm{d}\nu <\varepsilon \ .  \label{inequality}
\end{equation}%
For any $\varepsilon ,\beta \in \mathbb{R}^{+}$, $\upsilon \in (0,\beta /2)$
and $\alpha \in \lbrack \upsilon ,\beta -\upsilon ]$, we then decompose the
function $F_{\alpha }^{\beta }$ into two orthogonal functions of $\varkappa
\in {\mathbb{R}}$:%
\begin{eqnarray}
f_{\upsilon ,\varepsilon ,\alpha }^{\beta }\left( \varkappa \right)
&:=&\int\nolimits_{|\nu |\geq M_{\beta ,\upsilon ,\varepsilon }}\hat{F}_{\alpha }^{\beta }\left( \nu \right) \mathrm{e}^{i\nu \varkappa }\mathrm{d}\nu \ ,  \label{definition de f} \\
g_{\upsilon ,\varepsilon ,\alpha }^{\beta }\left( \varkappa \right)
&:=&\int\nolimits_{|\nu |<M_{\beta ,\upsilon ,\varepsilon }}\hat{F}_{\alpha
}^{\beta }\left( \nu \right) \mathrm{e}^{i\nu \varkappa }\mathrm{d}\nu \ .
\label{definition de g}
\end{eqnarray}%
Now, for any $\varepsilon ,\beta \in \mathbb{R}^{+}$, $\omega \in \Omega $, $%
\lambda \in \mathbb{R}_{0}^{+}$, $t\in {\mathbb{R}}$, $\upsilon \in (0,\beta
/2)$ and $\alpha \in \lbrack \upsilon ,\beta -\upsilon ]$, define the
bounded operators $A_{t+i\alpha ,\upsilon ,\varepsilon }^{(\omega )}\equiv
A_{t+i\alpha ,\upsilon ,\varepsilon }^{(\beta ,\omega ,\lambda )}$ and $%
B_{t+i\alpha ,\upsilon ,\varepsilon }^{(\omega )}\equiv B_{t+i\alpha
,\upsilon ,\varepsilon }^{(\beta ,\omega ,\lambda )}$ acting on $\ell ^{2}(%
\mathfrak{L})$ by their kernels
\begin{eqnarray}
\langle \mathfrak{e}_{x^{(2)}},A_{t+i\alpha ,\upsilon ,\varepsilon
}^{(\omega )}\mathfrak{e}_{x^{(1)}}\rangle &\equiv &A_{t+i\alpha ,\upsilon
,\varepsilon }^{(\omega )}\left( \mathbf{x}\right) :=\int f_{\upsilon
,\varepsilon ,\alpha }^{\beta }\left( \varkappa \right) \mathrm{e}%
^{-it\varkappa }\mathrm{d}\kappa _{\mathbf{x}}^{(\omega )}\left( \varkappa
\right)  \label{definition de A} \\
\langle \mathfrak{e}_{x^{(2)}},B_{t+i\alpha ,\upsilon ,\varepsilon
}^{(\omega )}\mathfrak{e}_{x^{(1)}}\rangle &\equiv &B_{t+i\alpha ,\upsilon
,\varepsilon }^{(\omega )}\left( \mathbf{x}\right) :=\int g_{\upsilon
,\varepsilon ,\alpha }^{\beta }\left( \varkappa \right) \mathrm{e}%
^{-it\varkappa }\mathrm{d}\kappa _{\mathbf{x}}^{(\omega )}\left( \varkappa
\right)  \label{definition de B}
\end{eqnarray}%
for any $\mathbf{x}\in \mathfrak{L}^{2}$. Indeed, by construction (cf. (\ref%
{inequality})--(\ref{definition de f})),%
\begin{equation*}
\left\Vert A_{t+i\alpha ,\upsilon ,\varepsilon }^{(\omega )}\right\Vert _{%
\mathrm{op}}\leq \varepsilon \qquad \text{and}\qquad \left\Vert B_{t+i\alpha
,\upsilon ,\varepsilon }^{(\omega )}\right\Vert _{\mathrm{op}}\leq \pi
D_{\beta ,\upsilon }
\end{equation*}%
for all $\varepsilon ,\beta \in \mathbb{R}^{+}$, $\omega \in \Omega $, $%
\lambda \in \mathbb{R}_{0}^{+}$, $t\in {\mathbb{R}}$, $\upsilon \in (0,\beta
/2)$ and $\alpha \in \lbrack \upsilon ,\beta -\upsilon ]$. By (\ref%
{inequality=}), recall that $D_{\beta ,\upsilon }$ only depends on $\beta $
and $\upsilon \in (0,\beta /2)$.

\noindent (ii) We first invoke Fubini's theorem to observe from (\ref%
{definition de g})--(\ref{definition de B}) that%
\begin{eqnarray}
B_{t+i\alpha ,\upsilon ,\varepsilon }^{(\omega )}\left( \mathbf{x}\right)
&=&\int\nolimits_{|\nu |<M_{\beta ,\upsilon ,\varepsilon }}\mathrm{d}\nu \
\hat{F}_{\alpha }^{\beta }\left( \nu \right) \int \mathrm{d}\kappa _{\mathbf{%
x}}^{(\omega )}\left( \varkappa \right) \mathrm{e}^{-i\varkappa \left( t-\nu
\right) }  \notag \\
&=&\int\nolimits_{|\nu |<M_{\beta ,\upsilon ,\varepsilon }}\mathrm{d}\nu \
\hat{F}_{\alpha }^{\beta }\left( \nu \right) \langle \mathfrak{e}_{x^{(2)}},%
\mathrm{e}^{-i\left( t-\nu \right) \left( \Delta _{\mathrm{d}}+\lambda
V_{\omega }\right) }\mathfrak{e}_{x^{(1)}}\rangle  \label{B deux}
\end{eqnarray}%
for all $\mathbf{x}\in \mathfrak{L}^{2}$. If $T\in \mathbb{R}^{+}$, $t\in
\lbrack -T,T]$ and $|\nu |<M_{\beta ,\upsilon ,\varepsilon }$, then
\begin{equation*}
(t-\nu )\in \lbrack -M_{\beta ,\upsilon ,\varepsilon }-T,M_{\beta ,\upsilon
,\varepsilon }+T]\ .
\end{equation*}%
Thus, by \cite[Lemma 4.2]{OhmI} with $\epsilon =d^{2}-d+1$ ($d\in \mathbb{N}$%
), for any $\varepsilon ,\beta ,T\in \mathbb{R}^{+}$ and $\upsilon \in
(0,\beta /2)$, there is a finite constant $\tilde{D}_{\beta ,\upsilon
,\varepsilon ,T,d}\in \mathbb{R}^{+}$ such that%
\begin{equation}
\left\vert \langle \mathfrak{e}_{x^{(2)}},\mathrm{e}^{-i\left( t-\nu \right)
\left( \Delta _{\mathrm{d}}+\lambda V_{\omega }\right) }\mathfrak{e}%
_{x^{(1)}}\rangle \right\vert \leq \frac{\tilde{D}_{\beta ,\upsilon
,\varepsilon ,T,d}}{1+|x^{(1)}-x^{(2)}|^{d^{2}+1}}
\label{estimate ohm idiot}
\end{equation}%
for all $\omega \in \Omega $, $\lambda \in \mathbb{R}_{0}^{+}$, $t\in
\lbrack -T,T\mathbb{]}$, $\nu \in \lbrack -M_{\beta ,\upsilon ,\varepsilon
},M_{\beta ,\upsilon ,\varepsilon }]$ and $\mathbf{x}\in \mathfrak{L}^{2}$.
We now combine this last inequality with (\ref{inequality=}) and (\ref{B
deux}) to derive the bound%
\begin{equation*}
\left\vert B_{t+i\alpha ,\upsilon ,\varepsilon }^{(\omega )}\left( \mathbf{x}%
\right) \right\vert \leq \frac{\pi D_{\beta ,\upsilon }\tilde{D}_{\beta
,\upsilon ,\varepsilon ,T,d}}{1+|x^{(1)}-x^{(2)}|^{d^{2}+1}}\ ,\qquad
\mathbf{x}\in \mathfrak{L}^{2}\ .
\end{equation*}

\noindent (iii) By Equation (\ref{B deux}), note that
\begin{equation}
\partial _{t}B_{t+i\alpha ,\upsilon ,\varepsilon }^{(\omega )}\left( \mathbf{%
x}\right) =-i\int\nolimits_{|\nu |<M_{\beta ,\upsilon ,\varepsilon }}\mathrm{%
d}\nu \ \hat{F}_{\alpha }^{\beta }\left( \nu \right) \langle \left( \Delta _{%
\mathrm{d}}+\lambda V_{\omega }\right) \mathfrak{e}_{x^{(2)}},\mathrm{e}%
^{-i\left( t-\nu \right) \left( \Delta _{\mathrm{d}}+\lambda V_{\omega
}\right) }\mathfrak{e}_{x^{(1)}}\rangle  \label{B deuxbis}
\end{equation}%
for all $\varepsilon ,\beta \in \mathbb{R}^{+}$, $\omega \in \Omega $, $%
\lambda \in \mathbb{R}_{0}^{+}$, $t\in {\mathbb{R}}$, $\upsilon \in (0,\beta
/2)$, $\alpha \in \lbrack \upsilon ,\beta -\upsilon ]$ and $\mathbf{x}\in
\mathfrak{L}^{2}$. Since, for any $\mathbf{x}\in \mathfrak{L}^{2}$,%
\begin{eqnarray*}
&&\langle \left( \Delta _{\mathrm{d}}+\lambda V_{\omega }\right) \mathfrak{e}%
_{x^{(2)}},\mathrm{e}^{-i\left( t-\nu \right) \left( \Delta _{\mathrm{d}%
}+\lambda V_{\omega }\right) }\mathfrak{e}_{x^{(1)}}\rangle \\
&=&-\sum\limits_{z\in \mathfrak{L},|z|=1}\langle \mathfrak{e}_{x^{(2)}+z},%
\mathrm{e}^{-i\left( t-\nu \right) \left( \Delta _{\mathrm{d}}+\lambda
V_{\omega }\right) }\mathfrak{e}_{x^{(1)}}\rangle \\
&&+(\lambda V_{\omega }(x^{(2)})+2d)\langle \mathfrak{e}_{x^{(2)}},\mathrm{e}%
^{-i\left( t-\nu \right) \left( \Delta _{\mathrm{d}}+\lambda V_{\omega
}\right) }\mathfrak{e}_{x^{(1)}}\rangle \ ,
\end{eqnarray*}%
we use again (\ref{inequality=}) and (\ref{estimate ohm idiot}) together
with (\ref{B deuxbis}) and $|V_{\omega }\left( x\right) |\leq 1$ to arrive
at the third assertion.

\noindent (iv) Take any sequence $\{\omega _{n}\}_{n=1}^{\infty }\subset
\Omega $ converging to $\omega _{\infty }\in \Omega $ w.r.t. the topology of
which $\mathfrak{A}_{\Omega }$ is the Borel $\sigma $--algebra. This means
that the functions $\omega _{n}:\mathfrak{L}\rightarrow \lbrack -1,1]$, $%
n\in \mathbb{N}$, converges pointwise to $\omega _{\infty }$, as $%
n\rightarrow \infty $. By Lebesgue's dominated convergence theorem, it
follows that the sequence $\{\Delta _{\mathrm{d}}+\lambda V_{\omega
_{n}}\}_{n=1}^{\infty }$ of uniformly bounded operators at fixed $\lambda
\in \mathbb{R}_{0}^{+}$ converges strongly to $\Delta _{\mathrm{d}}+\lambda
V_{\omega _{\infty }}$. By \cite[Chap. VIII, Problem 28 and Theorem VIII.20
(b)]{ReedSimonI}, for any bounded and continuous function $\varphi $ on $%
\mathbb{R}$, the sequence $\{\varphi (\Delta _{\mathrm{d}}+\lambda V_{\omega
_{n}})\}_{n=1}^{\infty }$ converges also strongly to $\varphi (\Delta _{%
\mathrm{d}}+\lambda V_{\omega _{\infty }})$.

Now, similar to Equation (\ref{correlation operator}), Definitions (\ref%
{definition de A})--(\ref{definition de B}) can be rewritten as
\begin{eqnarray}
A_{t+i\alpha ,\upsilon ,\varepsilon }^{(\omega )}\left( \mathbf{x}\right)
&=&\langle \mathfrak{e}_{x^{(2)}},\mathrm{e}^{-it\left( \Delta _{\mathrm{d}%
}+\lambda V_{\omega }\right) }f_{\upsilon ,\varepsilon ,\alpha }^{\beta
}\left( \Delta _{\mathrm{d}}+\lambda V_{\omega }\right) \mathfrak{e}%
_{x^{(1)}}\rangle  \label{definition de Abis} \\
B_{t+i\alpha ,\upsilon ,\varepsilon }^{(\omega )}\left( \mathbf{x}\right)
&=&\langle \mathfrak{e}_{x^{(2)}},\mathrm{e}^{-it\left( \Delta _{\mathrm{d}%
}+\lambda V_{\omega }\right) }g_{\upsilon ,\varepsilon ,\alpha }^{\beta
}\left( \Delta _{\mathrm{d}}+\lambda V_{\omega }\right) \mathfrak{e}%
_{x^{(1)}}\rangle  \label{definition de Bbis}
\end{eqnarray}%
for every $\mathbf{x}\in \mathfrak{L}^{2}$. By (\ref{Schwartz function F})
and (\ref{definition de f})--(\ref{definition de g}), for any $\varepsilon
,\beta \in \mathbb{R}^{+}$, $\omega \in \Omega $, $\lambda \in \mathbb{R}%
_{0}^{+}$, $t\in {\mathbb{R}}$, $\upsilon \in (0,\beta /2)$ and $\alpha \in
\lbrack \upsilon ,\beta -\upsilon ]$, $F_{\alpha }^{\beta }$, $f_{\upsilon
,\varepsilon ,\alpha }^{\beta }$ and $g_{\upsilon ,\varepsilon ,\alpha
}^{\beta }$ are bounded and continuous function on $\mathbb{R}$. Therefore,
for every $\mathbf{x}\in \mathfrak{L}^{2}$ and as $n\rightarrow \infty $,
the correlation functions $A_{t+i\alpha ,\upsilon ,\varepsilon }^{(\omega
_{n})}\left( \mathbf{x}\right) $, $B_{t+i\alpha ,\upsilon ,\varepsilon
}^{(\omega _{n})}\left( \mathbf{x}\right) $ and $C_{t+i\alpha }^{(\omega
_{n})}\left( \mathbf{x}\right) $ converges to $A_{t+i\alpha ,\upsilon
,\varepsilon }^{(\omega _{\infty })}\left( \mathbf{x}\right) $, $%
B_{t+i\alpha ,\upsilon ,\varepsilon }^{(\omega _{\infty })}\left( \mathbf{x}%
\right) $ and $C_{t+i\alpha }^{(\omega _{\infty })}\left( \mathbf{x}\right) $%
, respectively.
\end{proof}

Better estimates on complex--time two--point correlation functions $%
C_{t+i\alpha }^{(\omega )}$ can certainly be obtained by using that the
spectrum of the self--adjoint operator $(\Delta _{\mathrm{d}}+\lambda
V_{\omega })$ belongs to some ($\lambda $--dependant) compact set. This
property is however not used in Theorem \ref{decay bound theorem} to get
bounds (i)--(ii) that do not depend on $\lambda \in \mathbb{R}_{0}^{+}$.
Note that we only need here the measurability w.r.t. the $\sigma $--algebra $%
\mathfrak{A}_{\Omega }$ of all operators of Theorem \ref{decay bound theorem}%
, which is a direct consequence of their continuity, see Theorem \ref{decay
bound theorem} (iv).

\subsection{Ergodic Theorem for some Non--Regular Sequences\label{Sec
Akcoglu--Krengel}}

The second important ingredient we use in our proofs is the Akcoglu--Krengel
ergodic theorem. We present it for completeness. This result is rather
standard and can be found in textbooks. Therefore, we keep the exposition as
short as possible and only concentrate on results used in this paper. For
more details, we recommend \cite{birkoff}. It is important to note, however,
that Theorem \ref{Ackoglu--Krengel ergodic theorem II} is a relatively
simple extension of \cite[Theorem VI.1.7, Remark VI.1.8]{birkoff} to
non--regular sequences.

We restrict ourselves to \emph{additive }processes associated with the
probability space $(\Omega ,\mathfrak{A}_{\Omega },\mathfrak{a}_{\Omega })$
defined in Section \ref{Section impurities}, even if the Akcoglu--Krengel
ergodic theorem holds for superadditive or subadditive ones (cf. \cite[%
Definition VI.1.6]{birkoff}).

\begin{definition}[Additive process associated with random variables]
\label{Additive process}\mbox{ }\newline
$\{\mathfrak{F}^{(\omega )}\left( \Lambda \right) \}_{\Lambda \in \mathcal{P}%
_{f}(\mathfrak{L})}$ is an additive process if:\newline
\noindent (i) the map $\omega \mapsto \mathfrak{F}^{(\omega )}\left( \Lambda
\right) $ is bounded and measurable w.r.t. the $\sigma $--algebra $\mathfrak{%
A}_{\Omega }$ for any $\Lambda \in \mathcal{P}_{f}(\mathfrak{L})$.\newline
\noindent (ii) For all disjoint $\Lambda _{1},\Lambda _{2}\in \mathcal{P}%
_{f}(\mathfrak{L})$,
\begin{equation*}
\mathfrak{F}^{(\omega )}\left( \Lambda _{1}\cup \Lambda _{2}\right) =%
\mathfrak{F}^{(\omega )}\left( \Lambda _{1}\right) +\mathfrak{F}^{(\omega
)}\left( \Lambda _{2}\right) \ ,\qquad \omega \in \Omega \ .
\end{equation*}
\noindent (iii) For all $\Lambda \in \mathcal{P}_{f}(\mathfrak{L})$ and any
space shift $x \in \mathfrak{L}$,
\begin{equation}
\mathbb{E}\left[ \mathfrak{F}^{(\omega )}\left( \Lambda \right) \right] =
\mathbb{E}\left[ \mathfrak{F}^{(\omega )}\left(x+ \Lambda \right) \right] \ .
\label{equality a la con}
\end{equation}
\end{definition}

The random potentials used here are independently and identically
distributed (i.i.d.), see Equation (\ref{probability measure}), and (\ref%
{equality a la con}) will trivially hold for the processes we consider
below. Recall that $\mathbb{E}[\ \cdot \ ]$ is the expectation value
associated with the probability measure $\mathfrak{a}_{\Omega }$. Note
further that additive processes $\{\mathfrak{F}^{(\omega )}\left( \Lambda
\right) \}_{\Lambda \in \mathcal{P}_{f}(\mathfrak{L})}$ as defined in
Definition \ref{Additive process} are superadditive and subadditive in the
sense of \cite[Definition VI.1.6]{birkoff}.

We now define \emph{regular }sequences (cf. \cite[Remark VI.1.8]{birkoff})
as follows:

\begin{definition}[Regular sequences]
\label{regular sequences}\mbox{ }\newline
The family $\{\Lambda ^{(l)}\}_{l\in \mathbb{R}^+}\subset \mathcal{P}_{f}(%
\mathfrak{L})$ of non--decreasing (possibly non--cubic) boxes of $\mathfrak{L%
}$ is a regular sequence if there is a finite constant $D\in \mathbb{R}^{+}$
and another non--decreasing sequence of boxes $\{\Lambda _{l}\}_{l\geq 1}$,
given by (\ref{eq:def lambda n}), such that $\mathfrak{L}=\cup _{l\geq
1}\Lambda _{l}$, $\Lambda ^{(l)}\subset \Lambda _{l}$ and $0<|\Lambda
_{l}|\leq D|\Lambda ^{(l)}|$ for all $l\geq 1$.
\end{definition}

Then, the form of Akcoglu--Krengel ergodic theorem we use in the sequel is
the lattice version of \cite[Theorem VI.1.7, Remark VI.1.8]{birkoff} for
additive processes associated with the probability space $(\Omega ,\mathfrak{%
A}_{\Omega },\mathfrak{a}_{\Omega })$:

\begin{satz}[Akcoglu--Krengel ergodic theorem]
\label{Ackoglu--Krengel ergodic theorem II copy(1)}\mbox{
}\newline
Let $\{\mathfrak{F}^{(\omega )}\left( \Lambda \right) \}_{\Lambda \in
\mathcal{P}_{f}(\mathfrak{L})}$ be an additive process. Then, for any
regular sequence $\{\Lambda ^{(l)}\}_{l\in \mathbb{R}^{+}}\subset \mathcal{P}%
_{f}(\mathfrak{L})$, there is a measurable subset $\tilde{\Omega}\subset
\Omega $ of full measure such that, for all $\tilde{\omega}\in \tilde{\Omega}
$,%
\begin{equation*}
\underset{l\rightarrow \infty }{\lim }\left\{ \left\vert \Lambda
^{(l)}\right\vert ^{-1}\mathfrak{F}^{(\tilde{\omega})}\left( \Lambda
^{(l)}\right) \right\} =\mathbb{E}\left[ \mathfrak{F}^{(\omega )}\left(
\left\{ 0\right\} \right) \right] \ .
\end{equation*}
\end{satz}

The Ackoglu--Krengel (superadditive) ergodic theorem, cornerstone of ergodic
theory, generalizes the celebrated Birkhoff additive ergodic theorem.
Unfortunately, this theorem, in the above form, is not sufficiently general
to be applied in our proofs. Indeed, Theorem \ref{Ackoglu--Krengel ergodic
theorem II copy(1)} requires regular sequences. This is too restrictive
w.r.t. our applications because we have to evaluate space--inhomogeneous
limits of the form
\begin{equation}
\underset{l\rightarrow \infty }{\lim }\frac{1}{\left\vert \Lambda
_{l}\right\vert }\underset{x\in \Lambda _{l}}{\sum }\mathfrak{F}^{(\omega
)}\left( \left\{ x\right\} \right) f\left( l^{-1}x\right)  \label{limit easy}
\end{equation}%
with $f\in C_{0}\left( \mathbb{R}^{d},\mathbb{R}\right) $ and $\{\Lambda
_{l}\}_{l\geq 1}$ defined by (\ref{eq:def lambda n}). See for instance
Section \ref{Section AC--Ohm's Law copy(1)}.

To this end, we divide the compact support of $f$, say for simplicity $%
[-1/2,1/2]^{d}$, in $n^{d}$ boxes $\{b_{j}\}_{j\in \mathcal{D}_{n}}$ of
side--length $1/n$, where
\begin{equation}
\mathcal{D}_{n}:=\{-\left( n-1\right) /2,-\left( n-3\right) /2,\ldots
,\left( n-3\right) /2,\left( n-1\right) /2\}^{d}\ .  \label{boxes b1}
\end{equation}%
Explicitly, for any $j\in \mathcal{D}_{n}$,%
\begin{equation}
b_{j}:=jn^{-1}+n^{-1}[-1/2,1/2]^{d}\text{\quad and\quad }[-1/2,1/2]^{d}=%
\underset{j\in \mathcal{D}_{n}}{\bigcup }b_{j}\ .  \label{boxes b2}
\end{equation}%
We then need to analyze the limit%
\begin{equation*}
\underset{l\rightarrow \infty }{\lim }\left\vert \mathfrak{L}\cap
(lb_{j})\right\vert ^{-1}\mathfrak{F}^{(\omega )}\left( \mathfrak{L}\cap
(lb_{j})\right)
\end{equation*}%
for $n\in \mathbb{N}$ and $j\in \mathcal{D}_{n}$. However, $\{\mathfrak{L}%
\cap (lb_{j})\}_{l\in \mathbb{N}}$ is \emph{non--regular}, in general. For
instance, if $n$ is an odd integer then this situation appears for all $j\in
\mathcal{D}_{n}\backslash \{(0,\ldots ,0)\}$ because $\{\mathfrak{L}\cap
(lb_{j})\}_{l\in \mathbb{N}}$ is not a non--decreasing sequence in this
case. To overcome this difficulty, we proof the following extension of
Theorem \ref{Ackoglu--Krengel ergodic theorem II copy(1)}:

\begin{satz}[Ergodic theorem for some non--regular sequences]
\label{Ackoglu--Krengel ergodic theorem II}\mbox{
}\newline
Let $\{\mathfrak{F}^{(\omega )}\left( \Lambda \right) \}_{\Lambda \in
\mathcal{P}_{f}(\mathfrak{L})}$ be an additive process. Then, there is a
measurable subset $\tilde{\Omega}\subset \Omega $ of full measure such that,
for all $\tilde{\omega}\in \tilde{\Omega}$, $n\in \mathbb{N}$, and $j\in
\mathcal{D}_{n}$,
\begin{equation*}
\underset{l\rightarrow \infty }{\lim }\left\{ \left\vert \mathfrak{L}\cap
(lb_{j})\right\vert ^{-1}\mathfrak{F}^{(\tilde{\omega})}\left( \mathfrak{L}%
\cap (lb_{j})\right) \right\} =\mathbb{E}\left[ \mathfrak{F}^{(\omega
)}\left( \left\{ 0\right\} \right) \right] \ .
\end{equation*}
\end{satz}

\begin{proof}
Let $n\in \mathbb{N}$. By Theorem \ref{Ackoglu--Krengel ergodic theorem II
copy(1)}, we can fix w.l.o.g. the parameter $j\equiv (j_{1},\ldots
,j_{d})\in \mathcal{D}_{n}$ such that the family $\{\mathfrak{L}\cap
(lb_{j})\}_{l\in \mathbb{N}}$ is non--regular. Then, we take the sequences $%
\{\Lambda ^{(l,j)}\}_{l\in \mathbb{N}}$ and $\{\tilde{\Lambda}%
^{(l,j)}\}_{l\in \mathbb{N}}$ defined, for any $l\in \mathbb{R}^{+}$, by
\begin{equation}
\Lambda ^{(l,j)}:=\left\{ (x_{1},\ldots ,x_{d})\in \mathfrak{L}\,:\,\forall
k\in \{1,\ldots ,d\},\ \left\vert x_{k}\right\vert \leq
l(|j_{k}|+1/2)n^{-1}+1\right\}  \label{regular1}
\end{equation}%
and
\begin{equation}
\tilde{\Lambda}^{(l,j)}:=\Lambda ^{(l,j)}\backslash \{\mathfrak{L}\cap
(lb_{j})\}\ .  \label{regular2}
\end{equation}%
In particular,
\begin{equation}
\mathfrak{F}^{(\omega )}(lb_{j})=\mathfrak{F}^{(\omega )}(\Lambda ^{(l,j)})-%
\mathfrak{F}^{(\omega )}(\tilde{\Lambda}^{(l,j)})  \label{regular3}
\end{equation}%
because $\{\mathfrak{F}^{(\omega )}\left( \Lambda \right) \}_{\Lambda \in
\mathcal{P}_{f}(\mathfrak{L})}$ is by assumption an additive process. Note
that $\{\Lambda ^{(l,j)}\}_{l\in \mathbb{N}}$ is a regular sequence and thus
\begin{equation*}
\underset{l\rightarrow \infty }{\lim }\left\{ |\Lambda ^{(l,j)}|^{-1}%
\mathfrak{F}^{(\tilde{\omega})}(\Lambda ^{(l,j)})\right\} =\mathbb{E}\left[
\mathfrak{F}^{(\omega )}\left( \left\{ 0\right\} \right) \right]
\end{equation*}%
almost surely, by Theorem \ref{Ackoglu--Krengel ergodic theorem II copy(1)}.
$\{\tilde{\Lambda}^{(l,j)}\}_{l\in \mathbb{N}}$ satisfies Definition \ref%
{regular sequences}, up to the fact that it is not a sequence of boxes.
Indeed, we can obtain $\tilde{\Lambda}^{(l,j)}$ by subtracting from $\Lambda
^{(l,j)}$ $d$ boxes of the form
\begin{equation*}
\Lambda ^{(l,j)}\cap \{x\in \mathfrak{L}\;|\;x_{k}\lessgtr l(j_{k}\pm
1/2)n^{-1}\pm 1\},\quad k=1,\ldots ,d\ ,
\end{equation*}%
containing the origin of $\mathfrak{L}$. By applying Theorem \ref%
{Ackoglu--Krengel ergodic theorem II copy(1)} to the corresponding regular
sequences of boxes we arrive at:
\begin{equation*}
\underset{l\rightarrow \infty }{\lim }\left\{ |\tilde{\Lambda}^{(l,j)}|^{-1}%
\mathfrak{F}^{(\tilde{\omega})}(\tilde{\Lambda}^{(l,j)})\right\} =\mathbb{E}%
\left[ \mathfrak{F}^{(\omega )}\left( \left\{ 0\right\} \right) \right] \ .
\end{equation*}%
We omit the details. Therefore, by Theorem \ref{Ackoglu--Krengel ergodic
theorem II copy(1)} and (\ref{regular3}), there is a measurable subset $\hat{%
\Omega}_{j,n}\equiv \hat{\Omega}_{j,n}^{(\beta ,\lambda )}\subset \Omega $
of full measure such that, for any $\tilde{\omega}\in \hat{\Omega}_{j,n}$,
\begin{equation}
\underset{l\rightarrow \infty }{\lim }\left\{ \left\vert \mathfrak{L}\cap
(lb_{j})\right\vert ^{-1}\mathfrak{F}^{(\tilde{\omega})}\left( \mathfrak{L}%
\cap (lb_{j})\right) \right\} =\mathbb{E}\left[ \mathfrak{F}^{(\omega
)}\left( \left\{ 0\right\} \right) \right] \ .  \label{additive new}
\end{equation}%
Note that we have used here that the intersection of any \emph{countable}
intersection of measurable sets of full measure has full measure. This fact
is used many times in our proofs.

It follows that (\ref{additive new}) holds true for any $n\in \mathbb{N}$, $%
j\in \mathcal{D}_{n}$, and $\tilde{\omega}\in \hat{\Omega}_{j,n}$, while the
measurable subset defined by%
\begin{equation*}
\tilde{\Omega}:=\underset{n\in \mathbb{N}}{\bigcap }\ \underset{j\in
\mathcal{D}_{n}}{\bigcap }\ \hat{\Omega}_{j,n}\subset \Omega
\end{equation*}%
has full measure.
\end{proof}

Note that the notion of a regular sequence is not completely consistent in
the literature. We used here the definition given in \cite[Remark VI.1.8]%
{birkoff} and then generalized Theorem \ref{Ackoglu--Krengel ergodic theorem
II copy(1)} to some, w.r.t. this definition, non--regular sequences in the
above theorem.

Theorem \ref{Ackoglu--Krengel ergodic theorem II} directly yields the limit (%
\ref{limit easy}):

\begin{satz}[Space--inhomogeneous ergodic theorem]
\label{Ackoglu--Krengel ergodic theorem III}\mbox{
}\newline
Let $\{\mathfrak{F}^{(\omega )}\left( \Lambda \right) \}_{\Lambda \in
\mathcal{P}_{f}(\mathfrak{L})}$ be an additive process. Then, for any $f\in
C_{0}\left( \mathbb{R}^{d},\mathbb{R}\right) $, there is a measurable subset
$\tilde{\Omega}\subset \Omega $ of full measure such that, for all $\tilde{%
\omega}\in \tilde{\Omega}$,
\begin{equation*}
\underset{l\rightarrow \infty }{\lim }\frac{1}{\left\vert \Lambda
_{l}\right\vert }\underset{x\in \Lambda _{l}}{\sum }\mathfrak{F}^{(\tilde{%
\omega})}\left( \left\{ x\right\} \right) f\left( l^{-1}x\right) =\mathbb{E}%
\left[ \mathfrak{F}^{(\omega )}\left( \left\{ 0\right\} \right) \right]
\int\nolimits_{\mathbb{R}^{d}}f\left( x\right) \mathrm{d}^{d}x\ .
\end{equation*}
\end{satz}

\begin{proof}
Since $f\in C_{0}\left( \mathbb{R}^{d},\mathbb{R}\right) $ has compact
support, $f$ is uniformly continuous. Assume w.l.o.g. that
\begin{equation*}
\mathrm{supp}(f)\subset \lbrack -1/2,1/2]^{d}\ .
\end{equation*}%
Then, there is a finite constant $D$ not depending on $j\in \mathcal{D}_{n}$%
, $t\in \mathbb{R}$, $k\in \{1,\ldots ,d\}$ and $x,y\in b_{j}$ such that%
\begin{equation}
\left\vert f\left( x\right) -f\left( y\right) \right\vert \leq Dn^{-1}\ .
\label{equicontinuous}
\end{equation}%
Using this and Theorem \ref{Ackoglu--Krengel ergodic theorem II} we obtain
the assertion.
\end{proof}

This last theorem could be extended to continuous functions $f\in C\left(
\mathbb{R}^{d},\mathbb{R}\right) $ vanishing sufficiently fast when $%
|l|\rightarrow \infty $ as well as for ergodic probability measures $%
\mathfrak{a}_{\Omega }$. This generalization is however not necessary here
and we refrain from proving it in detail.

\subsection{Diamagnetic Transport Coefficient and Density\label{Section
AC--Ohm's Law copy(1)}}

The aim of this section is to obtain the deterministic diamagnetic transport
coefficient $\mathbf{\Xi }_{\mathrm{d}}$ as well as the diamagnetic energy
density $\mathfrak{i}_{\mathrm{d}}$. See (\ref{def sigma_d}) and (\ref%
{diamagnetic energy density}) for their definitions. It is an simple
application of the ergodic theorems of Section \ref{Sec Akcoglu--Krengel}
and serves as a sort of ``warm up'' for the technically more involved case
of paramagnetic quantities.

We consider here the limit $l\rightarrow \infty $ of the current density (%
\ref{free current}) at equilibrium and the space--averaged diamagnetic
energy production coefficient $\Xi _{\mathrm{d},l}^{(\omega )}$ that is
defined by (\ref{average conductivity +1}). Indeed, as explained in Section %
\ref{Sect Classical Ohm's Law}, there exist, in general, currents coming
from the inhomogeneity of the fermion system for $\lambda \in \mathbb{R}^{+}$%
, even in absence of electromagnetic fields. We want to prove that, for
large samples, there are almost surely no currents within the fermion system
at thermal equilibrium. This result yields Assertion (th) of Theorem \ref%
{main 1 copy(8)}. We also would like to show that, as $l\rightarrow \infty $%
, $\Xi _{\mathrm{d},l}^{(\omega )}$ converges almost surely to the
diamagnetic transport coefficient $\mathbf{\Xi }_{\mathrm{d}}$, see Theorem %
\ref{thm charged transport coefficient} (d).


\begin{koro}[Currents and diamagnetic conductivity]
\label{main 1 copy(21)}\mbox{
}\newline
Let $\beta \in \mathbb{R}^{+}$ and $\lambda \in \mathbb{R}_{0}^{+}$. Then
one has: \newline
\emph{(th)}\ Current densities at thermal equilibrium: For any $z\in \mathbb{%
Z}^{d}$, there is a measurable subset $\tilde{\Omega}\left( z\right) \equiv
\tilde{\Omega}^{(\beta ,\lambda )}\left( z\right) \subset \Omega $ of full
measure such that, for all $\omega \in \tilde{\Omega}\left( z\right) $,%
\begin{equation*}
\underset{l\rightarrow \infty }{\lim }\frac{1}{\left\vert \Lambda
_{l}\right\vert }\underset{x\in \Lambda _{l}}{\sum }\varrho ^{(\beta ,\omega
,\lambda )}\left( I_{\left( x+z,x\right) }\right) =0\ .
\end{equation*}%
\emph{(d)}\ Diamagnetic charge transport coefficient: There is a measurable
subset $\tilde{\Omega}\equiv \tilde{\Omega}^{(\beta ,\lambda )}\subset
\Omega $ of full measure such that, for any $\tilde{\omega}\in \tilde{\Omega}
$,
\begin{equation}
\mathbf{\Xi }_{\mathrm{d}}:=\underset{l\rightarrow \infty }{\lim }\mathbb{E}%
\left[ \Xi _{\mathrm{d},l}^{(\omega )}\right] =\underset{l\rightarrow \infty
}{\lim }\Xi _{\mathrm{d},l}^{(\tilde{\omega})}\in \left[ -2,2\right] \ .
\notag
\end{equation}
\end{koro}

\begin{proof}
Let $\beta \in \mathbb{R}^{+}$, $\lambda \in \mathbb{R}_{0}^{+}$ and $z\in
\mathbb{Z}^{d}$. We define an additive process $\{\mathfrak{F}_{z}^{(\omega
)}\left( \Lambda \right) \}_{\Lambda \in \mathcal{P}_{f}(\mathfrak{L})}$ by
\begin{equation}
\mathfrak{F}_{z}^{(\omega )}\left( \Lambda \right) :=\sum\limits_{x\in
\Lambda }\varrho ^{(\beta ,\omega ,\lambda )}\left( a_{x+z}^{\ast
}a_{x}\right) =\sum\limits_{x\in \Lambda }\langle \mathfrak{e}%
_{x},F_{0}^{\beta }\left( \Delta _{\mathrm{d}}+\lambda V_{\omega }\right)
\mathfrak{e}_{x+z}\rangle  \label{first additive process}
\end{equation}%
for any finite subset $\Lambda \in \mathcal{P}_{f}(\mathfrak{L})$, see (\ref%
{Schwartz function F}) and Definition \ref{Additive process}. Similar to
Theorem \ref{decay bound theorem} (iv), the map $\omega \mapsto \mathfrak{F}%
_{z}^{(\omega )}\left( \Lambda \right) $ is bounded and measurable w.r.t.
the $\sigma $--algebra $\mathfrak{A}_{\Omega }$ for all $\Lambda \in
\mathcal{P}_{f}(\mathfrak{L})$. By the uniqueness of the KMS states $\varrho
^{(\beta ,\omega ,\lambda )}$, we moreover have
\begin{equation*}
\mathfrak{F}_{z}^{(\omega )}\left( z^{\prime }+\Lambda \right) =\mathfrak{F}%
_{z}^{(\omega )}\left( \Lambda \right)
\end{equation*}%
for all $z^{\prime }\in \mathbb{Z}^{d}$, $\Lambda \in \mathcal{P}_{f}(%
\mathfrak{L})$ and $\omega \in \Omega $. Clearly, $\{\Lambda _{l}\}_{l\in
\mathbb{R}^{+}}\subset \mathcal{P}_{f}(\mathfrak{L})$ is a regular sequence,
see Definition \ref{regular sequences}. Therefore, Theorem \ref%
{Ackoglu--Krengel ergodic theorem II copy(1)} implies the existence of a
measurable subset $\hat{\Omega}\left( z\right) \equiv \hat{\Omega}^{(\beta
,\lambda )}\left( z\right) \subset \Omega $ of full measure such that, for
all $\tilde{\omega}\in \tilde{\Omega}\left( z\right) $,
\begin{equation}
\underset{l\rightarrow \infty }{\lim }\left\{ \left\vert \Lambda
_{l}\right\vert ^{-1}\mathfrak{F}_{z}^{(\tilde{\omega})}\left( \Lambda
_{l}\right) \right\} =\mathbb{E}\left[ \varrho ^{(\beta ,\omega ,\lambda
)}\left( a_{z}^{\ast }a_{0}\right) \right] =\mathbb{E}\left[ \langle
\mathfrak{e}_{0},\mathbf{d}_{\mathrm{fermi}}^{(\beta ,\omega ,\lambda )}%
\mathfrak{e}_{z}\rangle \right] \ .  \label{limit easy ergodic}
\end{equation}%
Recall that $\mathbb{E}[\ \cdot \ ]$ is the expectation value associated
with the probability measure $\mathfrak{a}_{\Omega }$ (\ref{probability
measure}), $\mathbf{d}_{\mathrm{fermi}}^{(\beta ,\omega ,\lambda )}$ is the
symbol (\ref{Fermi statistic}), and $\left\{ \mathfrak{e}_{x}\right\} _{x\in
\mathfrak{L}}$ is the canonical orthonormal basis of $\ell ^{2}(\mathfrak{L}%
) $ with scalar product $\langle \cdot ,\cdot \rangle $.

By (\ref{current observable}), observe that%
\begin{equation*}
\frac{1}{\left\vert \Lambda _{l}\right\vert }\underset{x\in \Lambda _{l}}{%
\sum }\varrho ^{(\beta ,\omega ,\lambda )}\left( I_{\left( x+z,x\right)
}\right) =2\mathrm{Im}\left\{ \left\vert \Lambda _{l}\right\vert ^{-1}%
\mathfrak{F}_{z}^{(\omega )}\left( \Lambda _{l}\right) \right\} \ ,
\end{equation*}%
while, from the definition (\ref{average conductivity +1}),
\begin{equation*}
\left\{ \Xi _{\mathrm{d},l}^{(\omega )}\right\} _{k,q}=2\delta _{k,q}\mathrm{%
Re}\left\{ \left\vert \Lambda _{l}\right\vert ^{-1}\mathfrak{F}%
_{e_{k}}^{(\omega )}\left( \Lambda _{l}\right) \right\}
\end{equation*}%
for any $k,q\in \{1,\ldots ,d\}$. Combined with (\ref{def sigma_dbis}), (\ref%
{explicit form conductivity}) and (\ref{limit easy ergodic}), these two
equalities yield Assertions (th) and (d), respectively. Indeed, $V_{\omega }$
is an i.i.d. potential and $\mathbb{E}[I_{\left( z,0\right) }]=0$ for any $%
z\in \mathfrak{L}$.
\end{proof}

We study now the limit $(\eta ,l^{-1})\rightarrow (0,0)$ of the diamagnetic
energy $\mathfrak{I}_{\mathrm{d}}^{(\omega ,\eta \mathbf{A}_{l})}$ defined
by (\ref{lim_en_incr dia}). An asymptotic expansion of the diamagnetic
energy is given by \cite[Theorem 5.12]{OhmII} for small parameters $|\eta
|\ll 1$: For any $\mathbf{A}\in \mathbf{C}_{0}^{\infty }$, there is $\eta
_{0}\in \mathbb{R}^{+}$ such that, for all $|\eta |\in (0,\eta _{0}]$, $%
l,\beta \in \mathbb{R}^{+}$, $\omega \in \Omega $, $\lambda \in \mathbb{R}%
_{0}^{+}$ and $t\geq t_{0}$,%
\begin{eqnarray}
\mathfrak{I}_{\mathrm{d}}^{(\omega ,\eta \mathbf{A}_{l})}\left( t\right) &=&-%
\frac{\eta }{2}\underset{\mathbf{x}\in \mathfrak{K}}{\sum }\varrho ^{(\beta
,\omega ,\lambda )}(I_{\mathbf{x}})\int\nolimits_{t_{0}}^{t}\mathbf{E}_{s}^{%
\mathbf{A}_{l}}(\mathbf{x})\mathrm{d}s  \label{asymptotic expansion dia} \\
&&+\frac{\eta ^{2}}{2}\int\nolimits_{t_{0}}^{t}\mathrm{d}s_{1}%
\int_{t_{0}}^{s_{1}}\mathrm{d}s_{2}\underset{\mathbf{x}\in \mathfrak{K}}{%
\sum }\sigma _{\mathrm{d}}^{(\omega )}\left( \mathbf{x}\right) \mathbf{E}%
_{s_{2}}^{\mathbf{A}_{l}}(\mathbf{x})\mathbf{E}_{s_{1}}^{\mathbf{A}_{l}}(%
\mathbf{x})+\mathcal{O}(\eta ^{3}l^{d})\ .  \notag
\end{eqnarray}%
The correction terms of order $\mathcal{O}(\eta ^{3}l^{d})$ is uniformly
bounded in $\beta \in \mathbb{R}^{+}$, $\omega \in \Omega $, $\lambda \in
\mathbb{R}_{0}^{+}$ and $t\geq t_{0}$. Here,%
\begin{equation}
\mathbf{E}_{t}^{\mathbf{A}}\left( \mathbf{x}\right) \equiv \mathbf{E}_{t}^{%
\mathbf{A}}(x^{(1)},x^{(2)}):=\int\nolimits_{0}^{1}\left[ E_{\mathbf{A}%
}(t,\alpha x^{(2)}+(1-\alpha )x^{(1)})\right] (x^{(2)}-x^{(1)})\mathrm{d}%
\alpha \ ,  \label{V bar 0bis}
\end{equation}%
is the integrated electric field between $x^{(2)}\in \mathfrak{L}$ and $%
x^{(1)}\in \mathfrak{L}$ at time $t\in \mathbb{R}$ and%
\begin{equation}
\mathfrak{K}:=\left\{ \mathbf{x}:=(x^{(1)},x^{(2)})\in \mathfrak{L}^{2}\ :\
|x^{(1)}-x^{(2)}|=1\right\}  \label{proche voisins0}
\end{equation}%
is the set of bonds of nearest neighbors.

The asymptotic expansion (\ref{asymptotic expansion dia}) ensures the
existence of the limit%
\begin{equation}
\underset{\eta \rightarrow 0}{\lim }\left\{ \left( \eta \left\vert \Lambda
_{l}\right\vert \right) ^{-1}\mathfrak{I}_{\mathrm{d}}^{(\omega ,\eta
\mathbf{A}_{l})}\left( t\right) \right\} =\mathfrak{G}_{l}^{(\omega )}\left(
t\right) \ .  \label{G2-1}
\end{equation}%
Here, for any $l,\beta \in \mathbb{R}^{+}$, $\omega \in \Omega $, $\lambda
\in \mathbb{R}_{0}^{+}$, $\mathbf{A}\in \mathbf{C}_{0}^{\infty }$ and $t\in
\mathbb{R}$, the function%
\begin{equation}
\mathfrak{G}_{l}^{(\omega )}\left( t\right) \equiv \mathfrak{G}_{l}^{(\beta
,\omega ,\lambda ,\eta \mathbf{A}_{l})}\left( t\right) =-\frac{1}{2}%
\int\nolimits_{t_{0}}^{t}\frac{1}{\left\vert \Lambda _{l}\right\vert }%
\underset{\mathbf{x}\in \mathfrak{K}}{\sum }\varrho ^{(\beta ,\omega
,\lambda )}(I_{\mathbf{x}})\mathbf{E}_{s}^{\mathbf{A}_{l}}(\mathbf{x})%
\mathrm{d}s  \label{G1}
\end{equation}%
is the electric work density produced by thermal currents within the box $%
\Lambda _{l}$.

The limit $l\rightarrow \infty $ of the function $\mathfrak{G}_{l}^{(\omega
)}$ is a little bit more complicated than in the first two examples because
the electric field $\mathbf{E}_{t}^{\mathbf{A}_{l}}$ is
space--inhomogeneous. In fact, we can divide the (compact) support $\mathrm{%
supp}(\mathbf{A}(t,.))\subset \mathbb{R}^{d}$ of the vector potential $%
\mathbf{A}(t,.)$ at $t\in \mathbb{R}$ in small regions to combine the
piecewise--constant approximation of the smooth electric field $E_{\mathbf{A}%
}$ (\ref{V bar 0}) in (\ref{G1}) with Theorem \ref{Ackoglu--Krengel ergodic
theorem II}. A similar problem is already treated in Theorem \ref%
{Ackoglu--Krengel ergodic theorem III}. In fact, one gets the following
assertion:

\begin{satz}[Electric work density produced by thermal currents]
\label{main 1 copy(18)}\mbox{
}\newline
Let $\beta \in \mathbb{R}^{+}$, $\lambda \in \mathbb{R}_{0}^{+}$. Then,
there is a measurable subset $\tilde{\Omega}\equiv \tilde{\Omega}^{(\beta
,\lambda )}\subset \Omega $ of full measure such that, for all $\mathbf{A}%
\in \mathbf{C}_{0}^{\infty }$,
\begin{equation*}
\underset{l\rightarrow \infty }{\lim }\mathfrak{G}_{l}^{(\omega )}\left(
t\right) =0\ ,\qquad \omega \in \tilde{\Omega}\ ,
\end{equation*}%
uniformly for all $t\geq t_{0}$.
\end{satz}

\begin{proof}
We study the limit $l\rightarrow \infty $ of the function%
\begin{equation}
\mathbf{M}_{l}^{(\omega )}\left( t\right) :=\frac{1}{\left\vert \Lambda
_{l}\right\vert }\underset{\mathbf{x}\in \mathfrak{K}}{\sum }\varrho
^{(\beta ,\omega ,\lambda )}\left( I_{\mathbf{x}}\right) \mathbf{E}_{t}^{%
\mathbf{A}_{l}}(\mathbf{x})  \label{G2}
\end{equation}%
for any $l,\beta \in \mathbb{R}^{+}$, $\omega \in \Omega $, $\lambda \in
\mathbb{R}_{0}^{+}$, $\mathbf{A}\in \mathbf{C}_{0}^{\infty }$ and $t\in
\mathbb{R}$. Indeed, because $\mathbf{A}\in \mathbf{C}_{0}^{\infty }$, note
that
\begin{equation}
\Vert \mathbf{E}^{\mathbf{A}}\Vert _{\infty }:=\max \left\{ \left\vert E_{%
\mathbf{A}}(t,x)\right\vert \ :\ (t,x)\in \mathrm{supp}(A)\right\} \in
\mathbb{R}^{+}\ ,  \label{bound easy}
\end{equation}%
which implies that
\begin{equation}
\underset{t\in \mathbb{R}}{\sup }\left\vert \mathbf{M}_{l}^{(\omega )}\left(
t\right) \right\vert \leq \Vert \mathbf{E}^{\mathbf{A}}\Vert _{\infty }%
\underset{z\in \mathfrak{L},\left\vert z\right\vert =1}{\sum }\Vert
I_{\left( 0,z\right) }\Vert <\infty \ .  \label{G2+1}
\end{equation}%
Therefore, by (\ref{G1}) and Lebesgue's dominated convergence theorem, in
order to get the assertion it suffices to show that, for any fixed $t \in
[t_0,t_1]$ and $\omega$ in a subset $\tilde{\Omega}\equiv \tilde{\Omega}%
^{(\beta ,\lambda )}\subset \Omega $ of full measure, the l.h.s. of (\ref%
{G2+1}) vanishes when $l\rightarrow \infty $. This is done like in Theorem %
\ref{Ackoglu--Krengel ergodic theorem III}.

Indeed, assume w.l.o.g. that, for all $t\in \mathbb{R}$,
\begin{equation}
\mathrm{supp}(\mathbf{A}(t,.))\subset \lbrack -1/2,1/2]^{d}\ .
\label{toto support}
\end{equation}%
For any integer $n\in \mathbb{N}$, we divide the elementary box $%
[-1/2,1/2]^{d}$ in $n^{d}$ boxes $\{b_{j}\}_{j\in \mathcal{D}_{n}}$ of
side--length $1/n$, see (\ref{boxes b1})--(\ref{boxes b2}). For any $j\in
\mathcal{D}_{n}$, let $z^{(j)}\in b_{j}$ be any fixed point of the box $%
b_{j} $. Then, we consider piecewise--constant approximations of the
(smooth) electric field (\ref{V bar 0}), that is,%
\begin{equation}
E_{\mathbf{A}}(t,x):=-\partial _{t}\mathbf{A}(t,x)\ ,\qquad t\in \mathbb{R}%
,\ x\in \mathbb{R}^{d}\ ,  \label{V bar}
\end{equation}%
and define the approximated energy density%
\begin{equation}
\mathbf{\tilde{M}}_{l}^{(\omega )}\left( t\right) :=\frac{1}{\left\vert
\Lambda _{l}\right\vert }\underset{j\in \mathcal{D}_{n}}{\sum }\ \underset{%
\mathbf{x}\in \mathfrak{K}\cap (lb_{j})^{2}}{\sum }\varrho ^{(\beta ,\omega
,\lambda )}\left( I_{\mathbf{x}}\right) \left[ E_{\mathbf{A}}(t,z^{(j)})%
\right] (x^{(1)}-x^{(2)})  \label{G2bis}
\end{equation}%
for any $l,\beta \in \mathbb{R}^{+}$, $\omega \in \Omega $, $\lambda \in
\mathbb{R}_{0}^{+}$, $\mathbf{A}\in \mathbf{C}_{0}^{\infty }$, $t\in \mathbb{%
R}$ and $n\in \mathbb{N}$.

We infer from (\ref{rescaled vector potential}), (\ref{V bar 0bis}) and (\ref%
{V bar}) that, for any $l\in \mathbb{R}^{+}$, $\mathbf{A}\in \mathbf{C}%
_{0}^{\infty }$, $j\in \mathcal{D}_{n}$, $t\in \mathbb{R}$, $k\in \{1,\ldots
,d\}$ and $x\in lb_{j}$,%
\begin{eqnarray*}
&&\left\vert \mathbf{E}_{t}^{\mathbf{A}_{l}}(x,x\pm e_{k})-\left[ E_{\mathbf{%
A}}(t,z^{(j)})\right] (\pm e_{k})\right\vert \\
&\leq &\int\nolimits_{0}^{1}\left\vert \left[ \partial _{t}\mathbf{A}%
(t,z^{(j)})\right] (e_{k})-\left[ \partial _{t}\mathbf{A}_{l}(t,x\pm
(1-\alpha )e_{k})\right] (e_{k})\right\vert \mathrm{d}\alpha \\
&\leq &\underset{y\in \tilde{b}_{j,l}}{\sup }\left\vert \left[ \partial _{t}%
\mathbf{A}(t,z^{(j)})\right] (e_{k})-\left[ \partial _{t}\mathbf{A}(t,y)%
\right] (e_{k})\right\vert <\infty \ ,
\end{eqnarray*}%
where
\begin{equation*}
\tilde{b}_{j,l}:=\left\{ x\in \mathbb{R}^{d}\ :\ \underset{y\in b_{j}}{\min }%
\left\vert x-y\right\vert \leq l^{-1}\right\} \ .
\end{equation*}%
In particular, since $\mathbf{A}\in \mathbf{C}_{0}^{\infty }$, there is a
finite constant $D\in \mathbb{R}^{+}$ not depending on $j\in \mathcal{D}_{n}$%
, $t\in \mathbb{R}$, $k\in \{1,\ldots ,d\}$ and $x\in b_{j}$ such that%
\begin{equation}
\left\vert \mathbf{E}_{t}^{\mathbf{A}_{l}}(x,x\pm e_{k})-\left[ E_{\mathbf{A}%
}(t,z^{(j)})\right] (\pm e_{k})\right\vert \leq D(n^{-1}+l^{-1})\ .
\label{inequality G cool1}
\end{equation}%
This upper bound is the analogue of (\ref{equicontinuous}) in the proof of
Theorem \ref{Ackoglu--Krengel ergodic theorem III}. Using also (\ref{toto
support}), it follows that
\begin{equation}
\left\vert \mathbf{M}_{l}^{(\omega )}\left( t\right) -\mathbf{\tilde{M}}%
_{l}^{(\omega )}\left( t\right) \right\vert \leq D(n^{-1}+l^{-1})\underset{%
z\in \mathfrak{L},\left\vert z\right\vert =1}{\sum }\Vert I_{\left(
z,0\right) }\Vert \ .  \label{inequality G cool2}
\end{equation}%
Therefore, by (\ref{proche voisins0}) and (\ref{G2bis}), for any $z\in
\mathbb{Z}^{d}$ such that $\left\vert z\right\vert =1$, it suffices to
compute the limit
\begin{equation*}
\underset{l\rightarrow \infty }{\lim }\frac{1}{\left\vert \mathfrak{L}\cap
(lb_{j})\right\vert }\underset{x\in \mathfrak{L}\cap (lb_{j})}{\sum }\varrho
^{(\beta ,\omega ,\lambda )}\left( I_{\left( x+z,x\right) }\right)
\end{equation*}%
like in Theorem \ref{Ackoglu--Krengel ergodic theorem III}. For any $\beta
\in \mathbb{R}^{+}$, $\lambda \in \mathbb{R}_{0}^{+}$ and $z\in \mathbb{Z}%
^{d}$, we invoke Theorem \ref{Ackoglu--Krengel ergodic theorem II} to get
the existence of a measurable subset $\hat{\Omega}_{z}\equiv \hat{\Omega}%
_{z}^{(\beta ,\lambda )}\subset \Omega $ of full measure such that, for any $%
\tilde{\omega}\in \hat{\Omega}_{z}$,
\begin{equation}
\underset{l\rightarrow \infty }{\lim }\left\{ \frac{1}{\left\vert \mathfrak{L%
}\cap (lb_{j})\right\vert }\underset{x\in \mathfrak{L}\cap (lb_{j})}{\sum }%
\varrho ^{(\beta ,\tilde{\omega},\lambda )}\left( I_{\left( x+z,x\right)
}\right) \right\} =\mathbb{E}[\varrho ^{(\beta ,\omega ,\lambda )}\left(
I_{\left( z,0\right) }\right) ]=0\ .  \label{birkoofnew1}
\end{equation}%
Note that the last equality is a consequence of the identity $%
I_{(z,0)}=-I_{(0,z)}$ and the translation and reflection invariance of the
probability measure $\mathbf{\mathfrak{a}}_{\Omega }$. Meanwhile, the
measurable subset defined by%
\begin{equation*}
\tilde{\Omega}\equiv \tilde{\Omega}^{(\beta ,\lambda )}:=\underset{z\in
\mathfrak{L},\left\vert z\right\vert =1}{\bigcap }\hat{\Omega}_{z}\subset
\Omega
\end{equation*}%
has full measure and we obtain from (\ref{inequality G cool2})--(\ref%
{birkoofnew1}) that, for any $\omega \in \tilde{\Omega}$,
\begin{equation*}
\underset{l\rightarrow \infty }{\lim }\ \underset{t\in \mathbb{R}}{\sup }%
\left\vert \mathbf{M}_{l}^{(\omega )}\left( t\right) \right\vert =0\ .
\end{equation*}
\end{proof}

It remains to study the diamagnetic energy density $\mathfrak{i}_{\mathrm{d}%
} $ defined by (\ref{diamagnetic energy density}), that is,
\begin{equation*}
\mathfrak{i}_{\mathrm{d}}\left( t\right) :=\underset{\eta \rightarrow 0}{%
\lim }\ \underset{l\rightarrow \infty }{\lim }\left\{ \left( \eta
^{2}\left\vert \Lambda _{l}\right\vert \right) ^{-1}\mathfrak{I}_{\mathrm{d}%
}^{(\omega ,\eta \mathbf{A}_{l})}\left( t\right) \right\}
\end{equation*}%
for $\beta \in \mathbb{R}^{+}$, $\omega \in \Omega $, $\lambda \in \mathbb{R}%
_{0}^{+}$, $\mathbf{A}\in \mathbf{C}_{0}^{\infty }$ and $t\geq t_{0}$.
Thanks to the asymptotic expansion (\ref{asymptotic expansion dia}), its
derivation is done like in the proof of Theorem\ \ref{main 1 copy(18)} by
replacing current observables $I_{\mathbf{x}}$ (\ref{current observable})
and the integrated electric field $\mathbf{E}_{t}^{\mathbf{A}_{l}}(\mathbf{x}%
)$ in Equation (\ref{G2}) with fermion fields $P_{\mathbf{x}}$ (\ref{R x})
and products $\mathbf{E}_{s_{2}}^{\mathbf{A}_{l}}(\mathbf{x})\mathbf{E}%
_{s_{1}}^{\mathbf{A}_{l}}(\mathbf{x})$. Then, one gets the diamagnetic
energy density $\mathfrak{i}_{\mathrm{d}}$ as stated in Theorem \ref{main 1}
(d).

\begin{satz}[Diamagnetic energy density]
\label{main 1 copy(14)}\mbox{
}\newline
Let $\beta \in \mathbb{R}^{+}$ and $\lambda \in \mathbb{R}_{0}^{+}$. Then,
there is a measurable subset $\tilde{\Omega}\equiv \tilde{\Omega}^{(\beta
,\lambda )}\subset \Omega $ of full measure such that, for any $\omega \in
\tilde{\Omega}$ and $\mathbf{A}\in \mathbf{C}_{0}^{\infty }$,
\begin{equation*}
\mathfrak{i}_{\mathrm{d}}\left( t\right) =\int\nolimits_{\mathbb{R}^{d}}%
\mathrm{d}^{d}x\int\nolimits_{t_{0}}^{t}\mathrm{d}s_{1}\int%
\nolimits_{t_{0}}^{s_{1}}\mathrm{d}s_{2}\ \left\langle E_{\mathbf{A}%
}(s_{1},x),\mathbf{\Xi }_{\mathrm{d}}E_{\mathbf{A}}(s_{2},x)\right\rangle
\end{equation*}%
uniformly for all $t\geq t_{0}$ in compact sets. Here, $\mathbf{\Xi }_{%
\mathrm{d}}$ is defined by (\ref{def sigma_d}).
\end{satz}

\begin{proof}
Let $\beta \in \mathbb{R}^{+}$ and $\lambda \in \mathbb{R}_{0}^{+}$. By (\ref%
{backwards -1bispara}) and (\ref{asymptotic expansion dia}), for any $%
\mathbf{A}\in \mathbf{C}_{0}^{\infty }$, there is $\eta _{0}\in \mathbb{R}%
^{+}$ such that, for all $|\eta |\in (0,\eta _{0}]$, $l,\beta \in \mathbb{R}%
^{+}$, $\omega \in \Omega $, $\lambda \in \mathbb{R}_{0}^{+}$ and $t\geq
t_{0}$,%
\begin{equation}
\mathfrak{I}_{\mathrm{d}}^{(\omega ,\mathbf{A})}\left( t\right) -\eta
\left\vert \Lambda _{l}\right\vert \mathfrak{G}_{l}^{(\omega )}\left(
t\right) =\eta ^{2}\left\vert \Lambda _{l}\right\vert
\int\nolimits_{t_{0}}^{t}\mathrm{d}s_{1}\int_{t_{0}}^{s_{1}}\mathrm{d}s_{2}%
\mathbf{\tilde{X}}_{l}^{(\omega )}\left( s_{1},s_{2}\right) +\mathcal{O}%
(\eta ^{3}l^{d})  \label{dia0}
\end{equation}%
with the energy density $\mathfrak{G}_{l}^{(\omega )}$ defined by (\ref{G1})
while%
\begin{equation}
\mathbf{\tilde{X}}_{l}^{(\omega )}\left( s_{1},s_{2}\right) :=\frac{1}{%
2\left\vert \Lambda _{l}\right\vert }\underset{\mathbf{x}\in \mathfrak{K}}{%
\sum }\varrho ^{(\beta ,\omega ,\lambda )}\left( P_{\mathbf{x}}\right)
\mathbf{E}_{s_{1}}^{\mathbf{A}_{l}}(\mathbf{x})\mathbf{E}_{s_{2}}^{\mathbf{A}%
_{l}}(\mathbf{x})  \label{dia1}
\end{equation}%
for any $s_{1},s_{2}\in \mathbb{R}$. The correction term of order $\mathcal{O%
}(\eta ^{3}l^{d})$ is uniformly bounded in $\beta \in \mathbb{R}^{+}$, $%
\omega \in \Omega $, $\lambda \in \mathbb{R}_{0}^{+}$ and $t\geq t_{0}$.

Similar to the proof of Theorem \ref{main 1 copy(18)} we use
piecewise--constant approximations of the (smooth) electric field and
Theorem \ref{Ackoglu--Krengel ergodic theorem II} together with (\ref{bound
easy}), (\ref{inequality G cool1}) and Lebesgue's dominated convergence
theorem to compute the limit $l\rightarrow \infty $ of the r.h.s. of (\ref%
{dia0}) (without the factor $\eta ^{2}l^{d}$). More precisely, one finds the
existence a measurable subset $\tilde{\Omega}\equiv \tilde{\Omega}^{(\beta
,\lambda )}\subset \Omega $ of full measure such that, for all $\omega \in
\tilde{\Omega}$,%
\begin{eqnarray}
&&\underset{l\rightarrow \infty }{\lim }\left\{ \int\nolimits_{t_{0}}^{t}%
\mathrm{d}s_{1}\int_{t_{0}}^{s_{1}}\mathrm{d}s_{2}\ \mathbf{\tilde{X}}%
_{l}^{(\omega )}\left( s_{1},s_{2}\right) \right\}  \label{birkoff easy} \\
&=&\int\nolimits_{t_{0}}^{t}\mathrm{d}s_{1}\int\nolimits_{t_{0}}^{s_{1}}%
\mathrm{d}s_{2}\int\nolimits_{\mathbb{R}^{d}}\mathrm{d}^{d}x\ \left\langle
E_{\mathbf{A}}(s_{1},x),\mathbf{\Xi }_{\mathrm{d}}E_{\mathbf{A}%
}(s_{2},x)\right\rangle  \notag
\end{eqnarray}%
uniformly for all $t\geq t_{0}$. The assertion follows from (\ref{dia0}) and
(\ref{birkoff easy}) together with Theorem \ref{main 1 copy(18)} and
Fubini's theorem.
\end{proof}

\subsection{Paramagnetic Energy Density\label{Section AC--Ohm's Law}}


The aim of this section is to prove the existence of the paramagnetic energy
density $\mathfrak{i}_{\mathrm{p}}$ defined by (\ref{paramagnetic energy
density}), that is,
\begin{equation}
\mathfrak{i}_{\mathrm{p}}\left( t\right) =\underset{(\eta
,l^{-1})\rightarrow (0,0)}{\lim }\left\{ \left( \eta ^{2}\left\vert \Lambda
_{l}\right\vert \right) ^{-1}\mathfrak{I}_{\mathrm{p}}^{(\omega ,\eta
\mathbf{A}_{l})}\left( t\right) \right\}
\label{paramagnetic energy densitybis}
\end{equation}%
for $\beta \in \mathbb{R}^{+}$, $\omega \in \Omega $, $\lambda \in \mathbb{R}%
_{0}^{+}$, $\mathbf{A}\in \mathbf{C}_{0}^{\infty }$ and $t\geq t_{0}$. Our
proof requires similar arguments to those proving Theorems \ref{main 1
copy(18)}--\ref{main 1 copy(14)}:

\begin{itemize}
\item The asymptotic expansion given by \cite[Theorem 5.12]{OhmII} for the
paramagnetic energy increment.

\item We divide the (compact) support $\mathrm{supp}(\mathbf{A}(t,.))\subset
\mathbb{R}^{d}$ of the vector potential $\mathbf{A}(t,.)$ at $t\in \mathbb{R}
$ in small regions to use the piecewise--constant approximation of the
smooth electric field $E_{\mathbf{A}}$.

\item Theorem \ref{Ackoglu--Krengel ergodic theorem II} and the fact that
any countable intersection of measurable sets of full measure has full
measure.

\item Lebesgue's dominated convergence theorem.
\end{itemize}

The proof for the paramagnetic case is, however, technically more involved
than those of Section \ref{Section AC--Ohm's Law copy(1)}. Indeed, to use
\cite[Lemma 5.2, Theorem 5.12]{OhmII}, we additionally need some (space)
decay of complex--time two--point correlation functions. To this end, we
invoke Theorem \ref{decay bound theorem}. The application of the latter
requires some technical preparation and we present the corresponding
additional arguments in various lemmata which then yield a proposition and a
few corollaries and theorems. This rather technical study ends with Theorem %
\ref{main 1 copy(2)}, which serves as a springboard to obtain Theorem \ref%
{main 1}.

First, by \cite[Lemma 5.2, Theorem 5.12]{OhmII}, for any $\mathbf{A}\in
\mathbf{C}_{0}^{\infty }$, there is $\eta _{0}\in \mathbb{R}^{+}$ such that,
for all $|\eta |\in (0,\eta _{0}]$, $l,\beta \in \mathbb{R}^{+}$, $\omega
\in \Omega $, $\lambda \in \mathbb{R}_{0}^{+}$ and $t\geq t_{0}$,%
\begin{equation}
\mathfrak{I}_{\mathrm{p}}^{(\omega ,\eta \mathbf{A}_{l})}\left( t\right)
=\eta ^{2}\left\vert \Lambda _{l}\right\vert \int\nolimits_{t_{0}}^{t}%
\mathrm{d}s_{1}\int\nolimits_{t_{0}}^{s_{1}}\mathrm{d}s_{2}\ \mathbf{X}%
_{l,0}^{(\omega )}(s_{1},s_{2})+\mathcal{O}(\eta ^{3}l^{d})\ .
\label{Lemma non existing}
\end{equation}%
The correction term of order $\mathcal{O}(\eta ^{3}l^{d})$ is uniformly
bounded in $\beta \in \mathbb{R}^{+}$, $\omega \in \Omega $, $\lambda \in
\mathbb{R}_{0}^{+}$ and $t\geq t_{0}$. Here, $\mathbf{X}_{l,\upsilon
}^{(\omega )}$ is defined, for any $\upsilon \in \lbrack 0,\beta /2)$ and $%
s_{1},s_{2}\in \mathbb{R}$, by
\begin{eqnarray}
\mathbf{X}_{l,\upsilon }^{(\omega )}(s_{1},s_{2}) &:=&\frac{1}{4\left\vert
\Lambda _{l}\right\vert }\underset{\mathbf{x},\mathbf{y}\in \mathfrak{K}}{\sum }\int\nolimits_{\upsilon }^{\beta -\upsilon }\mathrm{d}\alpha \left(
\mathfrak{C}_{s_{1}-s_{2}+i\alpha }^{(\omega )}(\mathbf{x},\mathbf{y})-\mathfrak{C}_{i\alpha }^{(\omega )}(\mathbf{x},\mathbf{y})\right)   \notag \\
&&\qquad \qquad \qquad \qquad \qquad \qquad \times \mathbf{E}_{s_{1}}^{\mathbf{A}_{l}}\left( \mathbf{x}\right) \mathbf{E}_{s_{2}}^{\mathbf{A}_{l}}(\mathbf{y})\ .  \label{Lemma non existing2}
\end{eqnarray}%
$\mathfrak{C}_{t+i\alpha }^{(\omega )}$ is the map from $\mathfrak{L}^{4}$
to $\mathbb{C}$ defined by%
\begin{equation}
\mathfrak{C}_{t+i\alpha }^{(\omega )}(\mathbf{x},\mathbf{y}):=\underset{\pi
,\pi ^{\prime }\in S_{2}}{\sum }\varepsilon _{\pi }\varepsilon _{\pi
^{\prime }}C_{t+i\alpha }^{(\omega )}(y^{\pi ^{\prime }(1)},x^{\pi
(1)})C_{-t+i(\beta -\alpha )}^{(\omega )}(x^{\pi (2)},y^{\pi ^{\prime }(2)})
\label{map cool}
\end{equation}%
for any $\mathbf{x}:=(x^{(1)},x^{(2)})\in \mathfrak{L}^{2}$ and $\mathbf{y}%
:=(y^{(1)},y^{(2)})\in \mathfrak{L}^{2}$, where $\pi ,\pi ^{\prime }\in
S_{2} $ are by definition permutations of $\{1,2\}$ with signatures $%
\varepsilon _{\pi },\varepsilon _{\pi ^{\prime }}\in \{-1,1\}$. The
definition of the set $\mathfrak{K}\subset \mathfrak{L}^{2}$ of bonds of
nearest neighbors is given by (\ref{proche voisins0}). Note also that the
integral in (\ref{Lemma non existing2}) can be exchanged with the (finite)
sum because $\mathbf{A}\in \mathbf{C}_{0}^{\infty }$.

The first important result of the present subsection will be a proof that $%
\mathbf{X}_{l,0}^{(\omega )}$ almost surely converges to a deterministic
function, as $l\rightarrow \infty $. See Corollary \ref{lemma conductivty4
copy(6)}. Then, we will use Lebesgue's dominated convergence theorem to get
the paramagnetic energy increment $\mathfrak{I}_{\mathrm{p}}^{(\omega ,\eta
\mathbf{A}_{l})}$ in the limit $(\eta ,l^{-1})\rightarrow (0,0)$, see
Theorem \ref{main 1 copy(2)}.

By Theorem \ref{decay bound theorem}, note that, for all $\varepsilon ,\beta
\in \mathbb{R}^{+}$, $\lambda \in \mathbb{R}_{0}^{+}$, $t\in \mathbb{R}$, $%
\upsilon \in (0,\beta /2)$ and $\alpha \in \lbrack \upsilon ,\beta -\upsilon
]$, the complex--time two--point correlation function $C_{t+i\alpha
}^{(\omega )}$ can be written as the sum%
\begin{equation}
C_{t+i\alpha }^{(\omega )}\left( \mathbf{x}\right) =A_{t+i\alpha ,\upsilon
,\varepsilon }^{(\omega )}\left( \mathbf{x}\right) +B_{t+i\alpha ,\upsilon
,\varepsilon }^{(\omega )}\left( \mathbf{x}\right) \ ,\qquad \mathbf{x}%
:=(x^{(1)},x^{(2)})\in \mathfrak{L}^{2}\ ,  \label{inequality cool ohm}
\end{equation}%
of two maps $A_{t+i\alpha ,\upsilon ,\varepsilon }^{(\omega )},B_{t+i\alpha
,\upsilon ,\varepsilon }^{(\omega )}$ from $\mathfrak{L}^{2}$ to $\mathbb{C}$%
. This decomposition has the following useful property: $A_{t+i\alpha
,\upsilon ,\varepsilon }^{(\omega )}$\ can be seen as the kernel (w.r.t. the
canonical basis $\{\mathfrak{e}_{x}\}_{x\in \mathfrak{L}}$) of an operator,
again denoted by $A_{t+i\alpha ,\upsilon ,\varepsilon }^{(\omega )}\in
\mathcal{B}(\ell ^{2}(\mathfrak{L}))$, with arbitrarily small operator norm $%
\Vert A_{t+i\alpha ,\upsilon ,\varepsilon }^{(\omega )}\Vert _{\mathrm{op}%
}\leq \varepsilon $, whereas $B_{t+i\alpha ,\upsilon ,\varepsilon }^{(\omega
)}\left( \mathbf{x}\right) $ rapidly decays, as $|x^{(1)}-x^{(2)}|%
\rightarrow \infty $. This is however only satisfied if $\alpha \in \lbrack
\upsilon ,\beta -\upsilon ]$ with fixed $\upsilon \in (0,\beta /2)$, see
Theorem \ref{decay bound theorem}.

As a consequence, the first step is to approximate $\mathbf{X}%
_{l,0}^{(\omega )}$ with $\mathbf{X}_{l,\upsilon }^{(\omega )}$ for
arbitrarily small parameters $\upsilon >0$:

\begin{lemma}[Approximation I]
\label{lemma conductivty0}\mbox{
}\newline
Let $\mathbf{A}\in \mathbf{C}_{0}^{\infty }$. Then,
\begin{equation*}
\mathbf{X}_{l,0}^{(\omega )}\left( s_{1},s_{2}\right) =\mathbf{X}%
_{l,\upsilon }^{(\omega )}\left( s_{1},s_{2}\right) +\mathcal{O}(\upsilon )\
,
\end{equation*}%
uniformly for $l,\beta \in \mathbb{R}^{+}$, $\omega \in \Omega $, $\lambda
\in \mathbb{R}_{0}^{+}$ and $s_{1},s_{2}\in \mathbb{R}$.
\end{lemma}

\begin{proof}
The canonical orthonormal basis of $\ell ^{2}(\mathfrak{L})\otimes \ell ^{2}(%
\mathfrak{L})$ is defined by $\{\mathbf{e}_{\mathbf{x}}\}_{\mathbf{x}\in
\mathfrak{L}^{2}}$ with%
\begin{equation}
\mathbf{e}_{\mathbf{x}}:=\mathfrak{e}_{x^{(1)}}\otimes \mathfrak{e}%
_{x^{(2)}}\ ,\qquad \mathbf{x}:=(x^{(1)},x^{(2)})\in \mathfrak{L}^{2}\ .
\label{canonical orthonormal basis l4}
\end{equation}%
Recall that $\mathfrak{e}_{x}(y)\equiv \delta _{x,y}\in \ell ^{2}(\mathfrak{L%
})$. The coefficient $\mathfrak{C}_{t+i\alpha }^{(\omega )}$ defined by (\ref%
{map cool}) can be seen as the kernel (w.r.t. $\{\mathbf{e}_{\mathbf{x}}\}_{%
\mathbf{x}\in \mathfrak{L}^{2}}$) of a bounded operator on $\ell ^{2}(%
\mathfrak{L})\otimes \ell ^{2}(\mathfrak{L})$ that is again denoted by $%
\mathfrak{C}_{t+i\alpha }^{(\omega )}$. In particular, similar to (\ref%
{definition de B}),
\begin{equation}
\underset{\mathbf{x},\mathbf{y}\in \mathfrak{K}}{\sum }\mathfrak{C}%
_{s_{1}-s_{2}+i\alpha }^{(\omega )}(\mathbf{x},\mathbf{y})\mathbf{E}%
_{s_{1}}^{\mathbf{A}_{l}}\left( \mathbf{x}\right) \mathbf{E}_{s_{2}}^{%
\mathbf{A}_{l}}(\mathbf{y})=\underset{\mathbf{x},\mathbf{y}\in \mathfrak{K}}{%
\sum }\left\langle \mathbf{e}_{\mathbf{y}},\mathfrak{C}_{s_{1}-s_{2}+i\alpha
}^{(\omega )}\mathbf{e}_{\mathbf{x}}\right\rangle \mathbf{E}_{s_{1}}^{%
\mathbf{A}_{l}}\left( \mathbf{x}\right) \mathbf{E}_{s_{2}}^{\mathbf{A}_{l}}(%
\mathbf{y})\ .  \label{equlaity debile}
\end{equation}%
In particular, via \cite[Lemma 5.3]{OhmII}, i.e., $\Vert \mathfrak{C}%
_{t+i\alpha }^{(\omega )}\Vert _{\mathrm{op}}\leq 4$, and Equation (\ref%
{bound easy}) we arrive at the upper bound
\begin{equation}
\left\vert \frac{1}{4\left\vert \Lambda _{l}\right\vert }\underset{\mathbf{x}%
,\mathbf{y}\in \mathfrak{K}}{\sum }\mathfrak{C}_{s_{1}-s_{2}+i\alpha
}^{(\omega )}(\mathbf{x},\mathbf{y})\mathbf{E}_{s_{1}}^{\mathbf{A}%
_{l}}\left( \mathbf{x}\right) \mathbf{E}_{s_{2}}^{\mathbf{A}_{l}}(\mathbf{y}%
)\right\vert \leq 2d\Vert \mathbf{E}^{\mathbf{A}}\Vert _{\infty }^{2}%
\underset{t\in \mathbb{R}}{\max }\left\vert \mathrm{supp}(\mathbf{A}%
(t,.))\right\vert  \label{Lemma non existing2bis}
\end{equation}%
for any $\mathbf{A}\in \mathbf{C}_{0}^{\infty }$, $l,\beta \in \mathbb{R}%
^{+} $, $\omega \in \Omega $, $\lambda \in \mathbb{R}_{0}^{+}$, $\alpha \in
\lbrack 0,\beta ]$ and $s_{1},s_{2}\in \mathbb{R}$. Therefore, the assertion
follows from (\ref{Lemma non existing2}) combined with (\ref{Lemma non
existing2bis}).
\end{proof}

Because of (\ref{inequality cool ohm}) and Theorem \ref{decay bound theorem}%
, it is natural to define, at any $\varepsilon ,\beta \in \mathbb{R}^{+}$, $%
t\in \mathbb{R}$, $\upsilon \in (0,\beta /2)$ and $\alpha \in \lbrack
\upsilon ,\beta -\upsilon ]$, the map $\mathfrak{B}_{t+i\alpha ,\upsilon
,\varepsilon }^{(\omega )}$ from $\mathfrak{L}^{4}$ to $\mathbb{C}$ by%
\begin{equation}
\mathfrak{B}_{t+i\alpha ,\upsilon ,\varepsilon }^{(\omega )}(\mathbf{x},%
\mathbf{y}):=\underset{\pi ,\pi ^{\prime }\in S_{2}}{\sum }\varepsilon _{\pi
}\varepsilon _{\pi ^{\prime }}B_{t+i\alpha ,\upsilon ,\varepsilon }^{(\omega
)}(y^{\pi ^{\prime }(1)},x^{\pi (1)})B_{-t+i(\beta -\alpha ),\upsilon
,\varepsilon }^{(\omega )}(x^{\pi (2)},y^{\pi ^{\prime }(2)})
\label{B fract}
\end{equation}%
for any $\mathbf{x}:=(x^{(1)},x^{(2)})\in \mathfrak{L}^{2}$ and $\mathbf{y}%
:=(y^{(1)},y^{(2)})\in \mathfrak{L}^{2}$. In other words, this map is
defined by replacing in (\ref{map cool}) the complex--time two--point
correlation function $C_{t+i\alpha }^{(\omega )}$ by its approximation $%
B_{t+i\alpha ,\upsilon ,\varepsilon }^{(\omega )}$, which comes from the
decomposition (\ref{inequality cool ohm}). Similarly, for $s_{1},s_{2}\in
\mathbb{R}$, let%
\begin{eqnarray}
\mathbf{Y}_{l,\upsilon ,\varepsilon ,0}^{(\omega )}(s_{1},s_{2}) &:=&\frac{1}{4\left\vert \Lambda _{l}\right\vert }\underset{\mathbf{x},\mathbf{y}\in
\mathfrak{K}}{\sum }\int\nolimits_{\upsilon }^{\beta -\upsilon }\mathrm{d}\alpha \left( \mathfrak{B}_{s_{1}-s_{2}+i\alpha ,\upsilon ,\varepsilon
}^{(\omega )}(\mathbf{x},\mathbf{y})-\mathfrak{B}_{i\alpha ,\upsilon
,\varepsilon }^{(\omega )}(\mathbf{x},\mathbf{y})\right)   \notag \\
&&\qquad \qquad \qquad \qquad \qquad \qquad \times \mathbf{E}_{s_{1}}^{\mathbf{A}_{l}}\left( \mathbf{x}\right) \mathbf{E}_{s_{2}}^{\mathbf{A}_{l}}(\mathbf{y})\ .  \label{B fractbis}
\end{eqnarray}%
We show in the next lemma that it is a good approximation of (\ref{Lemma non
existing2}), provided $\upsilon >0$.

\begin{lemma}[Approximation II]
\label{lemma conductivty1 copy(1)}\mbox{
}\newline
Let $\varepsilon ,\beta \in \mathbb{R}^{+}$, $\mathbf{A}\in \mathbf{C}%
_{0}^{\infty }$ and $\upsilon \in (0,\beta /2)$. Then,
\begin{equation*}
\mathbf{X}_{l,\upsilon }^{(\omega )}\left( s_{1},s_{2}\right) =\mathbf{Y}%
_{l,\upsilon ,\varepsilon ,0}^{(\omega )}(s_{1},s_{2})+\mathcal{O}%
(\varepsilon )\ ,
\end{equation*}%
uniformly for $l\in \mathbb{R}^{+}$, $\omega \in \Omega $, $\lambda \in
\mathbb{R}_{0}^{+}$ and $s_{1},s_{2}\in \mathbb{R}$.
\end{lemma}

\begin{proof}
Let $\varepsilon ,\beta \in \mathbb{R}^{+}$, $\omega \in \Omega $, $\lambda
\in \mathbb{R}_{0}^{+}$, $\mathbf{A}\in \mathbf{C}_{0}^{\infty }$ and $%
\upsilon \in (0,\beta /2)$. By Theorem \ref{decay bound theorem} (i) and (%
\ref{definition de B}), $A_{t+i\alpha ,\upsilon ,\varepsilon }^{(\omega
)},B_{t+i\alpha ,\upsilon ,\varepsilon }^{(\omega )}$ can be seen as the
kernels (w.r.t. $\{\mathfrak{e}_{x}\}_{x\in \mathfrak{L}}$) of two bounded
operators on $\ell ^{2}(\mathfrak{L})$. Therefore, Theorem \ref{decay bound
theorem} (i) and the Cauchy--Schwarz inequality yield the existence of a
finite constant $D\in \mathbb{R}^{+}$ depending on $\beta ,\upsilon $ but
not on $\varepsilon \in \mathbb{R}^{+}$, $\omega \in \Omega $, $\lambda \in
\mathbb{R}_{0}^{+}$, $\alpha \in \lbrack \upsilon ,\beta -\upsilon ]$ and $%
t\in \mathbb{R}$ such that, for all $c_{x},c_{y}^{\prime }\in \mathbb{C}$, $%
x,y\in \mathfrak{L}$,
\begin{eqnarray*}
\left\vert \sum\nolimits_{x,y\in \mathfrak{L}}\overline{c_{x}}c_{y}^{\prime
}A_{t+i\alpha ,\upsilon ,\varepsilon }^{(\omega )}(x,y)\right\vert &\leq
&\left\Vert A_{t+i\alpha ,\upsilon ,\varepsilon }^{(\omega )}\right\Vert _{%
\mathrm{op}}\sqrt{\sum\nolimits_{x,y\in \mathfrak{L}}\left\vert
c_{x}\right\vert ^{2}\left\vert c_{y}^{\prime }\right\vert ^{2}} \\
&\leq &\varepsilon \ \sqrt{\sum\nolimits_{x,y\in \mathfrak{L}}\left\vert
c_{x}\right\vert ^{2}\left\vert c_{y}^{\prime }\right\vert ^{2}}, \\
\left\vert \sum\nolimits_{x,y\in \mathfrak{L}}\overline{c_{x}}c_{y}^{\prime
}B_{t+i\alpha ,\upsilon ,\varepsilon }^{(\omega )}(x,y)\right\vert &\leq
&\left\Vert B_{t+i\alpha ,\upsilon ,\varepsilon }^{(\omega )}\right\Vert _{%
\mathrm{op}}\sqrt{\sum\nolimits_{x,y\in \mathfrak{L}}\left\vert
c_{x}\right\vert ^{2}\left\vert c_{y}^{\prime }\right\vert ^{2}} \\
&\leq &D\sqrt{\sum\nolimits_{x,y\in \mathfrak{L}}\left\vert c_{x}\right\vert
^{2}\left\vert c_{y}^{\prime }\right\vert ^{2}}.
\end{eqnarray*}%
It obviously follows that, for all $c_{\mathbf{x}}\in \mathbb{C}$, $\mathbf{x%
}\in \mathfrak{K}$, and some similar constant $D\in \mathbb{R}^{+}$,%
\begin{eqnarray*}
\left\vert \underset{\mathbf{x,y}\in \mathfrak{K}}{\sum }\overline{c_{%
\mathbf{x}}}c_{\mathbf{y}}B_{t+i\alpha ,\upsilon ,\varepsilon }^{(\omega
)}(y^{(1)},x^{(1)})A_{-t+i(\beta -\alpha ),\upsilon ,\varepsilon }^{(\omega
)}(x^{(2)},y^{(2)})\right\vert &\leq &\varepsilon D\underset{\mathbf{x}\in
\mathfrak{K}}{\sum }\left\vert c_{\mathbf{x}}\right\vert ^{2}, \\
\left\vert \underset{\mathbf{x,y}\in \mathfrak{K}}{\sum }\overline{c_{%
\mathbf{x}}}c_{\mathbf{y}}A_{t+i\alpha ,\upsilon ,\varepsilon }^{(\omega
)}(y^{(1)},x^{(1)})B_{-t+i(\beta -\alpha ),\upsilon ,\varepsilon }^{(\omega
)}(x^{(2)},y^{(2)})\right\vert &\leq &\varepsilon D\underset{\mathbf{x}\in
\mathfrak{K}}{\sum }\left\vert c_{\mathbf{x}}\right\vert ^{2}, \\
\left\vert \underset{\mathbf{x,y}\in \mathfrak{K}}{\sum }\overline{c_{%
\mathbf{x}}}c_{\mathbf{y}}A_{t+i\alpha ,\upsilon ,\varepsilon }^{(\omega
)}(y^{(1)},x^{(1)})A_{-t+i(\beta -\alpha ),\upsilon ,\varepsilon }^{(\omega
)}(x^{(2)},y^{(2)})\right\vert &\leq &\varepsilon ^{2}D\underset{\mathbf{x}%
\in \mathfrak{K}}{\sum }\left\vert c_{\mathbf{x}}\right\vert ^{2},
\end{eqnarray*}%
provided $\alpha \in \lbrack \upsilon ,\beta -\upsilon ]$ with $\upsilon \in
(0,\beta /2)$. Here, $\mathbf{x}=(x^{(1)},x^{(2)})$, $\mathbf{y}%
=(y^{(1)},y^{(2)})$. Similar to (\ref{Lemma non existing2bis}), we then use
these three above bounds to get the existence of a finite constant $D\in
\mathbb{R}^{+}$ depending on $\beta ,\upsilon ,\mathbf{A}$ but not on $l\in
\mathbb{R}^{+}$, $\varepsilon \in (0,1)$, $\omega \in \Omega $, $\lambda \in
\mathbb{R}_{0}^{+}$ and $s_{1},s_{2}\in \mathbb{R}$ such that%
\begin{equation*}
\left\vert \mathbf{X}_{l,\upsilon }^{(\omega )}\left( s_{1},s_{2}\right) -%
\mathbf{Y}_{l,\upsilon ,\varepsilon ,0}^{(\omega )}(s_{1},s_{2})\right\vert
\leq \varepsilon D\ .
\end{equation*}
\end{proof}

The approximating correlation functions $B_{t+i\alpha ,\upsilon ,\varepsilon
}^{(\omega )}$ in (\ref{B fract}) rapidly vanish, as $|y^{\pi ^{\prime
}(1)}-x^{\pi (1)}|\rightarrow \infty $ or $|x^{\pi (2)}-y^{\pi ^{\prime
}(2)}|\rightarrow \infty $, see Theorem \ref{decay bound theorem} (ii). The
decay is uniform for times $t$ on compact sets. This property will allows us
further on to use piecewise--constant approximations of the smooth electric
field $E_{\mathbf{A}}$ (\ref{V bar 0}) in (\ref{B fractbis}), similar to
what is done in the preceding subsection.

To do this, as in the proof of Theorem \ref{main 1 copy(18)}, let us assume
w.l.o.g. that, for all $t\in \mathbb{R}$,
\begin{equation*}
\mathrm{supp}(\mathbf{A}(t,.))\subset \lbrack -1/2,1/2]^{d}\ .
\end{equation*}%
For all $n\in \mathbb{N}$, we divide the elementary box $[-1/2,1/2]^{d}$ in $%
n^{d}$ boxes $\{b_{j}\}_{j\in \mathcal{D}_{n}}$ of side--length $1/n$. The
sets $\mathcal{D}_{n}$ and $\{b_{j}\}_{j\in \mathcal{D}_{n}}$ are defined by
(\ref{boxes b1})--(\ref{boxes b2}), respectively. Then, for all $\varepsilon
,l,\beta \in \mathbb{R}^{+}$, $\omega \in \Omega $, $\lambda \in \mathbb{R}%
_{0}^{+}$ and $\upsilon \in (0,\beta /2)$, we extend the definition of $%
\mathbf{Y}_{l,\upsilon ,\varepsilon ,0}^{(\omega )}$ to all $n\in \mathbb{N}$
as
\begin{eqnarray}
\mathbf{Y}_{l,\upsilon ,\varepsilon ,n}^{(\omega )}(s_{1},s_{2}) &:=&\frac{1}{4\left\vert \Lambda _{l}\right\vert }\ \underset{j\in \mathcal{D}_{n}}{\sum }\ \underset{\mathbf{x},\mathbf{y}\in \mathfrak{K}\cap (lb_{j})^{2}}{\sum }\ \int\nolimits_{\upsilon }^{\beta -\upsilon }\ \mathrm{d}\alpha
\label{B fractbisbis} \\
&&\qquad \times \left( \mathfrak{B}_{s_{1}-s_{2}+i\alpha ,\upsilon
,\varepsilon }^{(\omega )}(\mathbf{x},\mathbf{y})-\mathfrak{B}_{i\alpha
,\upsilon ,\varepsilon }^{(\omega )}(\mathbf{x},\mathbf{y})\right) \mathbf{E}_{s_{1}}^{\mathbf{A}_{l}}\left( \mathbf{x}\right) \mathbf{E}_{s_{2}}^{\mathbf{A}_{l}}(\mathbf{y})  \notag
\end{eqnarray}%
for all $s_{1},s_{2}\in \mathbb{R}$. In fact, the accumulation points of $%
\mathbf{Y}_{l,\upsilon ,\varepsilon ,n}^{(\omega )}$, as $l\rightarrow
\infty $, do not depend on $n$:

\begin{lemma}[Approximation III]
\label{lemma conductivty1}\mbox{
}\newline
Let $n\in \mathbb{N}$, $\varepsilon ,\beta \in \mathbb{R}^{+}$, $\omega \in
\Omega $, $\lambda \in \mathbb{R}_{0}^{+}$ and $\upsilon \in (0,\beta /2) $.
Then,
\begin{equation*}
\underset{l\rightarrow \infty }{\lim }\left\vert \mathbf{Y}_{l,\upsilon
,\varepsilon ,0}^{(\omega )}\left( s_{1},s_{2}\right) -\mathbf{Y}%
_{l,\upsilon ,\varepsilon ,n}^{(\omega )}\left( s_{1},s_{2}\right)
\right\vert =0
\end{equation*}%
uniformly for $s_{1},s_{2}\in \mathbb{R}$.
\end{lemma}

\begin{proof}
We observe from (\ref{boxes b2}) and (\ref{B fractbis})--(\ref{B fractbisbis}%
) that%
\begin{eqnarray}
&&\left\vert \mathbf{Y}_{l,\upsilon ,\varepsilon ,0}^{(\omega
)}(s_{1},s_{2})-\mathbf{Y}_{l,\upsilon ,\varepsilon ,n}^{(\omega
)}(s_{1},s_{2})\right\vert  \notag \\
&\leq &\frac{1}{4\left\vert \Lambda _{l}\right\vert }\underset{j,k\in
\mathcal{D}_{n},j\neq k}{\sum }\ \underset{\mathbf{x}\in \mathfrak{K}\cap
(lb_{j})^{2}}{\sum }\ \underset{\mathbf{y}\in \mathfrak{K}\cap (lb_{k})^{2}}{%
\sum }\left\vert \mathbf{K}_{s_{1},s_{2}}\left( \mathbf{x},\mathbf{y}\right)
\right\vert  \label{B fractbisbisbis} \\
&&+\frac{1}{4\left\vert \Lambda _{l}\right\vert }\underset{j\in \mathcal{D}%
_{n}}{\sum }\ \underset{\mathbf{x}\in \partial (lb_{j})}{\sum }\ \underset{%
\mathbf{y}\in \mathfrak{K}}{\sum }\left\vert \mathbf{K}_{s_{1},s_{2}}\left(
\mathbf{x},\mathbf{y}\right) +\mathbf{K}_{s_{1},s_{2}}\left( \mathbf{y},%
\mathbf{x}\right) \right\vert \ ,  \notag
\end{eqnarray}%
where, for any $\Lambda \in \mathcal{P}_{f}(\mathfrak{L})$ with complement $%
\Lambda ^{c}\subset \mathfrak{L}$,%
\begin{equation*}
\partial \Lambda :=\left\{ (x^{(1)},x^{(2)})\in \mathfrak{K}%
:\{x^{(1)},x^{(2)}\}\cap \Lambda \neq 0,\text{ }\{x^{(1)},x^{(2)}\}\cap
\Lambda ^{c}\neq 0\right\}
\end{equation*}%
and
\begin{equation*}
\mathbf{K}_{s_{1},s_{2}}\left( \mathbf{x},\mathbf{y}\right)
:=\int\nolimits_{\upsilon }^{\beta -\upsilon }\mathrm{d}\alpha \left(
\mathfrak{B}_{s_{1}-s_{2}+i\alpha ,\upsilon ,\varepsilon }^{(\omega )}(%
\mathbf{x},\mathbf{y})-\mathfrak{B}_{i\alpha ,\upsilon ,\varepsilon
}^{(\omega )}(\mathbf{x},\mathbf{y})\right) \mathbf{E}_{s_{1}}^{\mathbf{A}%
_{l}}\left( \mathbf{x}\right) \mathbf{E}_{s_{2}}^{\mathbf{A}_{l}}(\mathbf{y}%
)\ .
\end{equation*}%
Meanwhile, for any $j,k\in \mathcal{D}_{n}$, $j\neq k$, and every
\begin{equation*}
\mathbf{x:=}(x^{(1)},x^{(2)})\in \mathfrak{K}\ ,\qquad \mathbf{y:=}%
(y^{(1)},y^{(2)})\in \mathfrak{K}\ ,
\end{equation*}%
one clearly has the lower bound%
\begin{equation*}
\underset{\pi ,\pi ^{\prime }\in S_{2}}{\min }\left\vert x^{\pi (1)}-y^{\pi
^{\prime }(1)}\right\vert \geq \min \left\{ \left\vert \left\vert
x^{(1)}-y^{(1)}\right\vert -2\right\vert, \left\vert
x^{(1)}-y^{(1)}\right\vert \right\} \ ,
\end{equation*}%
see (\ref{proche voisins0}). We use this simple inequality together with (%
\ref{bound easy}) and Theorem \ref{decay bound theorem} (ii) to obtain from (%
\ref{B fract}) and (\ref{B fractbisbisbis}) that, for all $s_{1},s_{2}\in
\mathbb{R}$,%
\begin{eqnarray}
&&\left\vert \mathbf{Y}_{l,\upsilon ,\varepsilon ,0}^{(\omega
)}(s_{1},s_{2})-\mathbf{Y}_{l,\upsilon ,\varepsilon ,n}^{(\omega
)}(s_{1},s_{2})\right\vert  \notag \\
&\leq &\frac{D}{l^{d}}\underset{j,k\in \mathcal{D}_{n},j\neq k}{\sum }%
\sum\limits_{x\in \mathfrak{L}\cap (lb_{j})}\ \sum\limits_{y\in \mathfrak{L}%
\cap (lb_{k})}\frac{1}{\left( 1+\left\vert x-y\right\vert \right) ^{2d^{2}+2}%
}  \label{eq ohm4} \\
&&+\frac{D}{l^{d}}\underset{j\in \mathcal{D}_{n}}{\sum }\ \underset{%
(x^{(1)},x^{(2)})\in \partial (lb_{j})}{\sum }\ \underset{y\in \mathfrak{L}}{%
\sum }\frac{1}{\left( 1+ \left\vert x^{(1)}- y\right\vert \right) ^{2d^{2}+2}%
}\ ,  \notag
\end{eqnarray}%
where $D\in \mathbb{R}^{+}$ is a finite constant only depending on $%
\varepsilon ,\beta ,\upsilon ,d$ and $\mathbf{A}\in \mathbf{C}_{0}^{\infty }$%
. Note that the second term of the r.h.s. of the above inequality is of
order $\mathcal{O}(l^{-1})$. For any small $\delta >0$ with $\delta l \geq 1$%
,%
\begin{equation*}
\frac{1}{l^{d}}\underset{j,k\in \mathcal{D}_{n},j\neq k}{\sum }%
\sum\limits_{x\in \mathfrak{L}\cap (lb_{j})}\ \sum\limits_{y\in \mathfrak{L}%
\cap (lb_{k})}\frac{\mathbf{1}\left[ |x-y|\geq \delta l\right] }{\left( 1+
\left\vert x-y\right\vert \right) ^{2d^{2}+2}}=\mathcal{O}\left( \frac{1}{%
l^{2d^{2}-d+2}\delta ^{2d^{2}+2}}\right)
\end{equation*}%
and%
\begin{equation*}
\frac{1}{l^{d}}\underset{j,k\in \mathcal{D}_{n},j\neq k}{\sum }%
\sum\limits_{x\in \mathfrak{L}\cap (lb_{j})}\ \sum\limits_{y\in \mathfrak{L}%
\cap (lb_{k})}\frac{\mathbf{1}\left[ |x-y|\leq \delta l\right] }{\left( 1+
\left\vert x-y\right\vert \right) ^{2d^{2}+2}}=\mathcal{O}\left( \delta
^{d+1}l^{d}\right) \ .
\end{equation*}%
Then, by choosing $\delta =l^{-\frac{2d^{2}+2}{2d^{2}+d+3}}$, the last two
sums are both of order $\mathcal{O}(l^{-\frac{d^{2}-d+2}{2d^{2}+d+3}})$ with
$d^{2}-d+2\geq 2$ for all $d\in \mathbb{N}$. Using this estimate to bound
the first term of the r.h.s. of (\ref{eq ohm4}) we arrive at the assertion.
\end{proof}

As already mentioned above, we now consider piecewise--constant
approximations of the (smooth) electric field $E_{\mathbf{A}}$ (\ref{V bar}%
). For any $j\in \mathcal{D}_{n}$, let $z^{(j)}\in b_{j}$ be any fixed point
of the box $b_{j}$. Then, for any $s_{1},s_{2}\in \mathbb{R}$, define the
function%
\begin{eqnarray}
\mathbf{\bar{Y}}_{l,\upsilon ,\varepsilon ,n}^{(\omega )}(s_{1},s_{2}) &:=&\frac{1}{4\left\vert \Lambda _{l}\right\vert }\ \underset{j\in \mathcal{D}_{n}}{\sum }\ \underset{\mathbf{x},\mathbf{y}\in \mathfrak{K}\cap
(lb_{j})^{2}}{\sum }\ \int\nolimits_{\upsilon }^{\beta -\upsilon }\ \mathrm{d}\alpha   \label{piece wise approx} \\
&&\qquad \times \left( \mathfrak{B}_{s_{1}-s_{2}+i\alpha ,\upsilon
,\varepsilon }^{(\omega )}(\mathbf{x},\mathbf{y})-\mathfrak{B}_{i\alpha
,\upsilon ,\varepsilon }^{(\omega )}(\mathbf{x},\mathbf{y})\right)   \notag
\\
&&\qquad \times \left[ E_{\mathbf{A}}(s_{1},z^{(j)})\right] (x^{(2)}-x^{(1)})\left[ E_{\mathbf{A}}(s_{2},z^{(j)})\right] (y^{(2)}-y^{(1)})\ .  \notag
\end{eqnarray}%
Recall that $\mathbf{x}:=(x^{(1)},x^{(2)})\in \mathfrak{L}^{2}$ and $\mathbf{%
y}:=(y^{(1)},y^{(2)})\in \mathfrak{L}^{2}$. See also (\ref{proche voisins0}%
). This new function approximates (\ref{B fractbisbis}) arbitrarily well, as
$l\rightarrow \infty $ and $n\rightarrow \infty $:

\begin{lemma}[Approximation IV]
\label{lemma conductivty2}\mbox{
}\newline
Let $\varepsilon ,\beta \in \mathbb{R}^{+}$, $\omega \in \Omega $, $\lambda
\in \mathbb{R}_{0}^{+}$ and $\upsilon \in (0,\beta /2)$. Then,%
\begin{equation*}
\underset{n\rightarrow \infty }{\lim }\left\{ \underset{l\rightarrow \infty }%
{\lim \sup }\left\vert \mathbf{Y}_{l,\upsilon ,\varepsilon ,n}^{(\omega
)}\left( s_{1},s_{2}\right) -\mathbf{\bar{Y}}_{l,\upsilon ,\varepsilon
,n}^{(\omega )}\left( s_{1},s_{2}\right) \right\vert \right\} =0
\end{equation*}%
uniformly for $s_{1},s_{2}\in \mathbb{R}$.
\end{lemma}

\begin{proof}
Using (\ref{bound easy}), (\ref{inequality G cool1}) and Theorem \ref{decay
bound theorem} (ii) as in (\ref{eq ohm4}), one gets that, for any $%
s_{1},s_{2}\in \mathbb{R}$,%
\begin{eqnarray}
&&\left\vert \mathbf{Y}_{l,\upsilon ,\varepsilon ,n}^{(\omega )}\left(
s_{1},s_{2}\right) -\mathbf{\bar{Y}}_{l,\upsilon ,\varepsilon ,n}^{(\omega
)}\left( s_{1},s_{2}\right) \right\vert  \label{eq ohm6} \\
&\leq &D(n^{-1}+l^{-1})\frac{1}{l^{d}}\underset{j\in \mathcal{D}_{n}}{\sum }%
\sum\limits_{x,y\in \mathfrak{L}\cap (lb_{j})}\frac{1}{\left( 1+ \left\vert
x-y\right\vert \right) ^{2d^{2}+2}}\ ,  \notag
\end{eqnarray}%
where $D\in \mathbb{R}^{+}$ is a finite constant only depending on $%
\varepsilon ,\beta ,\upsilon ,d$ and $\mathbf{A}\in \mathbf{C}_{0}^{\infty }$%
. For all $j\in \mathcal{D}_{n}$ and $l>1$, note that%
\begin{equation*}
\frac{1}{l^{d}}\sum\limits_{x,y\in \mathfrak{L}\cap (lb_{j})}\frac{1}{\left(
1+ \left\vert x-y\right\vert \right) ^{2d^{2}+2}}\leq \frac{(2l+1)^{d}}{%
n^{d}l^{d}}\sum\limits_{x\in \mathfrak{L}}\frac{1}{\left( 1+ \left\vert
x\right\vert \right) ^{2d^{2}+2}}\leq \frac{D}{n^{d}}
\end{equation*}%
for some finite constant $D\in \mathbb{R}^{+}$. Therefore, we arrive at the
assertion by combining this last bound with (\ref{eq ohm6}).
\end{proof}

By taking the canonical orthonormal basis $\{e_{k}\}_{k=1}^{d}$ of $\mathbb{R%
}^{d}$ and setting $e_{-k}:=-e_{k}$ for each $k\in \{1,\ldots ,d\}$, we
rewrite the function (\ref{piece wise approx}) as%
\begin{eqnarray}
\mathbf{\bar{Y}}_{l,\upsilon ,\varepsilon ,n}^{(\omega )}(s_{1},s_{2}) &=&%
\frac{1}{4n^{d}}\underset{j\in \mathcal{D}_{n}}{\sum }\ \underset{k,q\in
\{1,-1,\ldots ,d,-d\}}{\sum }\left( \mathbf{Z}_{l,j,k,q}^{(\omega
)}(s_{1}-s_{2})-\mathbf{Z}_{l,j,k,q}^{(\omega )}(0)\right)  \notag \\
&&\qquad \qquad \times \left[ E_{\mathbf{A}}(s_{1},z^{(j)})\right] (e_{q})%
\left[ E_{\mathbf{A}}(s_{2},z^{(j)})\right] (e_{k})  \label{Y bar}
\end{eqnarray}%
for any $s_{1},s_{2}\in \mathbb{R}$, where, for all $n\in \mathbb{N}$, $%
\varepsilon ,l,\beta \in \mathbb{R}^{+}$, $\omega \in \Omega $, $\lambda \in
\mathbb{R}_{0}^{+}$, $\upsilon \in (0,\beta /2)$, $j\in \mathcal{D}_{n}$, $%
k,q\in \{1,-1,\ldots ,d,-d\}$ and $t\in \mathbb{R}$,
\begin{equation*}
\mathbf{Z}_{l,j,k,q}^{(\omega )}(t):=\frac{n^{d}}{\left\vert \Lambda
_{l}\right\vert }\sum\limits_{x,y\in \mathfrak{L}\cap
(lb_{j})}\int\nolimits_{\upsilon }^{\beta -\upsilon }\mathrm{d}\alpha \
\mathfrak{B}_{t+i\alpha ,\upsilon ,\varepsilon }^{(\omega
)}(x,x-e_{q},y,y-e_{k})\ .
\end{equation*}%
Notice that we have added terms related to $x,y$ on the boundary of $%
\mathfrak{L}\cap (lb_{j})$, but we use the same notation $\mathbf{\bar{Y}}%
_{l,\upsilon ,\varepsilon ,n}^{(\omega )}$ for simplicity. These terms are
indeed irrelevant in the limit $l\rightarrow \infty $. For $\mathbf{x}%
:=(x^{(1)},x^{(2)})\in \mathfrak{L}^{2}$ and $\mathbf{y}:=(y^{(1)},y^{(2)})%
\in \mathfrak{L}^{2}$, we used the notation
\begin{equation}
\mathfrak{B}_{t+i\alpha ,\upsilon ,\varepsilon }^{(\omega
)}(x^{(1)},x^{(2)},y^{(1)},y^{(2)}) \equiv \mathfrak{B}_{t+i\alpha ,\upsilon
,\varepsilon }^{(\omega )}(\mathbf{x},\mathbf{y}) \ ,
\label{B fract notation}
\end{equation}%
see (\ref{B fract}). By Theorem \ref{decay bound theorem} (i), (iv) and
Lebesgue's dominated convergence theorem, note that, for all $\varepsilon
,\beta \in \mathbb{R}^{+}$, $\lambda \in \mathbb{R}_{0}^{+}$, $x,y\in
\mathfrak{L}$, $k,q\in \{1,-1,\ldots ,d,-d\}$, $t\in \mathbb{R}$ and $%
\upsilon \in (0,\beta /2)$, the map
\begin{equation*}
\omega \mapsto \int\nolimits_{0}^{\beta }\mathrm{d}\alpha \ \mathfrak{B}%
_{t+i\alpha ,\upsilon ,\varepsilon }^{(\omega )}(x,x-e_{q},y,y-e_{k})
\end{equation*}%
is bounded and measurable w.r.t. the $\sigma $--algebra $\mathfrak{A}%
_{\Omega }$. In particular, its expectation value $\mathbb{E}[\ \cdot \ ]$
w.r.t. the probability measure $\mathfrak{a}_{\Omega }$ (\ref{probability
measure}) is well--defined. It now remains to analyze the limit of $\mathbf{Z%
}_{l,j,k,q}^{(\omega )}$, as $l\rightarrow \infty $.

\begin{lemma}[Infinite volume limit and ergodicity]
\label{lemma conductivty3}\mbox{
}\newline
Let $\varepsilon ,\beta \in \mathbb{R}^{+}$, $\lambda \in \mathbb{R}_{0}^{+}$%
, $t\in \mathbb{R}$ and $\upsilon \in (0,\beta /2)$. Then, there is a
measurable subset $\tilde{\Omega}_{\upsilon ,\varepsilon }\left( t\right)
\equiv \tilde{\Omega}_{\upsilon ,\varepsilon }^{(\beta ,\lambda )}\left(
t\right) \subset \Omega $ of full measure such that, for any $n\in \mathbb{N}
$, $j\in \mathcal{D}_{n}$, $k,q\in \{1,-1,\ldots ,d,-d\}$ and any $\tilde{%
\omega}\in \tilde{\Omega}_{\upsilon ,\varepsilon }\left( t\right) $,%
\begin{equation*}
\underset{l\rightarrow \infty }{\lim }\mathbf{Z}_{l,j,k,q}^{(\tilde{\omega}%
)}(t)=\sum\limits_{x\in \mathfrak{L}}\mathbb{E}\left[ \int\nolimits_{%
\upsilon }^{\beta -\upsilon }\mathrm{d}\alpha \ \mathfrak{B}_{t+i\alpha
,\upsilon ,\varepsilon }^{(\omega )}(x,x-e_{q},0,-e_{k})\right] \in \mathbb{R%
}\ .
\end{equation*}
\end{lemma}

\begin{proof}
The arguments are similar to those proving Theorems \ref{Ackoglu--Krengel
ergodic theorem III} or \ref{main 1 copy(18)}, but a little bit more
complicated. For the reader's convenience, we give the proof in detail. For
any $\varepsilon ,\beta \in \mathbb{R}^{+}$, $\omega \in \Omega $, $\lambda
\in \mathbb{R}_{0}^{+}$, $t\in \mathbb{R}$, $\upsilon \in (0,\beta /2)$, $%
k,q\in \{1,-1,\ldots ,d,-d\}$ and $y\in \mathfrak{L}$, let
\begin{equation}
\mathfrak{F}_{t,\upsilon ,\varepsilon ,k,q}^{(\beta ,\omega ,\lambda
)}\left( \left\{ y\right\} \right) :=\sum\limits_{x\in \mathfrak{L}%
}\int\nolimits_{\upsilon }^{\beta -\upsilon }\mathrm{d}\alpha \ \mathfrak{B}%
_{t+i\alpha ,\upsilon ,\varepsilon }^{(\omega )}(x,x-e_{q},y,y-e_{k})\in
\mathbb{R}\ .  \label{eq ohm70}
\end{equation}%
This infinite sum absolutely converges because of (\ref{B fract}) and
Theorem \ref{decay bound theorem} (ii). We now define an additive process $\{%
\mathfrak{F}_{t,\upsilon ,\varepsilon ,k,q}^{(\beta ,\omega ,\lambda
)}\left( \Lambda \right) \}_{\Lambda \in \mathcal{P}_{f}(\mathfrak{L})}$ by
\begin{equation*}
\mathfrak{F}_{t,\upsilon ,\varepsilon ,k,q}^{(\beta ,\omega ,\lambda
)}\left( \Lambda \right) :=\sum\limits_{y\in \Lambda }\mathfrak{F}%
_{t,\upsilon ,\varepsilon ,k,q}^{(\beta ,\omega ,\lambda )}\left( \left\{
y\right\} \right)
\end{equation*}%
for any finite subset $\Lambda \in \mathcal{P}_{f}(\mathfrak{L})$, see
Definition \ref{Additive process}. Indeed, by Theorem \ref{decay bound
theorem} (i), (iv) and Lebesgue's dominated convergence theorem, the map $%
\omega \mapsto \mathfrak{F}_{t,\upsilon ,\varepsilon ,k,q}^{(\beta ,\omega
,\lambda )}\left( \Lambda \right) $ is bounded and measurable (in fact
continuous) w.r.t. the $\sigma $--algebra $\mathfrak{A}_{\Omega }$ for all $%
\Lambda \in \mathcal{P}_{f}(\mathfrak{L})$. Then, for any $\varepsilon
,\beta \in \mathbb{R}^{+}$, $\lambda \in \mathbb{R}_{0}^{+}$, $t\in \mathbb{R%
}$ and $\upsilon \in (0,\beta /2)$, we apply Theorem \ref{Ackoglu--Krengel
ergodic theorem II} on the previous additive process to get the existence of
a measurable subset
\begin{equation*}
\tilde{\Omega}_{\upsilon ,\varepsilon }\left( t\right) \equiv \tilde{\Omega}%
_{\upsilon ,\varepsilon }^{(\beta ,\lambda )}\left( t\right) \subset \Omega
\end{equation*}%
of full measure such that, for all $\tilde{\omega}\in \tilde{\Omega}%
_{\upsilon ,\varepsilon }\left( t\right) $, $n\in \mathbb{N}$, $j\in
\mathcal{D}_{n}$ and $k,q\in \{1,-1,\ldots ,d,-d\}$,%
\begin{equation}
\underset{l\rightarrow \infty }{\lim }\left\{ \frac{n^{d}}{\left\vert
\Lambda _{l}\right\vert }\ \mathfrak{F}_{t,\upsilon ,\varepsilon
,k,q}^{(\beta ,\tilde{\omega},\lambda )}\left( lb_{j}\right) \right\} =%
\mathbb{E}\left[ \mathfrak{F}_{t,\upsilon ,\varepsilon ,k,q}^{(\beta ,\omega
,\lambda )}\left( \left\{ 0\right\} \right) \right] \ .  \label{eq ohm7}
\end{equation}%
Note that to prove this equation we use once again that any countable
intersection of measurable sets of full measure has full measure. In the way
one proves Lemma \ref{lemma conductivty1}, one verifies that
\begin{equation*}
\underset{l\rightarrow \infty }{\lim }\left\{ \frac{n^{d}}{\left\vert
\Lambda _{l}\right\vert }\sum\limits_{y\in \mathfrak{L}\cap
(lb_{j})}\sum\limits_{x\in \mathfrak{L}\backslash
(lb_{j})}\int\nolimits_{\upsilon }^{\beta -\upsilon }\mathrm{d}\alpha \
\mathfrak{B}_{t+i\alpha ,\upsilon ,\varepsilon }^{(\omega
)}(x,x-e_{q},y,y-e_{k})\right\} =0\ .
\end{equation*}%
Using this with (\ref{eq ohm70})--(\ref{eq ohm7}) and observing that
\begin{equation*}
\mathbb{E}\left[ \mathfrak{F}_{t,\upsilon ,\varepsilon ,k,q}^{(\beta ,\omega
,\lambda )}\left( \left\{ 0\right\} \right) \right] =\sum\limits_{x\in
\mathfrak{L}}\mathbb{E}\left[ \int\nolimits_{\upsilon }^{\beta -\upsilon }%
\mathrm{d}\alpha \ \mathfrak{B}_{t+i\alpha ,\upsilon ,\varepsilon }^{(\omega
)}(x,x-e_{q},0,-e_{k})\right] \ ,
\end{equation*}%
we arrive at the assertion for any realization $\tilde{\omega}\in \tilde{%
\Omega}_{\upsilon ,\varepsilon }\left( t\right) $.
\end{proof}

For all $\varepsilon ,\beta \in \mathbb{R}^{+}$, $\lambda \in \mathbb{R}%
_{0}^{+}$, $\upsilon \in (0,\beta /2)$ and $k,q\in \{1,-1,\ldots ,d,-d\}$,
define the functions%
\begin{equation}
\tilde{\Gamma}_{\upsilon ,\varepsilon ,k,q}(t):=\sum\limits_{x\in \mathfrak{L%
}}\mathbb{E}\left[ \int\nolimits_{\upsilon }^{\beta -\upsilon }\mathrm{d}%
\alpha \ \mathfrak{B}_{t+i\alpha ,\upsilon ,\varepsilon }^{(\omega
)}(x,x-e_{q},0,-e_{k})\right]  \label{conductivity00}
\end{equation}%
for any $t\in \mathbb{R}$, and
\begin{eqnarray}
\mathbf{Y}_{\infty ,\upsilon ,\varepsilon }(s_{1},s_{2}) &:=&\underset{k,q\in \{1,-1,\ldots ,d,-d\}}{\sum }\left( \tilde{\Gamma}_{\upsilon
,\varepsilon ,k,q}(s_{1}-s_{2})-\tilde{\Gamma}_{\upsilon ,\varepsilon
,k,q}(0)\right)   \notag \\
&&\times \int\nolimits_{\mathbb{R}^{d}}\mathrm{d}^{d}x\left[ E_{\mathbf{A}}(s_{1},x)\right] (e_{q})\left[ E_{\mathbf{A}}(s_{2},x)\right] (e_{k})
\label{conductivity00-1}
\end{eqnarray}%
for any $s_{1},s_{2}\in \mathbb{R}$. We show next that the function $\mathbf{%
Y}_{l,\upsilon ,\varepsilon ,0}^{(\omega )}$ defined by (\ref{B fractbis})
almost surely converges to the deterministic function $\mathbf{Y}_{\infty
,\upsilon ,\varepsilon }$, as $l\rightarrow \infty $:

\begin{proposition}[Infinite volume limit of the $\mathbf{Y}$--approximation]

\label{main 1 copy(3)}\mbox{
}\newline
Let $\varepsilon ,\beta \in \mathbb{R}^{+}$, $\lambda \in \mathbb{R}_{0}^{+}$%
, $\upsilon \in (0,\beta /2)$ and $s_{1},s_{2}\in \mathbb{R}$. Then, there
is a measurable subset $\tilde{\Omega}_{\upsilon ,\varepsilon }\left(
s_{1},s_{2}\right) \equiv \tilde{\Omega}_{\upsilon ,\varepsilon }^{(\beta
,\lambda )}\left( s_{1},s_{2}\right) \subset \Omega $ of full measure such
that, for any $\mathbf{A}\in \mathbf{C}_{0}^{\infty }$ and $\omega \in
\tilde{\Omega}_{\upsilon ,\varepsilon }\left( s_{1},s_{2}\right) $,%
\begin{equation}
\underset{l\rightarrow \infty }{\lim }\mathbf{Y}_{l,\upsilon ,\varepsilon
,0}^{(\omega )}\left( s_{1},s_{2}\right) =\mathbf{Y}_{\infty ,\upsilon
,\varepsilon }(s_{1},s_{2})\ .  \notag
\end{equation}
\end{proposition}

\begin{proof}
Let $\varepsilon ,\beta \in \mathbb{R}^{+}$, $\lambda \in \mathbb{R}_{0}^{+}$%
, $\upsilon \in (0,\beta /2)$, $\mathbf{A}\in \mathbf{C}_{0}^{\infty }$ and $%
s_{1},s_{2}\in \mathbb{R}$. Using Lemmata \ref{lemma conductivty1}--\ref%
{lemma conductivty3} and (\ref{Y bar}), we obtain the existence of a
measurable subset $\tilde{\Omega}_{\upsilon ,\varepsilon }\left(
s_{1},s_{2}\right) \equiv \tilde{\Omega}_{\upsilon ,\varepsilon }^{(\beta
,\lambda )}\left( s_{1},s_{2}\right) \subset \Omega $ of full measure such
that, for any $\omega \in \tilde{\Omega}_{\upsilon ,\varepsilon }\left(
s_{1},s_{2}\right) $,
\begin{multline*}
\underset{l\rightarrow \infty }{\lim }\mathbf{Y}_{l,\upsilon ,\varepsilon
,0}^{(\omega )}\left( s_{1},s_{2}\right) =\underset{k,q\in \{1,-1,\ldots
,d,-d\}}{\sum }\left( \tilde{\Gamma}_{\upsilon ,\varepsilon
,k,q}(s_{1}-s_{2})-\tilde{\Gamma}_{\upsilon ,\varepsilon ,k,q}(0)\right) \\
\times \underset{n\rightarrow \infty }{\lim }\left\{ \frac{1}{4n^{d}}%
\underset{j\in \mathcal{D}_{n}}{\sum }\left[ E_{\mathbf{A}}(s_{1},z^{(j)})%
\right] (e_{q})\left[ E_{\mathbf{A}}(s_{2},z^{(j)})\right] (e_{k})\right\} \
.
\end{multline*}%
The latter implies the proposition because the term within the limit $%
n\rightarrow \infty $ is a Riemann sum and $E_{\mathbf{A}}\in \mathbf{C}%
_{0}^{\infty }$ for any $\mathbf{A}\in \mathbf{C}_{0}^{\infty }$, see (\ref%
{V bar}).
\end{proof}

This last limit depends on the two arbitrary parameters $\varepsilon \in
\mathbb{R}^{+}$ and $\upsilon \in (0,\beta /2)$, where $\beta \in \mathbb{R}%
^{+}$. The next step is to remove them by considering the limits $%
\varepsilon \rightarrow 0^{+}$ and $\upsilon \rightarrow 0^{+}$.

We first observe that the functions (\ref{conductivity00}) are
approximations of the function $\Gamma _{k,q}\equiv \Gamma _{k,q}^{(\beta
,\lambda )}$ defined, for any $\beta \in \mathbb{R}^{+}$, $\lambda \in
\mathbb{R}_{0}^{+}$, $k,q\in \{1,-1,\ldots ,d,-d\}$ and $t\in \mathbb{R}$, by%
\begin{equation}
\Gamma _{k,q}(t):=\underset{l\rightarrow \infty }{\lim }\frac{1}{\left\vert
\Lambda _{l}\right\vert }\sum\limits_{x,y\in \Lambda _{l}}\mathbb{E}\left[
\int\nolimits_{0}^{\beta }\mathrm{d}\alpha \ \mathfrak{C}_{t+i\alpha
}^{(\omega )}(x,x-e_{q},y,y-e_{k})\right] \ .  \label{conductivity0}
\end{equation}%
By Theorem \ref{decay bound theorem} (i), (iv) and Lebesgue's dominated
convergence theorem, note that the map%
\begin{equation*}
\omega \mapsto \int\nolimits_{0}^{\beta }\mathrm{d}\alpha \ \mathfrak{C}%
_{t+i\alpha }^{(\omega )}(x,x-e_{q},y,y-e_{k})
\end{equation*}%
is bounded and measurable w.r.t. the $\sigma $--algebra $\mathfrak{A}%
_{\Omega }$. Here, we use the same convention for the arguments of $%
\mathfrak{C}_{t+i\alpha }^{(\omega )}$ as in (\ref{B fract notation}) for $%
\mathfrak{B}_{t+i\alpha ,\upsilon ,\varepsilon }^{(\omega )}$. This function
is well--defined and it is the limit of $\tilde{\Gamma}_{\upsilon
,\varepsilon ,k,q}$, as $\varepsilon \rightarrow 0^{+}$ and $\upsilon
\rightarrow 0^{+}$:

\begin{lemma}[Approximation on the function $\Gamma $]
\label{lemma conductivty4}\mbox{
}\newline
Let $\varepsilon ,\beta \in \mathbb{R}^{+}$, $\lambda \in \mathbb{R}_{0}^{+}$%
, $t\in \mathbb{R}$, $k,q\in \{1,-1,\ldots ,d,-d\}$ and $\upsilon \in
(0,\beta /2)$. Then, $\Gamma _{k,q}(t)$ exists and equals%
\begin{equation*}
\Gamma _{k,q}(t)=\tilde{\Gamma}_{\upsilon ,\varepsilon ,k,q}(t)+\mathcal{O}%
(\upsilon )+\mathcal{O}_{\upsilon }(\varepsilon )
\end{equation*}%
uniformly for times $t$ in compact sets. The term of order $\mathcal{O}%
_{\upsilon }(\varepsilon )$ vanishes when $\varepsilon \rightarrow 0^{+}$
for any fixed $\upsilon \in (0,\beta /2)$.
\end{lemma}

\begin{proof}
Let $\varepsilon ,\beta \in \mathbb{R}^{+}$, $\lambda \in \mathbb{R}_{0}^{+}$%
, $\upsilon \in (0,\beta /2)$, $t\in \mathbb{R}$ and $k,q\in \{1,-1,\ldots
,d,-d\}$. Using similar arguments to the proof of Lemma \ref{lemma
conductivty1 copy(1)}, one shows that%
\begin{multline*}
\underset{l\rightarrow \infty }{\lim \sup }\frac{1}{\left\vert \Lambda
_{l}\right\vert }\sum\limits_{x,y\in \Lambda _{l}}\mathbb{E}%
\Big[%
\int\nolimits_{\upsilon }^{\beta -\upsilon }\mathrm{d}\alpha \ \left\vert
\mathfrak{B}_{t+i\alpha ,\upsilon ,\varepsilon }^{(\omega
)}(x,x-e_{q},y,y-e_{k})\right. \\
\left. -\mathfrak{C}_{t+i\alpha }^{(\omega
)}(x,x-e_{q},y,y-e_{k})\right\vert
\Big]%
=\mathcal{O}(\varepsilon )
\end{multline*}%
uniformly for $t\in \mathbb{R}$. Moreover, by Theorem \ref{decay bound
theorem} (ii) and translation invariance of $\mathfrak{a}_{\Omega }$ observe
that, for $\upsilon \in (0,\beta /2)$,%
\begin{eqnarray}
\underset{l\rightarrow \infty }{\lim } &&\left\{ \frac{1}{\left\vert \Lambda
_{l}\right\vert }\sum\limits_{x,y\in \Lambda _{l}}\mathbb{E}\left[
\int\nolimits_{\upsilon }^{\beta -\upsilon }\mathrm{d}\alpha \ \mathfrak{B}%
_{t+i\alpha ,\upsilon ,\varepsilon }^{(\omega )}(x,x-e_{q},y,y-e_{k})\right]
\right.  \notag \\
&&\left. -\sum\limits_{x\in \mathfrak{L}}\mathbb{E}\left[ \int\nolimits_{%
\upsilon }^{\beta -\upsilon }\mathrm{d}\alpha \ \mathfrak{B}_{t+i\alpha
,\upsilon ,\varepsilon }^{(\omega )}(x,x-e_{q},0,-e_{k})\right] \right\} =0
\label{oublie a la con}
\end{eqnarray}
uniformly for $t$ in compact sets. Then, one uses the same arguments as in
Lemma \ref{lemma conductivty0} to obtain the assertion, see (\ref%
{conductivity00}) and (\ref{conductivity0}). We omit the details.
\end{proof}

We now consider the limit of the integrand $\mathbf{X}_{l,0}^{(\omega )}$ in
(\ref{Lemma non existing}), as $l\rightarrow \infty $, and show that it
converges almost surely to the deterministic function $\mathbf{X}_{\infty
}\equiv \mathbf{X}_{\infty }^{(\beta ,\lambda )}$ defined, for any $\beta
\in \mathbb{R}^{+}$, $\lambda \in \mathbb{R}_{0}^{+}$ and $s_{1},s_{2}\in
\mathbb{R}$, by%
\begin{eqnarray}
\mathbf{X}_{\infty }(s_{1},s_{2}) &:=&\frac{1}{4}\underset{k,q\in
\{1,-1,\ldots ,d,-d\}}{\sum }\Big(
\Gamma _{k,q}(s_{1}-s_{2})-\Gamma _{k,q}(0)\Big)
\notag \\
&&\times \int\nolimits_{\mathbb{R}^{d}}\mathrm{d}^{d}x\left[ E_{\mathbf{A}}(s_{1},x)\right] (e_{q})\left[ E_{\mathbf{A}}(s_{2},x)\right] (e_{k})\ .
\label{X infinity}
\end{eqnarray}%
%
%
%
%
%
%
%
%
%
%
%
%

\begin{satz}[Infinite volume limit of the $\mathbf{X}$--integrands -- I]
\label{limit ohm1limit ohm1}\mbox{
}\newline
Let $\beta \in \mathbb{R}^{+}$, $\lambda \in \mathbb{R}_{0}^{+}$ and $%
s_{1},s_{2}\in \mathbb{R}$. Then, there is a measurable subset $\tilde{\Omega%
}\left( s_{1},s_{2}\right) \equiv \tilde{\Omega}^{(\beta ,\lambda )}\left(
s_{1},s_{2}\right) \subset \Omega $ of full measure such that, for any $%
\mathbf{A}\in \mathbf{C}_{0}^{\infty }$ and $\omega \in \tilde{\Omega}\left(
s_{1},s_{2}\right) $,%
\begin{equation}
\underset{l\rightarrow \infty }{\lim }\mathbf{X}_{l,0}^{(\omega )}\left(
s_{1},s_{2}\right) =\mathbf{X}_{\infty }\left( s_{1},s_{2}\right) \ .  \notag
\end{equation}
\end{satz}

\begin{proof}
Fix $\beta \in \mathbb{R}^{+}$, $\lambda \in \mathbb{R}_{0}^{+}$ and $%
s_{1},s_{2}\in \mathbb{R}$. Define also the countable sequences $\{\upsilon
_{n}\}_{n\in \mathbb{N}}$ and $\{\varepsilon _{m}\}_{m\in \mathbb{N}}$ by $%
\upsilon _{n}:=n^{-1}$ and $\varepsilon _{m}:=m^{-1}$ for $n,m\in \mathbb{N}$%
. Then, by Proposition \ref{main 1 copy(3)}, for any $n,m\in \mathbb{N}$,
there is a measurable subset $\hat{\Omega}_{n,m}\left( s_{1},s_{2}\right)
\equiv \hat{\Omega}_{n,m}^{(\beta ,\lambda )}\left( s_{1},s_{2}\right)
\subset \Omega $ of full measure such that, for any $\mathbf{A}\in \mathbf{C}%
_{0}^{\infty }$ and $\omega \in \hat{\Omega}_{n,m}\left( s_{1},s_{2}\right) $%
,%
\begin{equation}
\underset{l\rightarrow \infty }{\lim }\mathbf{Y}_{l,\upsilon
_{n},\varepsilon _{m},0}^{(\omega )}\left( s_{1},s_{2}\right) =\mathbf{Y}%
_{\infty ,\upsilon _{n},\varepsilon _{m}}\left( s_{1},s_{2}\right) \ .
\label{measure0ohm}
\end{equation}%
Thus, we define the subset
\begin{equation}
\tilde{\Omega}\left( s_{1},s_{2}\right) :=\underset{n,m\in \mathbb{N}}{%
\bigcap }\hat{\Omega}_{n,m}\left( s_{1},s_{2}\right) \ .  \label{measure1ohm}
\end{equation}%
It has full measure, since it is a countable intersection of measurable sets
of full measure.

Take $\mathbf{A}\in \mathbf{C}_{0}^{\infty }$ and any strictly positive
parameter $\epsilon \in \mathbb{R}^{+}$. Then, by Lemmata \ref{lemma
conductivty0}, \ref{lemma conductivty1 copy(1)} and \ref{lemma conductivty4}%
, there are $N_{\epsilon }, M_\epsilon \in \mathbb{N}$ such that, for all $%
l\in \mathbb{R}^{+}$ and $\omega \in \tilde{\Omega}\left( s_{1},s_{2}\right)
$,%
\begin{equation*}
\left\vert \mathbf{X}_{l,0}^{(\omega )}\left( s_{1},s_{2}\right) -\mathbf{X}%
_{\infty }\left( s_{1},s_{2}\right) \right\vert \leq \epsilon +\left\vert
\mathbf{Y}_{l,\upsilon _{N_\epsilon},\varepsilon _{M_\epsilon},0}^{(\omega
)}\left( s_{1},s_{2}\right) -\mathbf{Y}_{\infty ,\upsilon
_{N_\epsilon},\varepsilon _{M_\epsilon}}\left( s_{1},s_{2}\right)
\right\vert \ .
\end{equation*}%
Therefore, we arrive at the assertion by combining this bound together with (%
\ref{measure0ohm}) for any realization $\omega \in \tilde{\Omega}\left(
s_{1},s_{2}\right) $.
\end{proof}

To find the energy increment (\ref{Lemma non existing}) in the limit $(\eta
,l^{-1})\rightarrow (0,0)$, we use below Lebesgue's dominated convergence
theorem and it is crucial to remove the dependency of the measurable subset $%
\tilde{\Omega}\left( s_{1},s_{2}\right) $ on $s_{1},s_{2}\in \mathbb{R}$,
see Theorem \ref{limit ohm1limit ohm1}. To this end we first need to show
some uniform boundedness and continuity of the function (\ref{Lemma non
existing2}):

\begin{lemma}[Uniform Boundedness and Equicontinuity of $\mathbf{X}$%
--integrands]
\label{lemma conductivty4 copy(2)}\mbox{
}\newline
Let $\beta \in \mathbb{R}^{+}$, $\lambda \in \mathbb{R}_{0}^{+}$ and $%
\mathbf{A}\in \mathbf{C}_{0}^{\infty }$. The family%
\begin{equation*}
\left\{ (s_{1},s_{2})\mapsto \mathbf{X}_{l,0}^{(\omega
)}(s_{1},s_{2})\right\} _{l\in \mathbb{R}^{+},\omega \in \Omega }
\end{equation*}%
of maps from $\mathbb{R}^{2}$ to $\mathbb{C}$ is uniformly bounded and
equicontinuous.
\end{lemma}

\begin{proof}
The uniform boundedness of this collection of maps is an immediate
consequence of (\ref{Lemma non existing2bis}), see (\ref{Lemma non existing2}%
). To prove its equicontinuity, it suffices, by Lemmata \ref{lemma
conductivty0}--\ref{lemma conductivty1 copy(1)}, to verify that, for any
fixed $\beta \in \mathbb{R}^{+}$, $\lambda \in \mathbb{R}_{0}^{+}$, $\mathbf{%
A}\in \mathbf{C}_{0}^{\infty }$, $\varepsilon \in \mathbb{R}^{+}$ and $%
\upsilon \in (0,\beta /2)$, the family%
\begin{equation*}
\left\{ (s_{1},s_{2})\mapsto \mathbf{Y}_{l,\upsilon ,\varepsilon
,0}^{(\omega )}(s_{1},s_{2})\right\} _{l\in \mathbb{R}^{+},\omega \in \Omega
}
\end{equation*}%
of maps from $\mathbb{R}^{2}$ to $\mathbb{C}$ is equicontinuous, see (\ref{B
fractbis}). This property immediately follows from Theorem \ref{decay bound
theorem} (iii). We omit the details.
\end{proof}

Theorem \ref{limit ohm1limit ohm1} and Lemma \ref{lemma conductivty4 copy(2)}
imply two corollaries: The first one allows us to eliminate the $%
(s_{1},s_{2})$--dependency of the measurable set $\tilde{\Omega}\left(
s_{1},s_{2}\right) $ of Theorem \ref{limit ohm1limit ohm1}. The second one
concerns the continuity of the function $\Gamma _{k,q}$, which is in fact
related to a matrix--valued conductivity as explained after Theorem \ref%
{main 1 copy(2)}.

\begin{koro}[Infinite volume limit of the $\mathbf{X}$--integrand-II]
\label{lemma conductivty4 copy(6)}\mbox{
}\newline
Let $\beta \in \mathbb{R}^{+}$ and $\lambda \in \mathbb{R}_{0}^{+}$. Then,
there is a measurable subset $\tilde{\Omega}\equiv \tilde{\Omega}^{(\beta
,\lambda )}\subset \Omega $ of full measure such that, for any $%
s_{1},s_{2}\in \mathbb{R}$, $\mathbf{A}\in \mathbf{C}_{0}^{\infty }$ and $%
\omega \in \tilde{\Omega}$,%
\begin{equation}
\underset{l\rightarrow \infty }{\lim }\mathbf{X}_{l,0}^{(\omega )}\left(
s_{1},s_{2}\right) =\mathbf{X}_{\infty }(s_{1},s_{2})\ .
\label{limit ohm1bis}
\end{equation}
\end{koro}

\begin{proof}
Fix $\beta \in \mathbb{R}^{+}$ and $\lambda \in \mathbb{R}_{0}^{+}$. By
Theorem \ref{limit ohm1limit ohm1}, for any $s_{1},s_{2}\in \mathbb{Q}$,
there is a measurable subset $\hat{\Omega}\left( s_{1},s_{2}\right) \subset
\Omega $ of full measure such that (\ref{limit ohm1bis}) holds for any $%
\mathbf{A}\in \mathbf{C}_{0}^{\infty }$ and $\omega \in \hat{\Omega}\left(
s_{1},s_{2}\right) $. Let $\tilde{\Omega}$ be the intersection of all such
subsets $\hat{\Omega}\left( s_{1},s_{2}\right) $. Since this intersection is
countable, $\tilde{\Omega}$ is measurable and has full measure. By Lemma \ref%
{lemma conductivty4 copy(2)} and the density of $\mathbb{Q}$ in $\mathbb{R}$%
, it follows that (\ref{limit ohm1bis}) holds true for any $s_{1},s_{2}\in
\mathbb{R}$, $\mathbf{A}\in \mathbf{C}_{0}^{\infty }$ and $\omega \in \tilde{%
\Omega}$.
\end{proof}

\begin{koro}[Continuity of paramagnetic production coefficients]
\label{lemma conductivty4 copy(5)}\mbox{
}\newline
Let $\beta \in \mathbb{R}^{+}$, $\lambda \in \mathbb{R}_{0}^{+}$ and $k,q\in
\{1,-1,\ldots ,d,-d\}$. Then, the function $\Gamma _{k,q}$ from $\mathbb{R}$
to $\mathbb{C}$ defined by (\ref{conductivity0}) is continuous.
\end{koro}

\begin{proof}
For each $k,q\in \{1,-1,\ldots ,d,-d\}$ and $t\in \mathbb{R}$, choose $%
\mathbf{A}\in \mathbf{C}_{0}^{\infty }$ such that, in a fixed neighborhood
of $t$, the map $s\mapsto E_{\mathbf{A}}(s,x)$ is constant for any $x\in
\mathbb{R}^{d}$ and
\begin{equation*}
\int\nolimits_{\mathbb{R}^{d}}\mathrm{d}^{d}x\left[ E_{\mathbf{A}}(t,x)%
\right] (e_{q})\left[ E_{\mathbf{A}}(0,x)\right] (e_{k})\neq 0\ .
\end{equation*}%
Then, we combine the equicontinuity of the family
\begin{equation*}
\left\{ s\mapsto \mathbf{X}_{l,0}^{(\omega )}(s,0)\right\} _{l\in \mathbb{R}%
^{+},\omega \in \Omega }
\end{equation*}%
of maps from $\mathbb{R}$ to $\mathbb{C}$ given by Lemma \ref{lemma
conductivty4 copy(2)} with Corollary \ref{lemma conductivty4 copy(6)} to
show that the function $\Gamma _{k,q}$ is continuous at $t\in \mathbb{R}$
for each $k,q\in \{1,-1,\ldots ,d,-d\}$.
\end{proof}

Therefore, we can now use Lebesgue's dominated convergence theorem to get
the energy increment (\ref{Lemma non existing}) in the limit $(\eta
,l^{-1})\rightarrow (0,0)$:

\begin{satz}[Matrix--valued heat production coefficient]
\label{main 1 copy(2)}\mbox{
}\newline
Let $\beta \in \mathbb{R}^{+}$ and $\lambda \in \mathbb{R}_{0}^{+}$. Then,
there is a measurable subset $\tilde{\Omega}\equiv \tilde{\Omega}^{(\beta
,\lambda )}\subset \Omega $ of full measure such that, for any $\omega \in
\tilde{\Omega}$, $\mathbf{A}\in \mathbf{C}_{0}^{\infty }$ and $t\geq t_{0}$,
\begin{equation}
\mathfrak{i}_{\mathrm{p}}\left( t\right) :=\underset{(\eta
,l^{-1})\rightarrow (0,0)}{\lim }\left\{ \left( \eta ^{2}\left\vert \Lambda
_{l}\right\vert \right) ^{-1}\mathfrak{I}_{\mathrm{p}}^{(\omega ,\eta
\mathbf{A}_{l})}\left( t\right) \right\} =\int\nolimits_{t_{0}}^{t}\mathrm{d}%
s_{1}\int\nolimits_{t_{0}}^{s_{1}}\mathrm{d}s_{2}\ \mathbf{X}_{\infty
}(s_{1},s_{2})\ .  \notag
\end{equation}
\end{satz}

\begin{proof}
Recall (\ref{Lemma non existing}), that is, for any $t\geq t_{0}$,%
\begin{equation*}
\left( \eta ^{2}\left\vert \Lambda _{l}\right\vert \right) ^{-1}\mathfrak{I}%
_{\mathrm{p}}^{(\omega ,\eta \mathbf{A}_{l})}\left( t\right)
=\int\nolimits_{t_{0}}^{t}\mathrm{d}s_{1}\int\nolimits_{t_{0}}^{s_{1}}%
\mathrm{d}s_{2}\ \mathbf{X}_{l,0}^{(\omega )}(s_{1},s_{2})+\mathcal{O}(\eta
)\ .
\end{equation*}%
The assertion then follows from Lemma \ref{lemma conductivty4 copy(2)} and
Corollary \ref{lemma conductivty4 copy(6)} together with Lebesgue's
dominated convergence theorem.
\end{proof}

Notice at this point that the theorem above together with Equation (\ref{X
infinity}) means that the continuous functions $\Gamma _{k,q}$ define the
entries of a matrix--valued heat production coefficient. In fact, (\ref{X
infinity}) can be rewritten by using the deterministic paramagnetic
transport coefficient $\mathbf{\Xi }_{\mathrm{p}}$ defined by (\ref%
{paramagnetic transport coefficient macro}):

\begin{lemma}[$\mathbf{\Xi }_{\mathrm{p}}$ as heat production coefficient]
\label{lemma conductivty4 copy(1)}\mbox{
}\newline
For any $\beta \in \mathbb{R}^{+}$, $\lambda \in \mathbb{R}_{0}^{+}$ and $%
s_{1},s_{2}\in \mathbb{R}$,
\begin{equation}
\mathbf{X}_{\infty }(s_{1},s_{2})=\int\nolimits_{\mathbb{R}^{d}}\mathrm{d}%
^{d}x\ \left\langle E_{\mathbf{A}}(s_{1},x),\mathbf{\Xi }_{\mathrm{p}%
}(s_{1}-s_{2})E_{\mathbf{A}}(s_{2},x)\right\rangle \ .  \notag
\end{equation}
\end{lemma}

\begin{proof}
By combining \cite[Theorem 3.1, Lemma 5.2]{OhmII} with Equations (\ref%
{current observable}), (\ref{backwards -1bis}), (\ref{average conductivity}%
), (\ref{paramagnetic transport coefficient macro}), (\ref{conductivity0}), (%
\ref{X infinity}),
one gets
\begin{equation}
\left\{ \mathbf{\Xi }_{\mathrm{p}}\left( t\right) \right\} _{q,k}=\left\{
\mathbf{\Xi }_{\mathrm{p}}\left( t\right) \right\} _{k,q}=\Gamma
_{k,q}\left( t\right) -\Gamma _{k,q}\left( 0\right) \in \mathbb{R}
\label{well-defined sigma}
\end{equation}%
for any $k,q\in \{1,\ldots ,d\}$ and $t\in \mathbb{R}$. Hence,
\begin{equation}
\mathbf{X}_{\infty }(s_{1},s_{2})=\int\nolimits_{\mathbb{R}^{d}}\mathrm{d}%
^{d}x\sum_{q=1}^{d}\left[ E_{\mathbf{A}}(s_{1},x)\right] (e_{q})%
\sum_{k=1}^{d}\left\{ \mathbf{\Xi }_{\mathrm{p}}\left( s_{1}-s_{2}\right)
\right\} _{q,k}\left[ E_{\mathbf{A}}(s_{2},x)\right] (e_{k})\ .  \notag
\end{equation}
\end{proof}

Therefore, Theorem \ref{main 1} (p) directly results from Theorem \ref{main
1 copy(2)} and Lemma \ref{lemma conductivty4 copy(1)} together with Fubini's
theorem. In particular, $\mathbf{\Xi }_{\mathrm{p}}$ can also be seen as the
heat production coefficient. Under the assumption that the random variables
are independently and identically distributed (i.i.d.) (Section \ref{Section
impurities}), this coefficient becomes a scalar:

\begin{lemma}[Paramagnetic transport coefficient as a scalar]
\label{lemma conductivty4 copy(3)}\mbox{
}\newline
For any $\beta \in \mathbb{R}^{+}$, $\lambda \in \mathbb{R}_{0}^{+}$ , there
is a real function $\mathbf{\sigma }_{\mathrm{p}}\equiv \mathbf{\sigma }_{%
\mathrm{p}}^{(\beta ,\lambda )}$ such that
\begin{equation*}
\mathbf{\Xi }_{\mathrm{p}}\left( t\right) =\mathbf{\sigma }_{\mathrm{p}%
}\left( t\right) \ \mathrm{Id}_{\mathbb{R}^{d}}\ ,\qquad t\in \mathbb{R}\ .
\end{equation*}
\end{lemma}

\begin{proof}
Straightforward computations using the invariance of $\mathfrak{a}_{\Omega }$
under translations, reflections and permutations of axes (cf. (\ref%
{probability measure})) show from (\ref{average conductivity}) and (\ref%
{paramagnetic transport coefficient macro}) that at $t\in \mathbb{R}$, the
coefficient $\left\{ \mathbf{\Xi }_{\mathrm{p}}\left( t\right) \right\}
_{k,q}$ vanishes for all $k,q\in \{1,\ldots ,d\}$ with $k\neq q $, while,
for any $k,q\in \{1,\ldots ,d\}$ and $t\in \mathbb{R}$,
\begin{equation*}
\left\{ \mathbf{\Xi }_{\mathrm{p}}\left( t\right) \right\} _{k,k}=\left\{
\mathbf{\Xi }_{\mathrm{p}}\left( t\right) \right\} _{q,q}=:\mathbf{\sigma }_{%
\mathrm{p}}\left( t\right) \ .
\end{equation*}
\end{proof}

It remains to prove Assertions (\textbf{Q}) and (\textbf{P}) of Theorem \ref%
{main 1}. By Equation (\ref{total work}), it suffices to study the\ \emph{%
potential energy density difference} $\Delta _{\mathbf{P}}\equiv \Delta _{%
\mathbf{P}}^{(\beta ,\omega ,\lambda ,\mathbf{A})}$ defined by%
\begin{equation*}
\Delta _{\mathbf{P}}\left( t\right) :=\underset{(\eta ,l^{-1})\rightarrow
(0,0)}{\lim }\left\{ \left( \eta ^{2}\left\vert \Lambda _{l}\right\vert
\right) ^{-1}\left( \mathbf{P}^{(\omega ,\eta \mathbf{A}_{l})}\left(
t\right) -\mathfrak{I}_{\mathrm{d}}^{(\omega ,\eta \mathbf{A}_{l})}\left(
t\right) \right) \right\}
\end{equation*}%
for any $\beta \in \mathbb{R}^{+}$, $\omega \in \Omega $, $\lambda \in
\mathbb{R}_{0}^{+}$, $\mathbf{A}\in \mathbf{C}_{0}^{\infty }$ and $t\geq
t_{0}$. This analysis is done in the following theorem:

\begin{satz}[Potential energy density difference]
\label{Thm pot en density}\mbox{
}\newline
Let $\beta \in \mathbb{R}^{+}$ and $\lambda \in \mathbb{R}_{0}^{+}$. Then,
there is a measurable subset $\tilde{\Omega}\equiv \tilde{\Omega}^{(\beta
,\lambda )}\subset \Omega $ of full measure such that, for any $\omega \in
\tilde{\Omega}$, $\mathbf{A}\in \mathbf{C}_{0}^{\infty }$ and $t\geq t_{0}$,
\begin{equation*}
\Delta _{\mathbf{P}}\left( t\right) =\int\nolimits_{\mathbb{R}^{d}}\mathrm{d}%
^{d}x\int\nolimits_{t_{0}}^{t}\mathrm{d}s_{1}\int_{t_{0}}^{t}\mathrm{d}%
s_{2}\left\langle E_{\mathbf{A}}(s_{1},x),\mathbf{\Xi }_{\mathrm{p}}\left(
t-s_{2}\right) E_{\mathbf{A}}(s_{2},x)\right\rangle \ .
\end{equation*}
\end{satz}

\begin{proof}
By \cite[Lemmata 5.2 and 5.13]{OhmII}, for any $\mathbf{A}\in \mathbf{C}%
_{0}^{\infty }$, there is $\eta _{0}\in \mathbb{R}^{+}$ such that, for all $%
|\eta |\in (0,\eta _{0}]$ and $l\in \mathbb{R}^{+}$,
\begin{equation*}
\mathbf{P}^{(\omega ,\mathbf{A})}\left( t\right) -\mathfrak{I}_{\mathrm{d}%
}^{(\omega ,\mathbf{A})}\left( t\right) =\eta ^{2}\left\vert \Lambda
_{l}\right\vert \int\nolimits_{t_{0}}^{t}\mathrm{d}s_{1}\int_{t_{0}}^{t}%
\mathrm{d}s_{2}\ \mathbf{\hat{X}}_{l}^{(\omega )}(s_{1},s_{2})+\mathcal{O}%
(\eta ^{3}l^{d})
\end{equation*}%
uniformly for $\beta \in \mathbb{R}^{+}$, $\omega \in \Omega $, $\lambda \in
\mathbb{R}_{0}^{+}$ and $t\geq t_{0}$, where%
\begin{equation*}
\mathbf{\hat{X}}_{l}^{(\omega )}\left( s_{1},s_{2}\right) :=\frac{1}{%
4\left\vert \Lambda _{l}\right\vert }\underset{\mathbf{x},\mathbf{y}\in
\mathfrak{K}}{\sum }\int\nolimits_{0}^{\beta }\mathrm{d}\alpha \left(
\mathfrak{C}_{t-s_{2}+i\alpha }^{(\omega )}(\mathbf{x},\mathbf{y})-\mathfrak{%
C}_{i\alpha }^{(\omega )}(\mathbf{x},\mathbf{y})\right) \mathbf{E}_{s_{1}}^{%
\mathbf{A}_{l}}\left( \mathbf{x}\right) \mathbf{E}_{s_{2}}^{\mathbf{A}_{l}}(%
\mathbf{y})
\end{equation*}%
for any $s_{1},s_{2}\in \mathbb{R}$. The function $\mathbf{\hat{X}}%
_{l}^{(\omega )}$ is very similar to $\mathbf{X}_{l,0}^{(\omega )}$. Compare
indeed the last equation with (\ref{Lemma non existing2}). As a consequence,
one gets the assertion exactly in the same way one proves Theorem \ref{main
1 copy(2)} and Lemma \ref{lemma conductivty4 copy(1)}. We omit the details.
\end{proof}

\subsection{Hilbert Space of Current Fluctuations\label{section Current
Fluctuations copy(2)}}

\subsubsection{Positive Sesquilinear Forms\label{section fluct1}}

As explained in Section \ref{section Current Fluctuations} the linear
subspace
\begin{equation}
\mathcal{I}:=\mathrm{lin}\left\{ \mathrm{Im}(a^{\ast }\left( \psi
_{1}\right) a\left( \psi _{2}\right) ):\psi _{1},\psi _{2}\in \ell ^{1}(%
\mathfrak{L})\subset \ell ^{2}(\mathfrak{L})\right\} \subset \mathcal{U}
\label{space of currents}
\end{equation}%
is an invariant space of the one--parameter (Bogoliubov) group $\tau
^{(\omega ,\lambda )}$ for any $\omega \in \Omega $ and $\lambda \in \mathbb{%
R}_{0}^{+}$.

Let the random positive sesquilinear form $\langle \cdot ,\cdot \rangle _{%
\mathcal{I},l}^{(\omega )}\equiv \langle \cdot ,\cdot \rangle _{\mathcal{I}%
,l}^{(\beta ,\omega ,\lambda )}$ in $\mathcal{I}$ be defined by
\begin{equation}
\langle I,I^{\prime }\rangle _{\mathcal{I},l}^{(\omega )}:=\varrho ^{(\beta
,\omega ,\lambda )}\left( \mathbb{F}^{(l)}\left( I\right) ^{\ast }\mathbb{F}%
^{(l)}\left( I^{\prime }\right) \right) \ ,\qquad I,I^{\prime }\in \mathcal{I%
}\ ,  \label{Fluctuation2bisbis}
\end{equation}%
for any $l,\beta \in \mathbb{R}^{+}$, $\omega \in \Omega $ and $\lambda \in
\mathbb{R}_{0}^{+}$. Here, $\mathbb{F}^{(l)}$ is the fluctuation observable
defined by (\ref{Fluctuation2}), that is,
\begin{equation}
\mathbb{F}^{(l)}\left( I\right) :=\frac{1}{\left\vert \Lambda
_{l}\right\vert ^{1/2}}\underset{x\in \Lambda _{l}}{\sum }\left\{ \chi
_{x}\left( I\right) -\varrho ^{(\beta ,\omega ,\lambda )}\left( \chi
_{x}\left( I\right) \right) \mathbf{1}\right\} \ ,\qquad I\in \mathcal{I}\ ,
\label{Fluctuation2bis}
\end{equation}%
for each $l\in \mathbb{R}^{+}$. Recall that $\chi _{x}$, $x\in \mathfrak{L}$%
, are the (space) translation automorphisms.

In the following we aim to prove Theorem \ref{scalar product}. The latter
says that for $\omega $ in a subset $\tilde{\Omega}\equiv \tilde{\Omega}%
^{(\beta ,\lambda )}\subset \Omega $ of full measure, the limit $%
l\rightarrow \infty $ of $\langle \cdot ,\cdot \rangle _{\mathcal{I}%
,l}^{(\omega )}$ is a positive sesquilinear form $\langle \cdot ,\cdot
\rangle _{\mathcal{I}}\equiv \langle \cdot ,\cdot \rangle _{\mathcal{I}%
}^{(\beta ,\lambda )}$ on $\mathcal{I}$ which does not depend on $\omega \in
\tilde{\Omega}$. To prove this, it suffices to consider elements $%
I,I^{\prime }\in \mathcal{I}$ of the form
\begin{equation*}
I=\mathrm{Im}(a^{\ast }\left( \psi _{1}\right) a\left( \psi _{2}\right) )\
,\qquad I^{\prime }=\mathrm{Im}(a^{\ast }\left( \psi _{1}^{\prime }\right)
a\left( \psi _{2}^{\prime }\right) )\ ,
\end{equation*}%
with $\psi _{1},\psi _{2},\psi _{1}^{\prime },\psi _{2}^{\prime }\in \ell
^{1}(\mathfrak{L})$. In this case, we have a uniform estimate given by \cite[%
Lemma 5.10]{OhmII}: There is a constant $D\in \mathbb{R}^{+}$ such that, for
any $l,\beta \in \mathbb{R}^{+}$, $\omega \in \Omega $, $\lambda \in \mathbb{%
R}_{0}^{+}$ and all $\psi _{1},\psi _{2},\psi _{1}^{\prime },\psi
_{2}^{\prime }\in \ell ^{1}(\mathfrak{L})$,
\begin{equation}
\left\vert \langle \mathrm{Im}(a^{\ast }\left( \psi _{1}\right) a\left( \psi
_{2}\right) ),\mathrm{Im}(a^{\ast }\left( \psi _{1}^{\prime }\right) a\left(
\psi _{2}^{\prime }\right) )\rangle _{\mathcal{I},l}^{(\omega )}\right\vert
\leq D\left\Vert \psi _{1}\right\Vert _{1}\left\Vert \psi _{2}\right\Vert
_{1}\left\Vert \psi _{1}^{\prime }\right\Vert _{1}\left\Vert \psi
_{2}^{\prime }\right\Vert _{1}\ .  \label{bound incr 1 Lemma copy(3)-NEW}
\end{equation}%
Using this we can restrict the choice of $\psi _{1},\psi _{2},\psi
_{1}^{\prime },\psi _{2}^{\prime }$ to some convenient dense subset of $\ell
^{1}(\mathfrak{L})$: Let
\begin{equation}
\mathbb{\ell }_{0}^{\mathbb{Q}}:=\left\{ \psi \in \ell ^{1}(\mathfrak{L}%
):\psi \text{ is a }(\mathbb{Q}+i\mathbb{Q})\text{--valued function with
finite support}\right\}  \label{countable set}
\end{equation}%
and observe that it is a \emph{countable} and \emph{dense} subset of $\ell
^{1}(\mathfrak{L})$.

By countability of $\mathbb{\ell }_{0}^{\mathbb{Q}}$, it suffices to prove,
for each $\psi _{1},\psi _{2},\psi _{1}^{\prime },\psi _{2}^{\prime }\in
\mathbb{\ell }_{0}^{\mathbb{Q}}$, the existence of a subset $\tilde{\Omega}%
_{\psi _{1},\psi _{2},\psi _{1}^{\prime },\psi _{2}^{\prime }}\subset \Omega
$ of full measure such that the limit%
\begin{equation*}
\lim_{l\rightarrow \infty }\langle \mathrm{Im}(a^{\ast }\left( \psi
_{1}\right) a\left( \psi _{2}\right) ),\mathrm{Im}(a^{\ast }\left( \psi
_{1}^{\prime }\right) a\left( \psi _{2}^{\prime }\right) )\rangle _{\mathcal{%
I},l}^{(\omega )}
\end{equation*}%
exists and does not depend on $\omega \in \tilde{\Omega}_{\psi _{1},\psi
_{2},\psi _{1}^{\prime },\psi _{2}^{\prime }}$ in order to obtain a subset
\begin{equation*}
\tilde{\Omega}:=\underset{\psi _{1},\psi _{2},\psi _{1}^{\prime },\psi
_{2}^{\prime }\in \mathbb{\ell }_{0}^{\mathbb{Q}}}{\bigcap }\tilde{\Omega}%
_{\psi _{1},\psi _{2},\psi _{1}^{\prime },\psi _{2}^{\prime }}\subset \Omega
\end{equation*}%
of full measure with the required properties for all $\psi _{1},\psi
_{2},\psi _{1}^{\prime },\psi _{2}^{\prime }\in \mathbb{\ell }_{0}^{\mathbb{Q%
}}$. This is performed in the following lemma:

\begin{lemma}[Well--definiteness of $\langle \cdot ,\cdot \rangle _{\mathcal{%
I}}$]
\label{bound incr 1 Lemma copy(2)}\mbox{
}\newline
For any $\beta \in \mathbb{R}^{+}$, $\lambda \in \mathbb{R}_{0}^{+}$ and $%
\psi _{1},\psi _{2},\psi _{1}^{\prime },\psi _{2}^{\prime }\in \mathbb{\ell }%
_{0}^{\mathbb{Q}}$, there is a measurable subset $\tilde{\Omega}_{\psi
_{1},\psi _{2},\psi _{1}^{\prime },\psi _{2}^{\prime }}\subset \Omega $ of
full measure such that, for any $\tilde{\omega}\in \tilde{\Omega}_{\psi
_{1},\psi _{2},\psi _{1}^{\prime },\psi _{2}^{\prime }}$,
\begin{eqnarray*}
&&\underset{l\rightarrow \infty }{\lim }\langle \mathrm{Im}(a^{\ast }\left(
\psi _{1}\right) a\left( \psi _{2}\right) ),\mathrm{Im}(a^{\ast }\left( \psi
_{1}^{\prime }\right) a\left( \psi _{2}^{\prime }\right) )\rangle _{\mathcal{%
I},l}^{(\tilde{\omega})} \\
&=&\underset{l\rightarrow \infty }{\lim }\ \mathbb{E}\left[ \langle \mathrm{%
Im}(a^{\ast }\left( \psi _{1}\right) a\left( \psi _{2}\right) ),\mathrm{Im}%
(a^{\ast }\left( \psi _{1}^{\prime }\right) a\left( \psi _{2}^{\prime
}\right) )\rangle _{\mathcal{I},l}^{(\omega )}\right] \in \mathbb{R}\ .
\end{eqnarray*}
\end{lemma}

\begin{proof}
Let $\beta \in \mathbb{R}^{+}$, $\lambda \in \mathbb{R}_{0}^{+}$ and $\psi
_{1},\psi _{2},\psi _{1}^{\prime },\psi _{2}^{\prime }\in \mathbb{\ell }%
_{0}^{\mathbb{Q}}$. One uses the first equation of the proof of \cite[Lemma
5.10]{OhmII} as well as \cite[Eq. (134)]{OhmII}, that is all together,%
\begin{eqnarray}
&&\langle \mathrm{Im}(a^{\ast }\left( \psi _{1}\right) a\left( \psi
_{2}\right) ),\mathrm{Im}(a^{\ast }\left( \psi _{1}^{\prime }\right) a\left(
\psi _{2}^{\prime }\right) )\rangle _{\mathcal{I},l}^{(\omega )}  \notag \\
&=&\sum_{\mathbf{x}:=(x^{(1)},x^{(2)}),\mathbf{y}:=(y^{(1)},y^{(2)})\in
\mathfrak{L}^{2}}\psi _{1}(y^{(1)})\psi _{2}(y^{(2)})\psi _{1}^{\prime
}(x^{(1)})\psi _{2}^{\prime }(x^{(2)})  \notag \\
&&\times \left[ \frac{1}{4\left\vert \Lambda _{l}\right\vert }%
\sum_{z_{1},z_{2}\in \Lambda _{l}}\mathfrak{C}_{0}^{(\omega )}\left( \mathbf{%
x}+\left( z_{1},z_{1}\right) ,\mathbf{y}+\left( z_{2},z_{2}\right) \right) %
\right] \ .  \label{eq cool}
\end{eqnarray}%
See (\ref{map cool}) for the definition of the map $\mathfrak{C}_{t+i\alpha
}^{(\omega )}$. Then, one approximates $\mathfrak{C}_{0}^{(\omega )}$ by $%
\mathfrak{C}_{i\alpha }^{(\omega )}$ with $\alpha \ll \beta $ in (\ref{eq
cool}) by using the bounds%
\begin{multline*}
\left\vert \frac{1}{\left\vert \Lambda _{l}\right\vert }\sum_{z_{1},z_{2}\in
\Lambda _{l}}\mathfrak{C}_{0}^{(\omega )}\left( \mathbf{x}+\left(
z_{1},z_{1}\right) ,\mathbf{y}+\left( z_{2},z_{2}\right) \right) -\mathfrak{C%
}_{i\alpha }^{(\omega )}\left( \mathbf{x}+\left( z_{1},z_{1}\right) ,\mathbf{%
y}+\left( z_{2},z_{2}\right) \right) \right\vert \\
\leq \Vert \mathfrak{C}_{i\alpha }^{(\omega )}-\mathfrak{C}_{0}^{(\omega
)}\Vert _{\mathrm{op}}\leq D\alpha
\end{multline*}%
for sufficiently small $\alpha \in \lbrack 0,\beta ]$. Here, $D\in \mathbb{R}%
^{+}$ is a finite constant only depending on $\lambda \in \mathbb{R}_{0}^{+}$%
. For more details, see \cite[Lemma 5.3,\ Eq. (102)]{OhmII}. This allows us
to use Theorems \ref{decay bound theorem} and \ref{Ackoglu--Krengel ergodic
theorem II copy(1)} in order to prove the assertion. We omit the details as
it is a simpler version of results proven in this paper. See for instance
Theorem \ref{main 1 copy(2)}.
\end{proof}

Therefore, we define the \emph{deterministic} positive sesquilinear form $%
\langle \cdot ,\cdot \rangle _{\mathcal{I}}\equiv \langle \cdot ,\cdot
\rangle _{\mathcal{I}}^{(\beta ,\lambda )}$ in $\mathcal{I}$ to be
\begin{equation*}
\langle I,I^{\prime }\rangle _{\mathcal{I}}:=\underset{l\rightarrow \infty }{%
\lim }\ \mathbb{E}\left[ \langle I,I^{\prime }\rangle _{\mathcal{I}%
,l}^{(\omega )}\right] \in \mathbb{R}\ ,\qquad I,I^{\prime }\in \mathcal{I}\
.
\end{equation*}%
By combining (\ref{bound incr 1 Lemma copy(3)-NEW}) and Lemma \ref{bound
incr 1 Lemma copy(2)} we deduce the existence of this limit for all $%
I,I^{\prime }\in \mathcal{I}$ as well as Theorem \ref{scalar product}:

\begin{satz}[Sesquilinear form from current fluctuations]
\label{toto fluctbis}\mbox{
}\newline
Let $\beta \in \mathbb{R}^{+}$ and $\lambda \in \mathbb{R}_{0}^{+}$. Then,
one has:\newline
\emph{(i)} The positive sesquilinear form $\langle \cdot ,\cdot \rangle _{%
\mathcal{I}}$ is well--defined, i.e.,
\begin{equation*}
\langle I,I^{\prime }\rangle _{\mathcal{I}}:=\underset{l\rightarrow \infty }{%
\lim }\ \mathbb{E}\left[ \langle I,I^{\prime }\rangle _{\mathcal{I}%
,l}^{(\omega )}\right] \in \mathbb{R}\ ,\qquad I,I^{\prime }\in \mathcal{I}\
.
\end{equation*}%
\emph{(ii)} There is a measurable subset $\tilde{\Omega}\equiv \tilde{\Omega}%
^{(\beta ,\lambda )}\subset \Omega $ of full measure such that, for any $%
\omega \in \tilde{\Omega}$,
\begin{equation*}
\langle I,I^{\prime }\rangle _{\mathcal{I}}=\underset{l\rightarrow \infty }{%
\lim }\langle I,I^{\prime }\rangle _{\mathcal{I},l}^{(\beta ,\omega ,\lambda
)}\ ,\qquad I,I^{\prime }\in \mathcal{I}\ .
\end{equation*}
\end{satz}

We are now in position to introduce next the Hilbert space of current
fluctuations.

\subsubsection{Hilbert Space and Dynamics\label{Sect Hilbert Space and
Dynamics}}

As explained in Section \ref{section Current Fluctuations}, the quotient $%
\mathcal{\check{H}}_{\mathrm{fl}}:=\mathcal{I}/\mathcal{I}_{0}$ is a
pre--Hilbert space w.r.t. to the (well--defined) scalar product
\begin{equation}
\langle \lbrack I],[I^{\prime }]\rangle _{\mathcal{\check{H}}_{\mathrm{fl}%
}}:=\ \langle I,I^{\prime }\rangle _{\mathcal{I}}\ ,\qquad \lbrack
I],[I^{\prime }]\in \mathcal{\check{H}}_{\mathrm{fl}}\ ,
\label{Fluctuation3}
\end{equation}%
where
\begin{equation*}
\mathcal{I}_{0}:=\left\{ I\in \mathcal{I}:\langle I,I\rangle _{\mathcal{I}%
}=0\right\} \ .
\end{equation*}%
The completion of $\mathcal{\check{H}}_{\mathrm{fl}}$ w.r.t. the scalar
product $\langle \cdot ,\cdot \rangle _{\mathcal{\check{H}}_{\mathrm{fl}}}$
is the Hilbert space of current fluctuations denoted by%
\begin{equation}
\left( \mathcal{H}_{\mathrm{fl}},\langle \cdot ,\cdot \rangle _{\mathcal{H}_{%
\mathrm{fl}}}\right) \, .  \label{Fluctuation4}
\end{equation}%
The random dynamics defined by $\tau ^{(\omega ,\lambda )}$ on $\mathcal{U}$
induces a unitary time evolution on $\mathcal{H}_{\mathrm{fl}}$:

\begin{satz}[Dynamics of current fluctuations]
\label{bound incr 1 Lemma copy(1)}\mbox{
}\newline
Let $\beta \in \mathbb{R}^{+}$ and $\lambda \in \mathbb{R}_{0}^{+}$. Then,
there is a measurable subset $\tilde{\Omega}\equiv \tilde{\Omega}^{(\beta
,\lambda )}\subset \Omega $ of full measure such that, for any $\omega \in
\tilde{\Omega}$, there is a unique, strongly continuous one--parameter
unitary group $\{\mathrm{V}_{t}^{(\omega ,\lambda )}\}_{t\in {\mathbb{R}}}$
on the Hilbert space $\mathcal{H}_{\mathrm{fl}}$ obeying, for any $t\in {%
\mathbb{R}}$,
\begin{equation*}
\mathrm{V}_{t}^{(\omega ,\lambda )}([I])=[\tau _{t}^{(\omega ,\lambda
)}(I)]\ ,\qquad \lbrack I]\in \mathcal{\check{H}}_{\mathrm{fl}}\ .
\end{equation*}
\end{satz}

\begin{proof}
By (\ref{bound incr 1 Lemma copy(3)-NEW}), Theorem \ref{toto fluctbis} and
the stationarity of the KMS state $\varrho ^{(\beta ,\omega ,\lambda )}$\
(cf. (\ref{stationary})), for any $\beta \in \mathbb{R}^{+}$, $\omega \in
\Omega $ and $\lambda \in \mathbb{R}_{0}^{+}$, the one--parameter group $%
\tau ^{(\omega ,\lambda )}$ defines a unitary group $\{\mathrm{V}%
_{t}^{(\omega ,\lambda )}\}_{t\in {\mathbb{R}}}$ on the Hilbert space $(%
\mathcal{H}_{\mathrm{fl}},\langle \cdot ,\cdot \rangle _{\mathcal{H}_{%
\mathrm{fl}}})$ as explained in the theorem: $\tau _{t}^{(\omega ,\lambda )}(%
\mathcal{I})\subset \mathcal{I}$ while the stationarity of $\varrho ^{(\beta
,\omega ,\lambda )}$ implies
\begin{equation*}
\left\Vert \lbrack \tau _{t}^{(\omega ,\lambda )}(I)]\right\Vert _{\mathcal{H%
}_{\mathrm{fl}}}=\left\Vert [I]\right\Vert _{\mathcal{H}_{\mathrm{fl}}}\
,\qquad \lbrack I]\in \mathcal{\check{H}}_{\mathrm{fl}}\ ,
\end{equation*}%
for all $t\in {\mathbb{R}}$. In particular, $\tau _{t}^{(\omega ,\lambda )}(%
\mathcal{I}_{0})\subset \mathcal{I}_{0}$ and hence, $[\tau _{t}^{(\omega
,\lambda )}(I)]\in \mathcal{\check{H}}_{\mathrm{fl}}$ depends only on the
equivalence class $[I]\in \mathcal{\check{H}}_{\mathrm{fl}}$ for all $I\in
\mathcal{I}$ and $t\in {\mathbb{R}}$. It remains to show that, for any $%
\beta \in \mathbb{R}^{+}$, $\omega \in \Omega $ and $\lambda \in \mathbb{R}%
_{0}^{+}$, $\mathrm{V}_{t}^{(\omega ,\lambda )}$ is strongly continuous at $%
t=0$ on a dense subset of $\mathcal{H}_{\mathrm{fl}}$.

To this end, observe that, for any $[I]$ in the dense subspace $\mathcal{%
\check{H}}_{\mathrm{fl}}$ of $\mathcal{H}_{\mathrm{fl}}$ and any \emph{fixed}
$\omega \in \Omega $,
\begin{eqnarray*}
\left\Vert \mathrm{V}_{t}^{(\omega ,\lambda )}\left( [I]\right)
-[I]\right\Vert _{\mathcal{H}_{\mathrm{fl}}}^{2} &=&\underset{l\rightarrow
\infty }{\lim }\varrho ^{(\beta ,\omega ,\lambda )}\left( \mathbb{F}%
^{(l)}\left( I-\tau _{t}^{(\omega ,\lambda )}\left( I\right) \right) ^{\ast }%
\mathbb{F}^{(l)}\left( I\right) \right) \\
&&+\underset{l\rightarrow \infty }{\lim }\varrho ^{(\beta ,\omega ,\lambda
)}\left( \mathbb{F}^{(l)}\left( I-\tau _{-t}^{(\omega ,\lambda )}\left(
I\right) \right) ^{\ast }\mathbb{F}^{(l)}\left( I\right) \right) \ .
\end{eqnarray*}%
We assume w.l.o.g. that $I=a^{\ast }\left( \psi _{1}\right) a\left( \psi
_{2}\right) $ for some $\psi _{1},\psi _{2}\in \ell ^{1}(\mathfrak{L})$.
Then, explicit computations starting from the last equality lead to
\begin{eqnarray}
\left\Vert \mathrm{V}_{t}^{(\omega ,\lambda )}\left( [I]\right)
-[I]\right\Vert _{\mathcal{H}_{\mathrm{fl}}}^{2} &=&\underset{l\rightarrow
\infty }{\lim }\langle I_{1}^{(\omega )}\left( t\right) ,I\rangle _{\mathcal{%
I},l}^{(\omega )}+\underset{l\rightarrow \infty }{\lim }\langle
I_{2}^{(\omega )}\left( t\right) ,I\rangle _{\mathcal{I},l}^{(\omega )}
\label{ineq cool} \\
&&+\underset{l\rightarrow \infty }{\lim }\langle I_{1}^{(\omega )}\left(
-t\right) ,I\rangle _{\mathcal{I},l}^{(\omega )}+\underset{l\rightarrow
\infty }{\lim }\langle I_{2}^{(\omega )}\left( -t\right) ,I\rangle _{%
\mathcal{I},l}^{(\omega )}\ ,  \notag
\end{eqnarray}%
where, for any $\psi _{1},\psi _{2}\in \ell ^{1}(\mathfrak{L})$,%
\begin{equation*}
I_{1}^{(\omega )}\left( t\right) :=a^{\ast }(\psi _{1}-\mathrm{U}%
_{t}^{(\omega ,\lambda )}\psi _{1})a(\psi _{2})\quad \text{and}\quad
I_{2}^{(\omega )}\left( t\right) :=a^{\ast }(\mathrm{U}_{t}^{(\omega
,\lambda )}\psi _{1})a(\psi _{2}-\mathrm{U}_{t}^{(\omega ,\lambda )}\psi
_{2})\ .
\end{equation*}%
Then, by using (\ref{bound incr 1 Lemma copy(3)-NEW}) together with
\begin{equation*}
\underset{t\rightarrow 0}{\lim }\Vert \psi _{1}-\mathrm{U}_{t}^{(\omega
,\lambda )}\psi _{1}\Vert _{1}=\underset{t\rightarrow 0}{\lim }\Vert \psi
_{2}-\mathrm{U}_{t}^{(\omega ,\lambda )}\psi _{2}\Vert _{1}=0\ ,\text{ \ }%
\underset{t\rightarrow 0}{\lim }\Vert \mathrm{U}_{t}^{(\omega ,\lambda
)}\psi _{1}\Vert _{1}=\Vert \psi _{1}\Vert _{1}\ ,
\end{equation*}%
we infer from (\ref{ineq cool}) that
\begin{equation*}
\underset{t\rightarrow 0}{\lim }\left\Vert \mathrm{V}_{t}^{(\omega ,\lambda
)}\left( [I]\right) -[I]\right\Vert _{\mathcal{H}_{\mathrm{fl}}}=0
\end{equation*}%
for any $\beta \in \mathbb{R}^{+}$, $\omega \in \Omega $ and $\lambda \in
\mathbb{R}_{0}^{+}$.
\end{proof}

Note that the strongly continuous one--parameter unitary group $\{\mathrm{V}%
_{t}^{(\omega ,\lambda )}\}_{t\in {\mathbb{R}}}$ on the Hilbert space $(%
\mathcal{H}_{\mathrm{fl}},\langle \cdot ,\cdot \rangle _{\mathcal{H}_{%
\mathrm{fl}}})$ is a priori depending on the parameter $\omega \in \tilde{%
\Omega}$, even if Equation (\ref{CCR4}) does not depend on $\omega \in
\tilde{\Omega}$. In fact, one can also construct a \emph{direct integral}
Hilbert space to get a deterministic, strongly continuous one--parameter
unitary group $\{\mathrm{\bar{V}}_{t}^{(\lambda )}\}_{t\in {\mathbb{R}}}$.
For the interested reader, we sketch this construction in the next
subsection:

\subsubsection{Averaged Initial State and Dynamics\label{Section average
dynamics}}

Note that the map $\omega \mapsto \varrho ^{(\beta ,\omega ,\lambda )}$ is
continuous w.r.t. the topology on $\Omega $ of which $\mathfrak{A}_{\Omega }$
is the Borel $\sigma $--algebra and the weak$^\ast$--topology for states. It
is a consequence of a result similar to \cite[Proposition 5.3.25.]%
{BratteliRobinson} together with the uniqueness of $(\tau ^{(\omega ,\lambda
)},\beta )$--KMS states. Then, define, for any $\beta \in \mathbb{R}^{+}$
and $\lambda \in \mathbb{R}_{0}^{+}$, the averaged state $\bar{\varrho}%
^{(\beta ,\lambda )}\in \mathcal{U}^{\ast }$ by%
\begin{equation*}
\bar{\varrho}^{(\beta ,\lambda )}\left( B\right) :=\mathbb{E}\left[ \varrho
^{(\beta ,\omega ,\lambda )}\left( B\right) \right] \ ,\qquad B\in \mathcal{U%
}\ .
\end{equation*}%
%
%
%
%
%
%
%
%
%
%
%
%

For any $\beta \in \mathbb{R}^{+}$, $\lambda \in \mathbb{R}_{0}^{+}$ and $%
\omega \in \Omega $, let $(\mathcal{H}^{(\omega )},\pi ^{(\omega )},\Psi
^{(\omega )})$ be the GNS representation of the $(\tau ^{(\omega ,\lambda
)},\beta )$--KMS state $\varrho ^{(\beta ,\omega ,\lambda )}$. The vector $%
\Psi ^{(\omega )}$ is cyclic and the CAR $C^{\ast }$--algebra $\mathcal{U}$
is separable. Therefore, there is a sequence $\{B_{n}\}_{n\in \mathbb{N}%
}\subset \mathcal{U}$ such that the subset
\begin{equation*}
\{\pi ^{(\omega )}\left( B_{n}\right) \Psi ^{(\omega )}\}_{n\in \mathbb{N}%
}\subset \mathcal{H}^{(\omega )}\ ,\qquad \omega \in \Omega \ ,
\end{equation*}%
is dense in $\mathcal{H}^{(\omega )}$. Moreover, the map
\begin{equation*}
\omega \mapsto \langle \pi ^{(\omega )}\left( B_{n}\right) \Psi ^{(\omega
)},\pi ^{(\omega )}\left( B_{m}\right) \Psi ^{(\omega )}\rangle _{\mathcal{H}%
^{(\omega )}}=\varrho ^{(\beta ,\omega ,\lambda )}\left( B_{n}^{\ast
}B_{m}\right)
\end{equation*}%
is bounded and measurable w.r.t. the $\sigma $--algebra $\mathfrak{A}%
_{\Omega }$ for all $n,m\in \mathbb{N}$. It follows that $\{\mathcal{H}%
^{(\omega )}\}_{\omega \in \Omega }$ is a \emph{measurable} family, see \cite%
[Definition 4.4.1B.]{BratteliRobinsonI}. In particular, as the probability
measure $\mathfrak{a}_{\mathbf{0}}$ is a \emph{standard} measure, there is a
\emph{direct integral} Hilbert space
\begin{equation*}
\mathcal{\bar{H}}:=\int_{\Omega }^{\oplus }\mathcal{H}^{(\omega )}\text{ }%
\mathrm{d}\mathfrak{a}_{\mathbf{0}}(\omega )
\end{equation*}%
with scalar product
\begin{equation*}
\langle b_{1},b_{2}\rangle _{\mathcal{\bar{H}}}:=\int_{\Omega }\langle
b_{1}^{(\omega )},b_{2}^{(\omega )}\rangle _{\mathcal{H}^{(\omega )}}\mathrm{%
d}\mathfrak{a}_{\mathbf{0}}(\omega )\ .
\end{equation*}

Note that $\mathcal{U}$ is the inductive limit of (finite dimensional)
simple $C^{\ast }$--algebras $\{\mathcal{U}_{\Lambda }\}_{\Lambda \in
\mathcal{P}_{f}(\mathfrak{L})}$, see \cite[Lemma IV.1.2]{simon}. By \cite[%
Corollary 2.6.19.]{BratteliRobinsonI}, $\mathcal{U}$ is thus simple and
hence, the $(\tau ^{(\omega ,\lambda )},\beta )$--KMS state $\varrho
^{(\beta ,\omega ,\lambda )}$ is faithful. In particular, $\pi ^{(\omega )}$
is injective for any $\omega \in \Omega $. We define a \emph{separating}
vector
\begin{equation*}
\bar{\Psi}:=\int_{\Omega }^{\oplus }\Psi ^{(\omega )}\mathrm{d}\mathfrak{a}_{%
\mathbf{0}}(\omega )\in \mathcal{\bar{H}}
\end{equation*}%
and a non--degenerate and injective representation
\begin{equation*}
\bar{\pi}:=\int_{\Omega }^{\oplus }\pi ^{(\omega )}\mathrm{d}\mathfrak{a}_{%
\mathbf{0}}(\omega )
\end{equation*}%
of the $C^{\ast }$--algebra $\mathcal{U}$ into the space $\mathcal{B}(%
\mathcal{\bar{H}})$. Then we have
\begin{equation*}
\bar{\varrho}^{(\beta ,\lambda )}\left( B\right) =\langle \bar{\Psi},\bar{\pi%
}(B)\bar{\Psi}\rangle _{\mathcal{\bar{H}}}\ ,\qquad B\in \mathcal{U}\ .
\end{equation*}%
In other words, $(\mathcal{\bar{H}},\bar{\pi})$ is a faithful representation
of the $C^\ast$--algebra $\mathcal{U}$ and $\bar{\Psi}$ is a separating
vector representing the state $\bar{\varrho}^{(\beta ,\lambda )}$.

Observe that the one--parameter group $\tau ^{(\omega ,\lambda )}$ has a
unique unitary representation $\{\mathrm{e}^{it\mathcal{L}^{(\omega
)}}\}_{t\in \mathbb{R}}\subset \pi ^{(\omega )}\left( \mathcal{U}\right)
^{\prime \prime }$ with $\mathcal{L}^{(\omega )}$ being a self--adjoint
operator acting on the Hilbert space $\mathcal{H}^{(\omega )}$ such that $%
\Psi ^{(\omega )}\in \mathrm{Dom}(\mathcal{L}^{(\omega )})$ and $\mathcal{L}%
^{(\omega )}\Psi ^{(\omega )}=0$. The family $\{\mathrm{e}^{it\mathcal{L}%
^{(\omega )}}\}_{t\in \mathbb{R},\omega \in \Omega }$ defines a strongly
continuous one--parameter unitary group $\{\bar{U}_{t}\}_{t\in \mathbb{R}}$
on $\mathcal{\bar{H}}$ by
\begin{equation*}
\bar{U}_{t}:=\int_{\Omega }^{\oplus }\mathrm{e}^{it\mathcal{L}^{(\omega )}}%
\mathrm{d}\mathfrak{a}_{\mathbf{0}}(\omega )\ .
\end{equation*}%
It defines an averaged unitary dynamics on $\mathcal{\bar{H}}$ which
satisfies $\bar{U}_{t}\bar{\Psi}=0$. In particular we can define a \emph{%
deterministic} one--parameter group $\bar{\tau}^{(\lambda )}\equiv \{\bar{%
\tau}_{t}^{(\lambda )}\}_{t\in {\mathbb{R}}}$ of automorphisms of $\mathcal{B%
}(\mathcal{\bar{H}})$ by
\begin{equation*}
\forall t\in \mathbb{R},\ B\in \mathcal{B}(\mathcal{\bar{H}}) :\qquad \bar{%
\tau}_{t}\left( B\right) :=\bar{U}_{t}B\bar{U}_{t}^{\ast }\in \mathcal{B}%
\left( \mathcal{\bar{H}}\right) \ .
\end{equation*}

Using these constructions, one can perform all the arguments of Sections \ref%
{section fluct1}--\ref{Sect Hilbert Space and Dynamics} by taking the
invariant space
\begin{equation*}
\bar{\mathcal{I}}:=\int_{\Omega }^{\oplus }\pi ^{(\omega )}(\mathcal{I})%
\mathrm{d}\mathfrak{a}_{\mathbf{0}}(\omega )\subset \mathcal{B}(\mathcal{%
\bar{H}})
\end{equation*}%
(cf. (\ref{space of currents})) of the group $\bar{\tau}^{(\lambda )}$. See,
e.g., Theorem \ref{toto fluctbis} (i). Then, for any $\beta \in \mathbb{R}%
^{+}$ and $\lambda \in \mathbb{R}_{0}^{+}$, one obtains the existence of a
unique, strongly continuous one--parameter \emph{deterministic} unitary
group $\{\mathrm{\bar{V}}_{t}^{(\lambda )}\}_{t\in {\mathbb{R}}}$ on the
Hilbert space constructed from the space of equivalence classes $\bar{%
\mathcal{I}}/\bar{\mathcal{I}}_{0}$ and denoted again by $(\mathcal{H}_{%
\mathrm{fl}},\langle \cdot ,\cdot \rangle _{\mathcal{H}_{\mathrm{fl}}})$.
The unitary group $\{\mathrm{\bar{V}}_{t}^{(\lambda )}\}_{t\in {\mathbb{R}}}$
obeys, for any $t\in {\mathbb{R}}$,
\begin{equation*}
\mathrm{\bar{V}}_{t}^{(\lambda )}([\bar{\pi}(I)])=[\bar{\tau}_{t}(\bar{\pi}%
(I))]\ ,\quad I\in \mathcal{I}\ .
\end{equation*}%
Moreover, by Theorems \ref{thm charged transport coefficient} (p) and \ref%
{toto fluctbis} (i),%
\begin{equation*}
\left\{ \mathbf{\Xi }_{\mathrm{p}}\left( t\right) \right\} _{k,q}=2\mathrm{Im%
}\left\{ \left\langle [\bar{\pi}(I_{e_{k},0})],\int\nolimits_{0}^{t}\mathrm{%
\bar{V}}_{s}^{(\lambda )}([\bar{\pi}(I_{e_{q},0})])\mathrm{d}s\right\rangle
_{\mathcal{H}_{\mathrm{fl}}}\right\}
\end{equation*}%
for any $\beta \in \mathbb{R}^{+}$, $\lambda \in \mathbb{R}_{0}^{+}$, $t\in
\mathbb{R}$ and $k,q\in \{1,\ldots ,d\}$.\bigskip

\noindent \textit{Acknowledgments:} We would like to thank Volker Bach,
Horia Cornean, Abel Klein and Peter M\"{u}ller for relevant references and
interesting discussions as well as important hints. JBB and WdSP are also
very grateful to the organizers of the Hausdorff Trimester Program entitled
\textquotedblleft \textit{Mathematical challenges of materials science and
condensed matter physics}\textquotedblright\ for the opportunity to work
together on this project at the Hausdorff Research Institute for Mathematics
in Bonn. This work has also been supported by the grant MTM2010-16843, MTM2014-53850 and the BCAM Severo Ochoa accreditation SEV-2013-0323
(MINECO) of the Spanish {}\textquotedblleft Ministerio de Ciencia e Innovaci{%
\'{o}}n\textquotedblright as well as the FAPESP grant 2013/13215--5 and the
Basque Government through the grant IT641-13 and the BERC 2014-2017 program.
Finally, we thank very much the referees for their work and interest in the
improvement of the paper.

\end{document}